\documentclass[12pt]{article}
\usepackage{amsfonts}
\usepackage{amssymb}
\usepackage{amsthm}
\usepackage{amsmath}
\usepackage{braket}
\usepackage{graphicx}
\usepackage{indentfirst}
\usepackage{cite}
\usepackage{mathrsfs}
\usepackage{graphics}
\usepackage{tikz}
\usetikzlibrary{decorations.pathreplacing,angles,quotes}
\usepackage{subcaption}
\usepackage{xcolor}
\usepackage{multicol}
\usepackage{enumerate}

\newtheorem{theorem}{\textbf{Theorem}}[section]
\newtheorem{lemma}{\textbf{Lemma}}[section]
\newtheorem{proposition}{\textbf{Proposition}}[section]
\newtheorem{corollary}{\textbf{Corollary}}[section]
\newtheorem{remark}{\textbf{Remark}}[section]
\newtheorem{definition}{\textbf{Definition}}[section]

\allowdisplaybreaks[4]

\def\be{\begin{equation}}
\def\ee{\end{equation}}
\def\bea{\begin{eqnarray}}
\def\eea{\end{eqnarray}}
\def\bt{\begin{theorem}}
\def\et{\end{theorem}}
\def\bl{\begin{lemma}}
\def\el{\end{lemma}}
\def\br{\begin{remark}}
\def\er{\end{remark}}
\def\bp{\begin{proposition}}
\def\ep{\end{proposition}}
\def\bc{\begin{corollary}}
\def\ec{\end{corollary}}
\def\bd{\begin{definition}}
\def\ed{\end{definition}}

\def\R3{\mathbb{R}^3} 
\def\R{\mathbb{R}}
\def\F2o{\overline{F_2}}

\def \eps{\varepsilon}
\newcommand{\Tn}[0]{Ti\textsubscript{74}Nb\textsubscript{23}Al\textsubscript{3}}

\newcommand{\beq}[0]{\begin{equation}}
\newcommand{\eeq}[0]{\end{equation}}

\newcommand{\mint}[0]{\int_\Omega}

\newcommand{\EE}[0]{{\mathcal E}}

\newcommand{\ttr}{}

\newcommand{\vc}[1]{\mathbf{#1}}
\newcommand{\vcg}[1]{\boldsymbol{#1}}
\newcommand{\mt}[1]{\mathsf{#1}}

\DeclareMathOperator{\cof}{\mathsf{cof}}
\DeclareMathOperator{\diag}{diag}
\DeclareMathOperator{\sign}{sign}
\DeclareMathOperator{\rank}{rank}

\usepackage{mathtools}

\newcommand{\mres}{\mathbin{\vrule height 1.6ex depth 0pt width
0.13ex\vrule height 0.13ex depth 0pt width 1.3ex}}

\usepackage[hang,flushmargin]{footmisc}
\makeatletter
\def\blfootnote{\gdef\@thefnmark{}\@footnotetext}
\makeatother

\begin{document}

\title{\ttr{On non stress-free junctions between martensitic plates}}

\author{
{\sc Francesco Della Porta}
}
\blfootnote{
\textbf{Acknowledgements:} {The author would like to thank the Max Planck Institute for Mathematics in the Sciences where part of this work was carried out. This work was partially supported by the Engineering and Physical Sciences Research Council [EP/L015811/1]. The author would like to thank John Ball, Tomonari Inamura and Angkana R\"uland for the useful discussions. The author would
like to acknowledge the two anonymous reviewers for improving this paper with their comments.}\\
\noindent\rule{6.8cm}{0.4pt}
}
\date{\today}
\maketitle

\begin{abstract}
The analytical understanding of microstructures arising in martensitic phase transitions relies usually on the study of stress-free interfaces between different variants of martensite. However, in the literature there are experimental observations of non stress-free junctions between martensitic plates, where the compatibility theory fails to be predictive. In this work, we focus on $V_{II}$ junctions, which are non stress-free interfaces between different martensitic variants experimentally observed in \Tn. We first motivate the formation of some non stress-free junctions by studying {the two well problem under suitable boundary conditions.}
We then give a mathematical characterisation of $V_{II}$ junctions within the theory of elasto-plasticity, and show that for deformation gradients as in \Tn\,our characterisation {agrees with} experimental results. Furthermore, we are able to prove that, {under suitable hypotheses that are verified in the study of \Tn,} $V_{II}$ junctions are strict weak local minimisers of a simplified energy functional for martensitic transformations in the context of elasto-plasticity.


\end{abstract}

\section{Introduction}
Martensitic phase transitions are abrupt changes occurring in the crystalline structure of certain alloys or ceramics when the temperature is moved across a critical threshold. The high temperature phase is called austenite or parent phase, and usually enjoys cubic symmetry, while the low temperature phase is called martensite, and has lower symmetry (e.g., tetragonal, orthorhombic, monoclinic \cite{Batt}). For this reason, martensite has usually more variants, which are symmetry related, and which in experiments often appear finely mixed. 
Martensitic phase transitions are important because they are the physical motivation of shape memory, the ability of certain materials to recover on heat deformations which are apparently plastic.

After the seminal work of Ball and James \cite{BallJames1} modelling martensitic phase transitions in the context of nonlinear elasticity (see Section \ref{Nonlin}), a vast literature has been developed to study energy minimisers, and energy minimising sequences {for energy functionals describing this physical phenomenon at a continuum scale. Indeed, energy minimising sequences can be interpreted as microstructures, that is finely mixed martensitic variants, with no elastic energy at the macroscopic scale (see e.g., \cite{BallJames2,Batt,Muller} and references therein).}\ A key tool to understand and predict martensitic microstructures is the Hadamard jump condition (see e.g., \cite[Prop. 1]{BallJames1}) stating that if a continuous function $\vc y\colon\R^3\to\R^3$ is such that 
$$
\nabla\vc y(\vc x) = \mt F_1 \quad\text{ a.e. in $\{\vc x\cdot\vc m<0\}$},\qquad\text{ and }\qquad
\nabla\vc y(\vc x) = \mt F_2 \quad\text{ a.e. in $\{\vc x\cdot\vc m>0\}$},
$$
for some unit vector $\vc m\in\mathbb{S}^2$ and two matrices $\mt F_1,\mt F_2\in\R^{3\times3}$, then 
\beq
\label{rank1conn}
\mt F_1-\mt F_2 = \vc b\otimes \vc m,\qquad\text{for some $\vc b\in\R^3$.}
\eeq
This condition imposes some necessary compatibility between two martensitic variants, or between two average martensitic deformation gradients representing different homogeneous microstructures, in order to have stress-free junctions. If \eqref{rank1conn} holds, then we say that $\mt F_1,\mt F_2$ are compatible across the plane $\{\vc x\cdot \vc m =0\}$.
Compatibility is a key ingredient not only to understand microstructures (see e.g., \cite{BallJames1,Batt}) but also to understand hysteresis of the phase transformation \cite{JamesMuller} and recently to construct materials undergoing ultra-reversible phase transformations \cite{JamesNew,JamesHyst}. Nonetheless, in the literature experiments are reported where the above compatibility is not observed right off the phase interface, and where the phase junctions are not stress free. {More precisely, martensite is elastically or plastically deformed to achieve compatibility between variants/phases.} For example, in Figure \ref{VI} we show the situation of $V_I$ junctions observed in the cubic to orthorhombic transformation in \Tn \cite{Inamura}. We have two different deformation gradients $\mt F_1,\mt F_2\in\R^{3\times 3}$ corresponding to two different martensitic variants, and the identity matrix $\mt 1$, deformation gradient in the austenite region. In the case of $V_I$ junctions we have
$$
\rank(\mt F_1-\mt F_2)=1,\qquad \rank(\mt F_1-\mt 1)> 1,\qquad \rank(\mt F_2-\mt 1)>1,
$$
and therefore the interfaces between austenite and martensite are not stress-free close to the junction between $\mt F_1$ with $\mt F_2$. Similarly, in the case of $V_{II}$ junctions (see Figure \ref{VII}), also observed in \Tn \cite{Inamura}, we have
\beq
\label{IncompII}
\rank(\mt F_1-\mt F_2)>1,\qquad \rank(\mt F_1-\mt 1)=1,\qquad \rank(\mt F_2-\mt 1)=1,
\eeq
and therefore $\mt F_1$ and $\mt F_2$ are not compatible. In Figure \ref{BallS} we show an incompatible junction between the two average deformation gradients $\mt F_1,\mt F_2\in\R^{3\times3}$ representing the average of the martensitic microstructures on the left and on the right of the red line \cite{BallS1,BallS2}. In this case, as for the $V_{II}$ junctions, \eqref{IncompII} holds. Non stress-free phase interfaces have also been observed in the X--interface configuration (Figure \ref{Ruddock}) for which we refer the reader to \cite{RuddockEx,Ruddock}.
\begin{figure}
\centering
\begin{subfigure}{.40\textwidth}
\begin{tikzpicture}[scale=0.40]
   \draw[black, thick] ({0},{0}) -- ({-4},{-8});
   \draw[red, thick] ({0},{0}) -- ({-4},{-8});
   \draw[red, thick] ({0},{0}) -- ({4},{-8});
   \draw[black, thick] ({0},{4+2}) -- ({-4-2},{-8+2});
   \draw[red, thick] ({0},{4+2}) -- ({-4-2},{-8+2});
   \draw[red, thick] ({0},{4+2}) -- ({4+2},{-8+2});
   \draw[black, thick] ({0},{4+2}) -- ({0},{0});
   \filldraw [red] ({2.5*cos(90)+2},{2.5*sin(90)+4}) circle (0pt) node[blue,anchor=north west,black] {$\mt 1$};
\filldraw [red] ({2.5*cos(90)},{2.5*sin(90)-5}) circle (0pt) node[blue,anchor=north,black] {$\mt 1$};
\filldraw [red] ({0.5*sin(240)+2.5*cos(240)+1},{-0.5*cos(240)+2.5*sin(240)+2}) circle (0pt) node[red,anchor=north east,black] {$\mt F_1$};
\filldraw [red] ({-0.5*sin(300)+2.5*cos(300)-1},{0.5*cos(300)+2.5*sin(300)+2}) circle (0pt) node[blue,anchor=north west,black] {$\mt F_2$};
\end{tikzpicture}
\caption{\label{VI}}
\end{subfigure}%
\hfill
\begin{subfigure}{.40\textwidth}
\begin{tikzpicture}[scale=0.40]
   \draw[black, thick] ({0},{0}) -- ({-4},{-8});
   \draw[black, thick] ({0},{0}) -- ({4},{-8});
   \draw[black, thick] ({0},{4+2}) -- ({-4-2},{-8+2});
   \draw[black, thick] ({0},{4+2}) -- ({4+2},{-8+2});
   \draw[red, thick] ({0},{4+2}) -- ({0},{0}); 
   \filldraw [red] ({2.5*cos(90)+2},{2.5*sin(90)+4}) circle (0pt) node[blue,anchor=north west,black] {$\mt 1$};
\filldraw [red] ({2.5*cos(90)},{2.5*sin(90)-5}) circle (0pt) node[blue,anchor=north,black] {$\mt 1$};
\filldraw [red] ({0.5*sin(240)+2.5*cos(240)+1},{-0.5*cos(240)+2.5*sin(240)+2}) circle (0pt) node[red,anchor=north east,black] {$\mt F_1$};
\filldraw [red] ({-0.5*sin(300)+2.5*cos(300)-1},{0.5*cos(300)+2.5*sin(300)+2}) circle (0pt) node[blue,anchor=north west,black] {$\mt F_2 $};
\end{tikzpicture}
\caption{\label{VII}}
\end{subfigure}%
\\
\begin{subfigure}{.40\textwidth}
\begin{tikzpicture}[scale=0.50]
   \draw[red, thick] ({0},{8.85}) -- ({0},{2});
 \foreach \x in {1,2,3}{
\draw[black, thick] ({-6},{\x*2-1}) -- ({-0.5},{\x*2+1});
}
 \foreach \x in {1,2,3}{
\draw[black, thick] ({-6},{\x*2+1.2-1}) -- ({-0.5},{\x*2+1.2+1});
}
 \foreach \x in {1,2,3}{
\draw[black, thick] ({6},{\x*2-2}) -- ({0.5},{\x*2+1});
}
 \foreach \x in {1,2,3}{
\draw[black, thick] ({6},{\x*2+0.8-2}) -- ({0.5},{\x*2+0.8+1});
}
\end{tikzpicture}
\caption{\label{BallS}}
\end{subfigure}%
\hfill
\begin{subfigure}{.40\textwidth}
\begin{tikzpicture}[scale=0.50]

   \draw[red, thick] ({0},{0}) -- ({4},{4});
   \draw[red, thick] ({0},{0}) -- ({-4},{4});
 \foreach \x in {1,2,3,4}{
\draw[black, thick] ({-4},{(\x-1)*2-4}) -- ({ -0.5 - (\x-1)},{-0.5 + (\x-1)});
}
 \foreach \x in {1,2,3,4}{
\draw[black, thick] ({-4},{(\x-1)*2-4+0.3}) -- ({ -0.4 - (\x-1) - 0.4*0.7071},{-0.5 + (\x-1)+0.4- 0.4*0.7071});
}
 \foreach \x in {1,2,3,4}{
\draw[black, thick] ({+4},{(\x-1)*2-4}) -- ({ +0.5 + (\x-1)},{-0.5 + (\x-1)});
}
 \foreach \x in {1,2,3,4}{
\draw[black, thick] ({+4},{(\x-1)*2-4+0.3}) -- ({ +0.4 + (\x-1) + 0.4*0.7071},{-0.5 + (\x-1)+0.4- 0.4*0.7071});
}

\filldraw [red] ({0},{2.5}) circle (0pt) node[blue,anchor=south,black] {$\mt 1$};
\filldraw [red] ({0},{-2.5}) circle (0pt) node[red,anchor=north,black] {$\mt F_3$};
\end{tikzpicture}
\caption{\label{Ruddock}}
\end{subfigure}%
\caption{\label{figureJIn}Examples of non stress free junctions (in red in the picture) experimentally observed in martensitic transformations: \ref{VI}--\ref{VII} show respectively a $V_I$ and a $V_{II}$ junction, observed for example in \cite{Inamura,InamuraI,InamuraIII}. The case \ref{BallS} is a generalisation of $V_{II}$ junctions, where instead of two single variants of martensite we have two martensitic laminates, both compatible on average with austenite but not with each other (see \cite{BallS1,BallS2}). In Figure \ref{Ruddock} an example of an X--interface, experimentally observed in \cite{RuddockEx}, and studied in \cite{Ruddock}. {In Figure \ref{VI} and in Figure \ref{VII}, at the non stress-free junctions (red lines in the pictures) defects are observed in experiments.}
}
\end{figure}
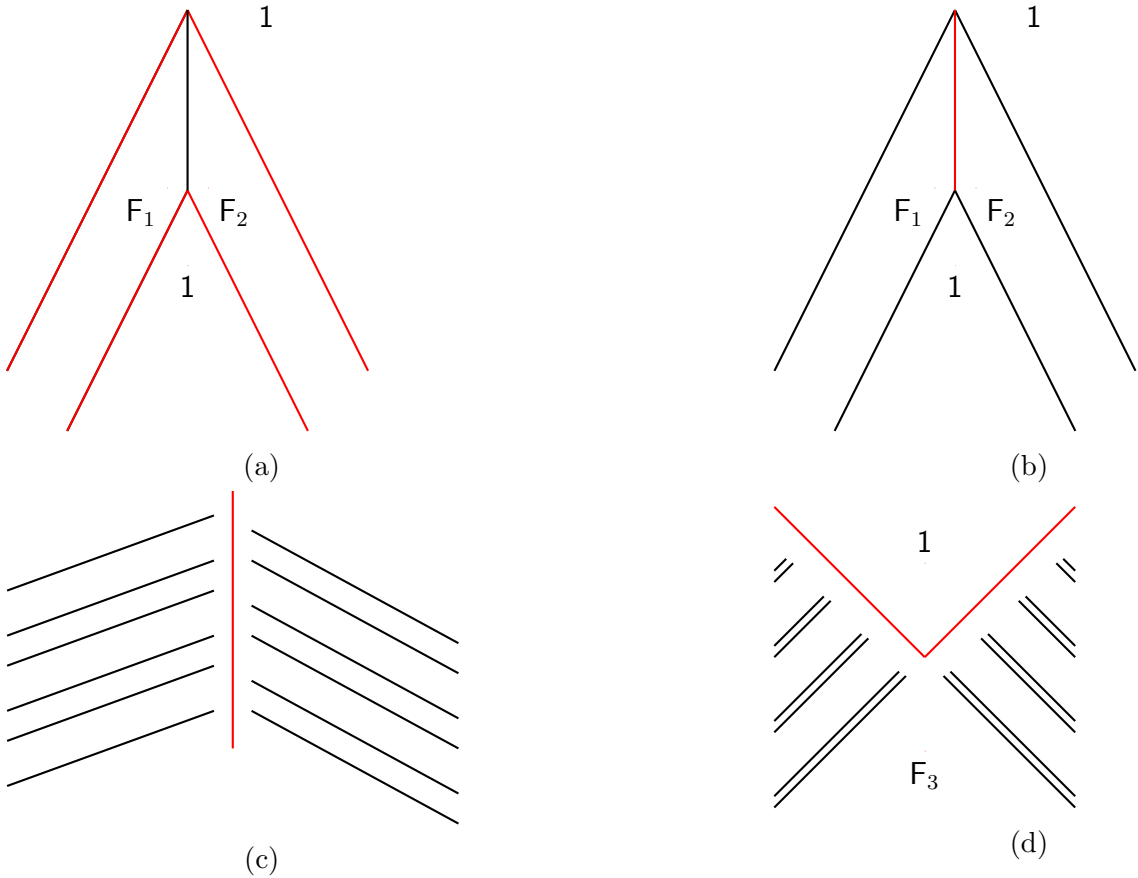

The following approach to measure the incompatibility between non-stress free junctions has been proposed in \cite{BallS1}. Assuming that $\mt F_1,\mt F_2\in\R^{3\times3}$ are such that $\rank(\mt F_1-\mt F_2)>1$, and that $\mt F_2^{-T}\mt F_1^T\mt F_1\mt F_2^{-1}$ has middle eigenvalue one, \cite[Prop. 4]{BallJames1} guarantees the existence of two rotations $\mt R_1,\mt R_2\in SO(3)$ such that $\rank(\mt F_1-\mt R_i\mt F_2) = 1$ for $i=1,2$. The incompatibility of $\mt F_1,\mt F_2$ can hence be measured by taking the minimum between the rotation angle of $\mt R_1$, and the rotation angle of $\mt R_2$. This is in agreement with the experimental results in \cite{BallS1,Inamura} where the observed non stress-free junctions are the ones where $\min\{angle(\mt R_1),angle(\mt R_2)\}$ is small. Another way to measure how far three deformations gradients, say $\mt F_1,\mt F_2,\mt 1$ are to form a triple junction, that is to be all pairwise rank one connected, can be found in \cite{FDP3}. However, in the case for example of \Tn \cite{Inamura} these approaches do not allow to predict when two martensitic variants will form a $V_I$ or a $V_{II}$ junction. {Indeed, experiments show that some martensitic variants tend to meet only in $V_I$ junctions, while others form just $V_{II}$ junctions (see e.g., \cite[Table 4]{Inamura})}.

The aim of this work is to study $V_{II}$ junctions and their stability in the context of elasto-plasticity. The paper is organised as follows: in Section \ref{Nonlin} we recall the nonlinear elasticity theory for martensitic phase transitions, and we introduce a simplified energy functional $I$ to describe the physical phenomenon when plastic shears occur. This energy functional is very general as it describes the transformation to all possible martensitic variants and all possible slip systems for body centred cubic austenite (as in \Tn). In Section \ref{Rigid} we give a partial explanation of why we observe non stress-free junctions of $V_{II}$ type or like the ones in Figure \ref{BallS}. Our explanation is the following: these type of junctions usually form when two different plates of martensite, with deformation gradients $\mt F_1,\mt F_2$, nucleate at different points in the domain, and expand until they meet (see Figure \ref{nucleazione VII} and Figure \ref{nucleazione BS}). We hence consider a bounded domain $\Omega\subset\R^3$ as in Figure \ref{domini} and two martensitic variants represented by their stretch tensors $\mt U_1,\mt U_2\in\R^{3\times3}_{Sym^+}$. We prove that, under some further geometric hypotheses which are verified by the non stress-free junctions in \Tn \cite{Inamura} and in Ni\textsubscript{65}Al\textsubscript{35} \cite{BallS1}, there exists a one-to-one map $\vc y\in W^{1,\infty}(\Omega;\R^3)$ satisfying
\beq
\label{diffInclIntro}
\begin{cases}
\nabla \vc y(\vc x) \in \bigl(SO(3)\mt U_1\cup SO(3)\mt U_2\bigr)^{qc},\qquad&\text{a.e. $\vc x\in\Omega$},\\
\vc y(\vc x) = \mt F_1\vc x,\qquad&\text{on $\Gamma_1,$}\\
\vc y(\vc x) = \mt F_2\vc x,\qquad&\text{on $\Gamma_2,$}
\end{cases}
\eeq
with $\mt F_1,\mt F_2\in \bigl(SO(3)\mt U_1\cup SO(3)\mt U_2\bigr)^{qc}$ if and only if $\rank(\mt F_1-\mt F_2)\leq1.$ Therefore, no stress-free microstructure built with the two martensitic variants $\mt U_1,\mt U_2$ can fill the domain $\Omega$ and match the previously nucleated plates $\mt F_1,\mt F_2$. 

In Section \ref{ShearComp} we study when two simple shears $\mt S_1,\mt S_2\in\R^{3\times3}$ are such that 
\beq
\label{ShearComp}
\rank(\mt F_1\mt S_1-\mt F_2\mt S_2)\leq 1, 
\eeq 
given $\mt F_1,\mt F_2$ with $\rank(\mt F_1-\mt F_2)=2$. 

In Section \ref{Minim} we give a mathematical characterisation of $V_{II}$ junctions as junctions reflecting \eqref{IncompII}, where the compatibility between $\mt F_1,\mt F_2$ is achieved thanks to single slips (and hence thanks to plastic effects). We also give sufficient conditions for $V_{II}$ junctions to be strict weak local minimisers for the simplified energy $I$ introduced in Section \ref{Nonlin}. 

In Section \ref{Appl} we study the possibility to form $V_{II}$ junctions in a one parameter family of deformation gradients, which approximates well the phase transformation in \Tn. The obtained results are discussed at the end of the section, and seem to be in good agreement with experimental observations. Finally, in Section \ref{ConcludingRmk} we give some concluding remarks and possible directions to extend the present work.
%
%
%

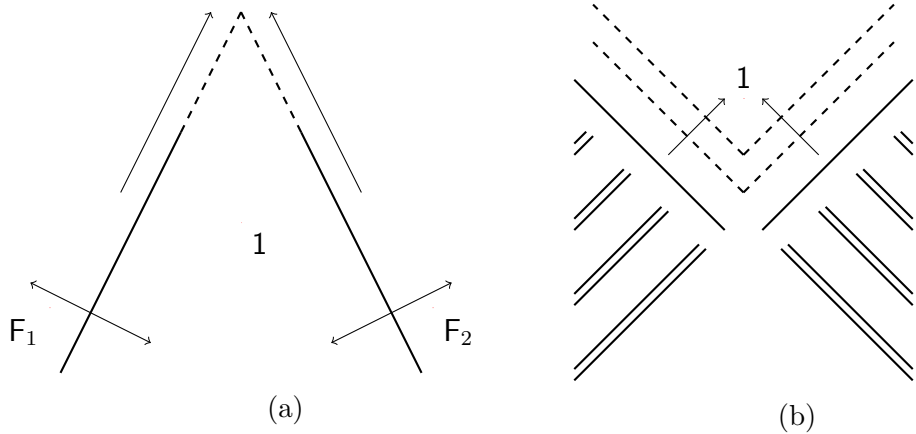
\begin{figure}
\centering

\begin{subfigure}{.45\textwidth}
\begin{tikzpicture}[scale=0.80]
   \draw[black, thick] ({-1,-2}) -- ({-3},{-6});
   \draw[dashed,black, thick] (-1,-2) -- (0,0);
   \draw[black, thick] ({1,-2}) -- ({3},{-6});
   \draw[dashed,black, thick] (1,-2) -- (0,0);

   \draw[black, ->] ({-2,-4+1}) -- ({-0.5},{-1+1});
   \draw[black, ->] ({2,-4+1}) -- ({0.5},{-1+1});

   \draw[black, <->] ({-2.5-1,-5+0.5}) -- ({-2.5+1},{-5-0.5});
   \draw[black, <->] ({2.5-1,-5-0.5}) -- ({2.5+1},{-5+0.5});

\filldraw [red] ({2.5*cos(90)},{2.5*sin(90)-6}) circle (0pt) node[blue,anchor=north west,black] {$\mt 1$};
\filldraw [red] ({0.5*sin(240)+2.5*cos(240)-1.5},{-0.5*cos(240)+2.5*sin(240)-3}) circle (0pt) node[red,anchor=north east,black] {$\mt F_1$};
\filldraw [red] ({-0.5*sin(300)+2.5*cos(300)+1.5},{0.5*cos(300)+2.5*sin(300)-3}) circle (0pt) node[blue,anchor=north west,black] {$\mt F_2 $};
\end{tikzpicture}
\caption{\label{nucleazione VII}}
\end{subfigure}%
\begin{subfigure}{.35\textwidth}
\begin{tikzpicture}[scale=0.50]
   \draw[black, thick] ({0+0.5},{0}) -- ({4+0.5},{4});
   \draw[black, thick] ({0-0.5},{0}) -- ({-4-0.5},{4});
 \foreach \x in {1,2,3,4}{
\draw[black, thick] ({-4-0.5},{(\x-1)*2-4}) -- ({ -0.5 - (\x-1)-0.5},{-0.5 + (\x-1)});
}
 \foreach \x in {1,2,3,4}{
\draw[black, thick] ({-4-0.5},{(\x-1)*2-4+0.3}) -- ({ -0.4 - (\x-1) - 0.4*0.7071-0.5},{-0.5 + (\x-1)+0.4- 0.4*0.7071});
}
 \foreach \x in {1,2,3,4}{
\draw[black, thick] ({+4+0.5},{(\x-1)*2-4}) -- ({ +0.5 + (\x-1)+0.5},{-0.5 + (\x-1)});
}
 \foreach \x in {1,2,3,4}{
\draw[black, thick] ({+4+0.5},{(\x-1)*2-4+0.3}) -- ({ +0.4 + (\x-1) + 0.4*0.7071+0.5},{-0.5 + (\x-1)+0.4- 0.4*0.7071});
}
\filldraw [red] ({0},{3.5}) circle (0pt) node[blue,anchor=south,black] {$\mt 1$};
   \draw[black, dashed, thick] ({0},{1}) -- ({4},{4+1});
   \draw[black, dashed, thick] ({0},{1}) -- ({-4},{4+1});
      \draw[black, dashed, thick] ({0},{2}) -- ({4},{4+2});
   \draw[black, dashed, thick] ({0},{2}) -- ({-4},{4+2});
   \draw[black, ->] ({-2,2}) -- ({-0.5},{3.5});
   \draw[black, ->] ({2,2}) -- ({0.5},{3.5});
\end{tikzpicture}
\caption{\label{nucleazione BS}}
\end{subfigure}%
\caption{\label{domini} Formation of $V_{II}$ junctions in \Tn \cite{Inamura} and of non stress-free junctions in Ni\textsubscript{65}Al\textsubscript{35} \cite{BallS1}, respectively represented in Figure \ref{nucleazione VII} and Figure \ref{nucleazione BS}. 
In the former, it is experimentally observed that two different plates of martensite $\mt F_1,\mt F_2$ nucleate in an austenite domain and propagate until they meet. When the thickness of the two martensite plates increases, a $V_{II}$ junction is formed. In the latter, two different laminates of martensite nucleate at two different points of the sample and expand until they coalesce \cite{BallS1}. Further expansion leads to a non stress-free junction. In both cases the average deformation gradient in the martensite regions is very close to be rank one connected to the identity matrix, consistently with the moving mask approximation in \cite{FDP2}. In the pictures, the arrows represent the directions of expansion of the phase boundaries.
}
\end{figure}

%

\begin{figure}
\centering
\begin{subfigure}{.45\textwidth}
\begin{tikzpicture}[scale=0.40]
   \draw [step=0.05,dashed, black,thick,domain=0:180] plot ({5*cos(\x)}, {5*sin(\x)});
   \draw [step=0.05, black,thick,domain=180:240] plot ({5*cos(\x)}, {5*sin(\x)});
   \draw [step=0.05,black,thick,domain=300:360] plot ({5*cos(\x)}, {5*sin(\x)});
   \draw[red, thick] ({5*cos(240)},{5*sin(240)}) -- (0,0);
   \draw[blue, thick] ({5*cos(300)},{5*sin(300)}) -- (0,0);

   \draw [step=0.05, black,thick,domain=0:240] plot ({5*cos(\x)}, {5*sin(\x)+4});
   \draw [step=0.05,black,thick,domain=300:360] plot ({5*cos(\x)}, {5*sin(\x)+4});
   \draw[red, thick] ({5*cos(240)},{5*sin(240)+4}) -- (0,4);
   \draw[blue, thick] ({5*cos(300)},{5*sin(300)+4}) -- (0,4);
   \draw[blue, thick] ({5*cos(300)},{5*sin(300)+4}) -- ({5*cos(300)},{5*sin(300)});
   \draw[red, thick] ({5*cos(240)},{5*sin(240)}) -- ({5*cos(240)},{5*sin(240)+4});
   \draw[black, thick] ({5*cos(180)},{5*sin(180)}) -- ({5*cos(180)},{5*sin(180)+4});
   \draw[black, thick] ({5*cos(0)},{5*sin(0)}) -- ({5*cos(0)},{5*sin(0)+4});
   \draw[red, thick] (-0.02,0) -- (-0.02,4);
   \draw[blue, thick] (0.02,0) -- (0.02,4);
\filldraw [red] ({2.5*cos(90)},{2.5*sin(90)+4}) circle (0pt) node[blue,anchor=north west,black] {$\Omega$};
\filldraw [red] ({0.5*sin(240)+2.5*cos(240)+1},{-0.5*cos(240)+2.5*sin(240)+2}) circle (0pt) node[red,anchor=north east,black] {$\Gamma_1$};
\filldraw [red] ({-0.5*sin(300)+2.5*cos(300)-1},{0.5*cos(300)+2.5*sin(300)+2}) circle (0pt) node[blue,anchor=north west,black] {$\Gamma_2 $};
\end{tikzpicture}
\end{subfigure}%
\begin{subfigure}{.35\textwidth}
\begin{tikzpicture}[scale=0.50]
   \draw [step=0.05, black,thick,domain=0:240] plot ({5*cos(\x)}, {5*sin(\x)});
   \draw [step=0.05,black,thick,domain=300:360] plot ({5*cos(\x)}, {5*sin(\x)});
   \draw[red, thick] ({5*cos(240)},{5*sin(240)}) -- (0,0);
   \draw[blue, thick] ({5*cos(300)},{5*sin(300)}) -- (0,0);
  \draw [->,ultra thick] ({3.5*cos(300)},{3.5*sin(300)}) -- ({1.5*sin(300)+3.5*cos(300)},{-1.5*cos(300)+3.5*sin(300)}); 
  \draw [->,ultra thick] ({3.5*cos(240)},{3.5*sin(240)}) -- ({-1.5*sin(240)+3.5*cos(240)},{1.5*cos(240)+3.5*sin(240)}); 
\filldraw [red] ({-1*sin(240)+3.5*cos(240)},{1*cos(240)+3.5*sin(240)}) circle (0pt) node[anchor=north east,black] {$\vc n_1 $};
\filldraw [red] ({1*sin(300)+3.5*cos(300)},{-1*cos(300)+3.5*sin(300)}) circle (0pt) node[anchor=north west,black] {$\vc n_2 $};
\filldraw [red] ({0.5*sin(240)+2.5*cos(240)},{-0.5*cos(240)+2.5*sin(240)}) circle (0pt) node[red,anchor=north east,black] {$\Gamma_1$};
\filldraw [red] ({-0.5*sin(300)+2.5*cos(300)},{0.5*cos(300)+2.5*sin(300)}) circle (0pt) node[blue,anchor=north west,black] {$\Gamma_2 $};
\filldraw [red] ({2.5*cos(90)},{2.5*sin(90)}) circle (0pt) node[blue,anchor=north west,black] {$\Omega$};
\end{tikzpicture}
\end{subfigure}%
\caption{\label{domini} Representation of $\Omega,$ $\Gamma_1$ and $\Gamma_2$ as defined in \eqref{Omegadef} (on the left), and their projection on the plane spanned by $\vc n_1,\vc n_2$ (on the right).}
\end{figure}
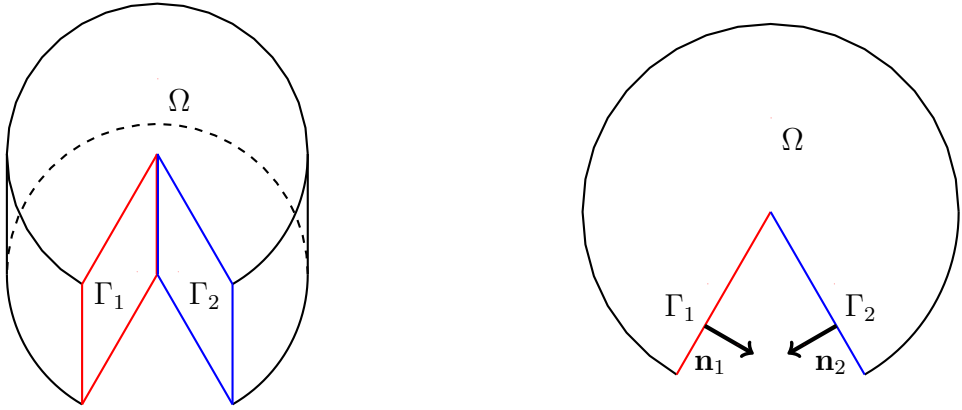

\section{A model for martensitic transformations with plastic shears}
\label{Nonlin}
The most successful mathematical theory to describe martensitic phase transitions at a continuum level is based on the theory of nonlinear elasticity and was first introduced in \cite{BallJames1}. This model has been successfully used to understand laminates and other microstructures (see \cite{BallJames1,Batt}), as much as the shape-memory effect (see \cite{BattSme}), and, more recently, hysteresis (see \cite{JamesMuller}).

In the nonlinear elasticity model, changes in the crystal lattice are interpreted as elastic deformations in the continuum mechanics framework, and legitimised by the Cauchy-Born hypothesis. The deformations minimize hence a free energy
\beq
\label{energia}
\EE(\vc y,\theta)  = \mint W_e(\nabla \vc y(\vc x),\theta)\,\mathrm d\vc x.
\eeq
Here, $\theta$ denotes the temperature of the crystal, the domain (open and connected) $\Omega$ stands for the region occupied by a single crystal in the undistorted defect-free austenite phase at the transition temperature $\theta=\theta_T$, while $\vc y(\vc x)$ denotes the position of the particle $\vc x\in\Omega$ after the deformation of the lattice has occurred. By $ W_e$ we denote the free-energy density, depending on the temperature $\theta$ and the deformation gradient $\nabla \vc y$. The behaviour of $ W_e$ on $\theta$ must reflect the phase transition, that is when $\theta<\theta_T$ and $\theta>\theta_T$, the energy is respectively minimised by martensite and austenite. At $\theta=\theta_T$ all phases are energetically equivalent. 

Below, we assume $\theta<\theta_T$ to be fixed, and we consider $ W_e$ to be defined by (omitting for ease of notation the dependence on $\theta$)
\[
 W_e(\mt F) =
\begin{cases}
0,\qquad &\text{ if  $\mt F \in \bigcup_{i=1}^N SO(3)\mt U_i$,}\\
+\infty,\qquad &\text{otherwise},
\end{cases}
\]
where $\mt U_i = \mt U_i(\theta)\in\R^{3\times3}_{Sym^+}$ are the $N$ positive definite symmetric matrices corresponding to the transformation from austenite to the $N$ variants of martensite at temperature $\theta$. Here and below $\R^{3\times3}_{Sym^+}$ represents the set of $3\times3$ symmetric and positive definite matrices. 
We remark that, defined $\mathcal{P}_a,\mathcal{P}_m$ as the point groups of austenite and martensite respectively (i.e., the sets of rotations that map the austenite and martensite lattices back to themselves), and denoting by $\#$ their cardinality, we have $N = \frac{\#\mathcal{P}_a}{\#\mathcal{P}_m}$. Also, for each $\mt U_i,\mt U_j$ there exists $\mt R\in \mathcal{P}_a$ such that $\mt R^T\mt U_j\mt R = \mt U_i$, so that $\mt U_i,\mt U_j$ share the same eigenvalues. We point out that this energy satisfies frame indifference. That is, for all $\mt F\in\R^{3\times 3}$ and all rotations $\mt R\in SO(3)$, $ W_e(\mt R \mt F)= W_e(\mt F)$, reflecting the invariance of the free-energy density under rotations. Furthermore, $ W_e$ respects lattice symmetries, i.e., $ W_e(\mt {FQ})= W_e(\mt F)$ for all $\mt F\in\mathcal \R^{3\times3}$ and all rotations $\mt Q\in\mathcal{P}_{a}$. Such a $ W_e$ has been already considered for example in \cite{BallJames1,BallChuJames,BallKoumatosQC,FDP1} and corresponds to the physical situation where the elastic constants are infinity, which, as remarked in \cite{BallChuJames}, is usually a reasonable approximation when studying martensitic phase transitions with no external (or at least small) load. Considering $ W_e$ to be $+\infty$ out of the energy wells is also known as the elastically rigid approximation, and is often used in the context of elasto-plasticity since elastic effects in metals are usually much smaller than plastic ones (see e.g., \cite{OR99}).
\\

We now want to take in account the presence of plastic effects in the nonlinear elasticity model. Following 
\cite{ReinaConti,ReinaContiS} and references therein, we use the multiplicative decomposition of the deformation gradient
$$
\nabla \vc y = \mt F^e\mt F^p, 
$$
where $\mt F^e, \mt F^p$ respectively represent the elastic and the plastic component of the deformation gradient. The former describes the part of the deformation gradient which is reversible, while the latter captures the irreversible deformations given by the slip of {atoms along planes.} In solid crystals, {atoms can slip just in particular directions on particular planes.} For this reason, $\mt F^p$ must be of the form
$$
\mt F^p  = \mt 1+ s\vcg\phi\otimes \vcg{\psi}
$$
where $s\in\R$, $\vcg \phi\in\R^3$, $\vcg\psi\in\mathbb{S}^2$, $\vcg\phi\cdot\vcg\psi=0$, and $\vcg \phi\otimes \vcg \psi\in \mathcal{S}\subset\R^{3\times 3}.$ Here, $\vcg\phi$ is called slip direction and $\vcg\psi$ is called the slip plane, while $s$ is the amount of shear. The set $\mathcal S$ is the set of all possible {slip systems}. For
body centred cubic austenite, which is the case of \Tn, 
there are six planes of type $\{1,1,0\}$ each with two orthogonal $\langle \bar 1,1,1\rangle$ $\langle \bar 1,1,-1\rangle$ directions, twenty-four planes $\{1,2,3\}$ and twelve planes $\{1,1,2\}$ each with one orthogonal $\langle \bar 1,1,1\rangle$ direction. 

Following the approach of \cite{CDK,AD,DRMD} and references therein, we adopt the time discrete variational approach to elasto-plasticity \cite{OR99}, restricting ourselves to the first time step where most of the plastic events take place. 
We further assume cross hardening \cite{AD}, which means that activity in one slip system suppresses the activity in all other slip systems at the same point. For this reason, we choose a plastic energy density $W_p$ of the type
\[
W_p:=
\begin{cases}
f(|s|), \qquad &\text{ if $\mt F^p = \mt 1+s\vcg\phi\otimes\vcg\psi,$ and $\vcg\phi\otimes\vcg\psi\in\mathcal{S}$,}\\
+\infty, \qquad &\text{ otherwise,}
\end{cases}
\]
where $f\colon[0,\infty)\to[0,\infty)$ is assumed to be continuous, strictly monotone and to satisfy $f(0)=0.$ Here, as for $W_e$, $W_p$ could be finite and continuous. This approximation however simplifies the analytical study of the energy and allows to neglect any dependence of the results on the shape of the energy density out of its minima. We are now ready to introduce an elasto-plastic energy density $W$ defined as
\[
W(\mt F):= \min \bigl\{W_e(\mt F^e) + W_p(\mt F^p) : \mt F^e\mt F^p =\mt F\bigr\},
\]
and an energy functional $I$ for the system
\beq
\label{Idef}
I(\vc y,\Omega) = \mint W(\nabla\vc y)\,\mathrm d\vc x.
\eeq
We remark that the energy $I$ is not weakly lower semicontinuous and in general minimisers do not exist.

\section{A rigidity result for the two well problem}
\label{Rigid}
In this section, we study the existence of solutions to Problem \eqref{diffInclIntro}. As explained in the introduction, this gives a way to justify the formation of non stress-free junctions between martensitic plates.
 
Let $\vc n_1,\vc n_2\in \mathbb{S}^2$, $\vc n_1\times\vc n_2\neq 0$ and let us set $\vc n_\perp := \frac{\vc n_1\times\vc n_2}{|\vc n_1\times\vc n_2|}$. For $R>0$, we define (see Figure \ref{domini})
\beq
\begin{split}
\label{Omegadef}
\Omega &:= \bigl\{\vc x\in\R^3\colon \min\{\vc x\cdot\vc n_1,\vc x\cdot\vc n_2\} < 0 
,\;\vc x\cdot\vc n_\perp \in (0,1)\text{ and }|\vc x-\vc n_\perp(\vc n_\perp\cdot \vc x)|<R \bigr\},\\
\Gamma_1 &: = \bigl\{\vc x\in\partial\Omega\colon \vc x\cdot\vc n_1= 0 \text{ and }\vc x\cdot\vc n_2>0 \bigr\},\\
\Gamma_2 &:= \bigl\{\vc x\in\partial\Omega\colon \vc x\cdot\vc n_2= 0 \text{ and }\vc x\cdot\vc n_1>0  \bigr\}.
\end{split}
\eeq
Theorem \ref{Incompatible wells} below states that, under suitable boundary conditions, the differential inclusion \eqref{diffInclIntro} has no solution. More precisely, under our assumptions, the boundary conditions on $\Gamma_1,\Gamma_2$ need to satisfy a compatibility condition, which is unexpected and strongly dictated by the structure of the two well problem. Also, in order to have no solution to the two well problem, we do not need to impose boundary conditions on the whole boundary of the domain, but just on a corner of it (namely, on $\Gamma_1\cup\Gamma_2$). By the work in \cite{MullerSverak} we know that, under suitable boundary conditions, there are infinitely many solutions to the differential inclusion $\nabla \vc y(\vc x) \in \bigl(SO(3)\mt U_1\cup SO(3)\mt U_2\bigr)$, a.e. $\vc x\in\Omega$. 
Our result provides an example of boundary conditions where the convex-integration techniques used in \cite{MullerSverak} cannot be applied. Further, our result holds also for the relaxed differential inclusion $\nabla \vc y(\vc x) \in \bigl(SO(3)\mt U_1\cup SO(3)\mt U_2\bigr)^{qc},$ a.e. $\vc x\in\Omega$. 
The proof relies on a result by Ball and James \cite{BallJamesPlane} which states that, after a suitable change of coordinates, in the two well problem there exists one direction (in the proof below $\vc u_2$) where the martensitic deformation coincides with a constant elongation/contraction composed with a constant rotation. The proof exploits the fact that this direction and this rotation must be coherent across the whole domain and compatible with the boundary conditions. The result reads as follows:
\begin{theorem}
\label{Incompatible wells}
Let $\mt U_1,\mt U_2\in \R^{3\times3}_{Sym^+}$ such that there exists $\hat{\vc e}\in\mathbb{S}^2$ satisfying
\beq
\label{ehatholds}
\mt U_1 = \bigl(2\hat{\vc e}\otimes\hat{\vc e}-\mt 1	\bigr)\mt U_2 \bigl(2\hat{\vc e}\otimes\hat{\vc e}-\mt 1	\bigr). 
\eeq
Suppose further that $\vc u_*:=\hat{\vc e}\times \mt U_1^2\hat{\vc e}$ is such that $\vc u_*\times \vc n_\perp\neq \vc 0$. Then, there exists $\vc y\in W^{1,\infty}(\Omega;\R^3)$ such that $\vc y$ is $1-1$ in $\Omega$,
\beq
\label{defKqc}
\nabla \vc y(\vc x) \in K^{qc} := \bigl(SO(3)\mt U_1\cup SO(3)\mt U_2\bigr)^{qc},\qquad\text{a.e. $\vc x\in\Omega$},
\eeq
and
\[
\vc y(\vc x) = 
\begin{cases}
\mt F_1\vc x,\qquad&\text{on $\Gamma_1,$}\\
\mt F_2\vc x,\qquad&\text{on $\Gamma_2,$}
\end{cases}
\]
for some $\mt F_1,\mt F_2\in K^{qc}$, if and only if 
there exists $\vc d\in\R^3$ such that
\beq
\label{conclusion}
\mt F_1-\mt F_2 = \vc d\otimes (\vc u_*\times \vc n_\perp).
\eeq
\end{theorem}
\begin{proof}
\textit{Necessity.} We first notice that $\Omega$ is Lipschitz, and therefore by Morrey's imbeddings $\vc y\in C^{0,1}(\overline{\Omega};\R^3)$ (see e.g., \cite{Adams}). Therefore, $\vc y$ is continuous on the line $\vc n_\perp$, that is
\beq
\label{rk1firststep}
(\mt F_1 - \mt F_2)\vc n_\perp = \vc 0.
\eeq 
Now, given \eqref{ehatholds}, \cite[Prop. 12]{JamesHyst} guarantees the existence of $\mt R\in SO(3)$, $\vc b\in\R^3$, $\vc m\in\mathbb{S}^2$ such that
\beq
\label{Twinning}
\mt R\mt U_2 = \mt U_1 + \vc b\otimes \vc m.
\eeq
Without loss of generality, we can take from standard twinning theory (see e.g., \cite{Batt}) $\vc m = \hat{\vc e},$ $\vc b = 2\Bigl(\frac{\mt U_1^{-1}\hat{\vc e}}{|\mt U_1^{-1}\hat{\vc e}|^2} -\mt U_1\hat{\vc e}\Bigr)$. The same results can be achieved by taking the only other solution of \eqref{Twinning}, that is $\vc b = \mt U_1\hat{\vc e},$ $\vc m = 2\Bigl(\hat{\vc e}- \frac{\mt U_1^2\hat{\vc e}}{|\mt U_1\hat{\vc e}|^2}\Bigr)$. We remark that by \eqref{Twinning} we have that $\det \mt U_2=\det \mt U_1 +\mt U_1^{-1}\vc m\cdot\vc b$ and hence, as $\det\mt U_1 = \det \mt U_2$, $\mt U_1^{-1}\vc m\cdot\vc b=0.$
Following the strategy of \cite{BallJames2}, let us define the orthonormal system of coordinates
$$
\vc u_1:=\frac{\mt U_1^{-1}\vc m}{|\mt U_1^{-1}\vc m|} ,\qquad \vc u_3 := \frac{\vc b}{|\vc b|}, \qquad\vc u_2 := \vc u_3\times\vc u_1,
$$
and let 
$$
\mt L:= \mt U^{-1}_1\bigl( \mt 1 -\delta\vc u_3\otimes \vc u_1\bigr),\qquad \delta:= \frac{1}{2}|\mt U_1^{-1}\vc m||\vc b|.
$$
Therefore, setting $\vc z(\vc x):=\vc y(\mt L\vc x)$ the problem becomes equivalent to finding a $1-1$ map $\vc z\in W^{1,\infty}(\mt L^{-1}\Omega;\R^3)$ such that
\beq
\label{defKqcT}
\nabla \vc z(\vc x) \in \bigl(SO(3)\mt S^-\cup SO(3)\mt S^+\bigr)^{qc},\qquad\text{a.e. $\vc x\in \Omega^L$},
\eeq
with $\mt S^\pm = \mt 1 \pm \delta\vc u_3\otimes\vc u_1,$ and
\beq
\label{bcz}
\vc z(\vc x) = 
\begin{cases}
\mt F_1\mt L\vc x,\qquad&\text{for every $\vc x\in\Gamma_1^L,$}\\
\mt F_2\mt L\vc x,\qquad&\text{for every $\vc x\in\Gamma_2^L.$}
\end{cases}
\eeq
Here,
$$
\Omega^L:= \bigl\{\vc x\in\R^3\colon \mt L\vc x\in\Omega \bigr\},\quad 
\Gamma_1^L:= \bigl\{\vc x\in\R^3\colon \mt L\vc x\in\Gamma_1 \bigr\},\quad 
\Gamma_2^L:= \bigl\{\vc x\in\R^3\colon \mt L\vc x\in\Gamma_2 \bigr\}.
$$
Following \cite{BallJames2}, we can characterise the set $K_L:=\bigl(SO(3)\mt S^-\cup SO(3)\mt S^+\bigr)^{qc}$ 
as
\[
K_L= \Set{
\mt F\in\R^{3\times3}
 \bigg|\; \text{\parbox{3.5in}{\centering 
$\mt F^T\mt F = \alpha\vc u_1\otimes\vc u_1 + \vc u_2\otimes\vc u_2 + \gamma \vc u_3\otimes \vc u_3 + \beta \vc u_1\odot\vc u_3,$ $0<\alpha\leq 1+\delta^2$, $0<\gamma\leq 1$, $\alpha\gamma-\beta^2 = 1$
 }}},
\]
and where we denoted $\vc u_1\odot\vc u_3 = \vc u_1\otimes\vc u_3+\vc u_3\otimes\vc u_1$. Let us now define 
$$s_i:=\vc x\cdot \vc u_i,\qquad\alpha_i := \mt L^T\vc n_1\cdot \vc u_i,\qquad\beta_i := \mt L^T\vc n_2\cdot \vc u_i,
$$
and remark that \cite{BallJamesPlane} together with the definition of $K_L$ yield
\beq
\label{plainstrain}
\vc z = \mt Q\bigl(z_1(s_1,s_3)\vc u_1 + s_2 \vc u_2 + z_3(s_1,s_3) \vc u_3\bigr),
\eeq
for some Lipschitz scalar functions $z_1,z_2$ and some $\mt Q\in SO(3)$. Assume now that $\alpha_3\neq 0$, the other cases 
can be treated similarly to deduce \eqref{sottovedremo} below. In this case, the fact that $\vc z(\vc x) = \mt F_1\mt L\vc x$ on $\Gamma_1^L$ (cf. \eqref{bcz}) together with $\mt 1 = \vc u_1\otimes \vc u_1 +\vc u_2\otimes \vc u_2+\vc u_3\otimes \vc u_3 $ imply that
$$
\vc u_2^T\mt Q^T \vc z = s_2 = \vc u_2^T\mt Q^T\mt F_1\mt L \Bigl(s_1\vc u_1 + s_2\vc u_2 - \frac{\vc u_3}{\alpha_3}(\alpha_1 s_1+\alpha_2s_2)\Bigr)  ,
$$
where $(s_1,s_2)$ are coordinates on $\Gamma_1^L$, that is
\beq
\label{possibilis}
(s_1,s_2)\in \bigl\{(t_1,t_2)\in\R^2\colon t_1=\vc u_1\cdot\vc x,\, t_2=\vc u_2\cdot\vc x, \vc x\in \Gamma_1^L\bigr\}.
\eeq
Therefore, varying $s_1$ and $s_2$ in an open interval we deduce that
\begin{align*}
\vc u_2^T\mt Q^T\mt F_1\mt L \Bigl(\vc u_1 - \vc u_3\frac{\alpha_1}{\alpha_3}\Bigr) = 0,\\
\vc u_2^T\mt Q^T\mt F_1\mt L \Bigl(\vc u_2 - \vc u_3\frac{\alpha_2}{\alpha_3}\Bigr) = 1.
\end{align*}
There exists hence $\lambda\in\R$ such that
$$
\bigl(\mt L^T\mt F_1^T\mt Q-\mt 1\bigr) \vc u_2 = -\frac{\lambda}{\alpha_3} \bigl(\alpha_3\vc u_1-\alpha_1\vc u_3\bigr)\times \bigl(\alpha_2\vc u_3-\alpha_3\vc u_2\bigr) = \lambda \mt L^T\vc n_1,
$$
that is
\beq
\label{sottovedremo}
\mt Q\vc u_2 = \mt F_1^{-T}\mt L^{-T}\bigl (\vc u_2 + \lambda \mt L^T\vc n_1\bigr).
\eeq
Taking the norm on both sides, we deduce that $\lambda$ must satisfy
\beq
\label{eqlambda1}
1 = |\mt F_1^{-T}\mt L^{-T}\vc u_2|^2+ \lambda^2|\mt F_1^{-T}\vc n_1|^2 + 2\lambda\bigl(\mt L^{-1}\mt F_1^{-1}\mt F_1^{-T}\mt L^{-T}\vc u_2\bigr)\cdot \mt L^T\vc n_1.
\eeq
We notice that $\mt F_1\in K^{qc}$ implies that $\mt F_1\mt L \in K_L$ and hence $\mt L^T\mt F_1^T\mt F_1\mt L\vc u_2 = \vc u_2$. This yields 
$$ 
\bigl (\mt L^T\mt F_1^T\mt F_1\mt L\bigr)^{-1}\vc u_2=\mt L^{-1}\mt F_1^{-1}\mt F_1^{-T}\mt L^{-T}\vc u_2 = \vc u_2.
$$
Therefore, $\mt F_1^{-T}\mt L^{-T}\vc u_2\cdot\mt F_1^{-T}\mt L^{-T}\vc u_2=1$ and \eqref{eqlambda1} simplifies to 
$$
0 = \lambda^2|\mt F_1^{-T}\vc n_1|^2 + 2\alpha_2\lambda,
$$
that is $\lambda=0$ or $\lambda = -\frac{2\alpha_2}{|\mt F_1^{-T}\vc n_1|^2}$.  In the same way, we can show that
\beq
\label{sottovedremo2}
\mt Q\vc u_2 = \mt F_2^{-T}\mt L^{-T}\bigl (\vc u_2 + \mu \mt L^T\vc n_2\bigr),
\eeq
with $\mu=0$ or $\mu = -\frac{2\beta_2}{|\mt F_2^{-T}\vc n_2|^2}$. We now claim that, even if $\alpha_2,\beta_2 \neq 0$, the only possible solution is $\lambda = \mu = 0$. Indeed, let $\alpha_2\neq 0$ (the case $\beta_2\neq 0$ can be treated similarly), and let us notice that
\begin{align*}
z_1(s_1,s_3) = \vc u_1\mt Q^T \mt F_1\mt L \Bigl(s_1\vc u_1 + s_3\vc u_3 - \frac{\vc u_2}{\alpha_2}(\alpha_1 s_1+\alpha_3s_3)\Bigr) ,\\ 
z_3(s_1,s_3) = \vc u_3\mt Q^T \mt F_1\mt L \Bigl(s_1\vc u_1 + s_3\vc u_3 - \frac{\vc u_2}{\alpha_2}(\alpha_1 s_1+\alpha_3s_3)\Bigr) ,
\end{align*}
for every $s_1,s_3$ as in \eqref{possibilis}. As a consequence, $z_1,z_3$ are linear on the boundary, and hence are linear on the set
$$
\Omega_1:=\Bigl\{\vc x \in \Omega_L\colon \vc x\cdot \mt L^T\vc n_1\leq 0,\, \frac{((\mt L\vc n_1\times \mt L\vc n_2)\times \vc u_2)\cdot \vc x }{\sign \alpha_2}\leq 0	\Bigr\}.
$$
This is the subset of $\Omega_L$ where the boundary condition is propagated along the characteristic lines in direction $\vc u_2$. 
Therefore, given \eqref{defKqcT}, we deduce the existence of $\mt G\in K_L$ such that $\vc z(\vc x)=\mt G\vc x$ in $\Omega_1$. A version of the Hadamard jump condition (see e.g., \cite[Prop. 1]{BallJames1}) yields 
\beq
\label{rankoneFG}
\mt G -\mt F_1\mt L = \vc c\otimes \mt L^T\vc n_1,
\eeq
for some $\vc c\in\R^3$. The fact that $\mt G\in K_L$ together with \eqref{plainstrain} imply
$$
\mt Q^T \mt G\vc u_2 =  \vc u_2.
$$
Exploiting \eqref{sottovedremo} and \eqref{rankoneFG} we deduce
\beq
\label{intermedioLambda0}
\mt F_1^{-T}\mt L^{-T}(\vc u_2 + \lambda \mt L^T\vc n_1) = \mt F_1\mt L\vc u_2 + \alpha_2\vc c .
\eeq
Now, polar decomposition implies $\mt F_1\mt L = \mt R_1\mt V_1$, for some $\mt R_1\in SO(3)$, $\mt V_1\in\R^{3\times 3}_{Sym^+}.$ As $\mt F_1\mt L\in K_L$ we also have $\mt V_1\vc u_2 = \vc u_2$ and $\mt V_1^{-1}\vc u_2 = \vc u_2$, as well as $(\mt F_1\mt L)^{-T}\vc u_2 = \mt R_1\vc u_2.$ Thus, \eqref{intermedioLambda0} becomes
\beq
\label{defc}
\vc c = \frac{\lambda}{\alpha_2}  \mt F_1^{-T}\vc n_1.
\eeq
At the same time, the fact that $\mt G,\mt F_1\mt L \in K_L$ implies that $\det \mt G = \det (\mt F_1\mt L) = 1$. But \eqref{rankoneFG} entails,
$$
\det \mt G = \det (\mt F_1\mt L)(1 + \mt L^{-1}\mt F_1^{-1}\vc c\cdot \mt L^T\vc n_1 )=  \det( \mt F_1\mt L) \Bigl(1 + \frac{\lambda}{\alpha_2}|\mt F_1^{-T}\vc n_1|^2 \Bigr),
$$
which implies that $\lambda=0$. The same argument can be applied to prove $\mu =0.$ Therefore, \eqref{sottovedremo} and \eqref{sottovedremo2} simplify to
$$
\mt Q\vc u_2 = \mt F_1^{-T}\mt L^{-T}\vc u_2 = \mt R_1 \vc u_2 = \mt F_1\mt L\vc u_2,\quad\text{ and }\quad \mt Q\vc u_2 = \mt F_2^{-T}\mt L^{-T}\vc u_2 = \mt R_2 \vc u_2 = \mt F_2\mt L\vc u_2
$$
from which we deduce 
\beq
\label{ranksecondSt}
\bigl(\mt F_1-\mt F_2\bigr) \mt L\vc u_2 = 0.
\eeq
Here $\mt R_2\in SO(3)$ is given by the polar decomposition of $\mt F_2\mt L$, and is such that $\mt F_2\mt L = \mt R_2\mt V_2$ for some $\mt V_2\in\R^{3\times3}_{Sym^+}.$ Now, as $\vc u_*\parallel \mt L\vc u_2,$ the hypothesis that $\vc u_*\times \vc n_\perp\neq 0$ implies that $\vc u_2$ and $\vc n_\perp$ are linearly independent. As a consequence, \eqref{rk1firststep} and \eqref{ranksecondSt} imply 
$$
\rank (\mt F_1-\mt F_2)\leq 1,
$$
and \eqref{conclusion}.\\
\textit{Sufficiency.} Let us define 
\[
\vc z(\vc x) =
\begin{cases}
\mt F_1\mt L\vc x,&\qquad\text{in $\Omega_1$,}\\
\mt F_2\mt L\vc x,&\qquad\text{in $\Omega\setminus\Omega_1$.}
\end{cases}
\]
It is easy to check that $\vc z$ satisfies \eqref{defKqcT}--\eqref{bcz}, proving the statement.
\end{proof}
\begin{remark}
\rm
Let $\mt F_1,\mt F_2$ be the deformation gradients measured experimentally in \Tn\,(see \cite{Inamura} or Section \ref{Appl} below) or in Ni\textsubscript{65}Al\textsubscript{35} \cite{BallS1,BallS2}. By \eqref{IncompII} we have $\mt F_1=\mt 1+\vc b_1\otimes \vc m_1$, $\mt F_2=\mt 1+\vc b_2\otimes \vc m_2$ for some $\vc b_1,\vc b_2\in\R^3$ and $\vc m_1,\vc m_2\in\mathbb{S}^2$ such that $\rank (\mt F_1-\mt F_2) = 2$. Taking $\vc n_1=\vc m_1$ and $\vc n_2=\vc m_2$ we have that $\vc u_*\times\vc n_\perp \neq 0$ is verified, and therefore Theorem \ref{Incompatible wells} implies that no stress-free junction involving just two martensitic variants can be observed in \Tn, nor in Ni\textsubscript{65}Al\textsubscript{35} between the nucleated plates $\mt F_1,\mt F_2$.
\end{remark}
\begin{remark}
\rm
The result is independent of the shape of $\partial\Omega\setminus (\Gamma_1\cup\Gamma_2).$
\end{remark}
\begin{remark}
By \cite[Prop. 12]{JamesHyst}, \eqref{ehatholds} is equivalent to the existence of $\mt R\in SO(3)$, $\vc b,\vc m\in\R^3$ satisfying \eqref{Twinning}. If \eqref{ehatholds} fails, then, under some further physically relevant restrictions on the parameters of $\mt U_1,\mt U_2$, \cite{DKMS} implies that $K=K^{qc}$, and that $\vc y$ is affine. 
\end{remark}
\begin{remark}
\label{RemarkPiccolo}
\rm
A similar result holds if we replace $\Omega$ with 
$$
\Omega_C := \bigl\{\vc x\in\R^3\colon 
\vc x\cdot\vc n_\perp \in (0,1)\text{ and }|\vc x-\vc n_\perp(\vc n_\perp\cdot \vc x)|<R \bigr\}\setminus \overline{\Omega},
$$
for which we refer to Figure \ref{DomainSmall}. In this case, however, necessary and sufficient conditions are \eqref{conclusion} and, if $\vc d\neq\vc 0,$ 
$$
\bigr(\vc u_* \cdot \vc n_1\bigl) \bigr( \vc u_* \cdot \vc n_2\bigl) \geq 0.
$$
This latter condition is to guarantee that the information carried by the characteristic lines in direction $\vc u_*$ from the boundary conditions do not overlap. 
\end{remark}

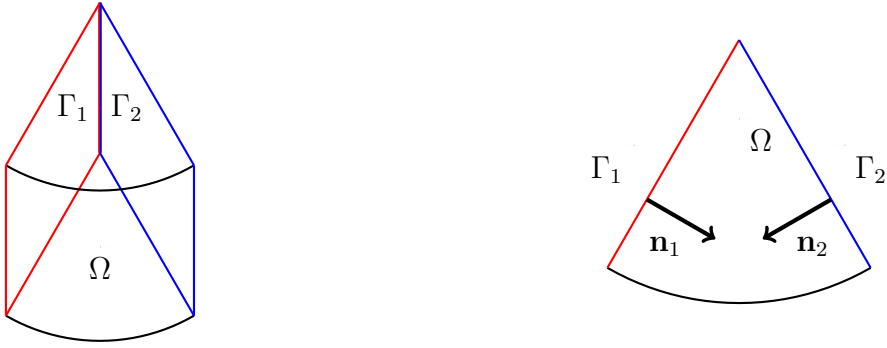
\begin{figure}
\centering
\begin{subfigure}{.45\textwidth}
\begin{tikzpicture}[scale=0.50]
   \draw [step=0.05, black,thick,domain=240:300] plot ({5*cos(\x)}, {5*sin(\x)});
   \draw[red, thick] ({5*cos(240)},{5*sin(240)}) -- (0,0);
   \draw[blue, thick] ({5*cos(300)},{5*sin(300)}) -- (0,0);
   \draw [step=0.05, black,thick,domain=240:300] plot ({5*cos(\x)}, {5*sin(\x)+4});
   \draw[red, thick] ({5*cos(240)},{5*sin(240)+4}) -- (0,4);
   \draw[blue, thick] ({5*cos(300)},{5*sin(300)+4}) -- (0,4);
   \draw[blue, thick] ({5*cos(300)},{5*sin(300)+4}) -- ({5*cos(300)},{5*sin(300)});
   \draw[red, thick] ({5*cos(240)},{5*sin(240)}) -- ({5*cos(240)},{5*sin(240)+4});
   \draw[red, thick] (-0.02,0) -- (-0.02,4);
   \draw[blue, thick] (0.02,0) -- (0.02,4);
\filldraw [red] ({2.5*cos(90)},{2.5*sin(90)-5}) circle (0pt) node[blue,anchor=north,black] {$\Omega$};
\filldraw [red] ({2.5*cos(90)},{2.5*sin(90)-1.3}) circle (0pt) node[red,anchor=east,black] {$\Gamma_1$};
\filldraw [red] ({2.5*cos(90)},{2.5*sin(90)-1.3}) circle (0pt) node[blue,anchor=west,black] {$\Gamma_2 $};
\end{tikzpicture}
\end{subfigure}%
\begin{subfigure}{.45\textwidth}
\begin{tikzpicture}[scale=0.70]
   \draw [step=0.05, black,thick,domain=240:300] plot ({5*cos(\x)}, {5*sin(\x)});
   \draw[red, thick] ({5*cos(240)},{5*sin(240)}) -- (0,0);
   \draw[blue, thick] ({5*cos(300)},{5*sin(300)}) -- (0,0);
  \draw [->,ultra thick] ({3.5*cos(300)},{3.5*sin(300)}) -- (-{-1.5*sin(300)+3.5*cos(300)},-{1.5*cos(300)+3.5*sin(300)}); 
  \draw [->,ultra thick] ({3.5*cos(240)},{3.5*sin(240)}) -- (-{1.5*sin(240)+3.5*cos(240)},-{-1.5*cos(240)+3.5*sin(240)}); 
\filldraw [red] (-{1*sin(240)+3.5*cos(240)},-{-1*cos(240)+3.5*sin(240)}) circle (0pt) node[anchor=north east,black] {$\vc n_1 $};
\filldraw [red] (-{-1*sin(300)+3.5*cos(300)},-{1*cos(300)+3.5*sin(300)}) circle (0pt) node[anchor=north west,black] {$\vc n_2 $};
\filldraw [red] ({2*sin(270)+2.5*cos(270)},{-2*cos(270)+2.5*sin(270)+0.5}) circle (0pt) node[red,anchor=north east,black] {$\Gamma_1$};
\filldraw [red] ({-2*sin(270)+2.5*cos(270)},{2*cos(270)+2.5*sin(270)+0.5}) circle (0pt) node[blue,anchor=north west,black] {$\Gamma_2 $};
\filldraw [red] ({-1.5*cos(90)},{-1.5*sin(90)}) circle (0pt) node[blue,anchor=north west,black] {$\Omega$};
\end{tikzpicture}
\end{subfigure}%
\caption{\label{DomainSmall}Representation of the domain considered in Remark \ref{RemarkPiccolo}. {This domain reflects the formation of incompatible junctions as in Figure \ref{nucleazione BS}.}}
\end{figure}

\begin{remark}
\label{RemarkFail}
\rm
In general, the statement of Theorem \ref{Incompatible wells} does not hold when $\vc u_*\times \vc n_\perp = \vc 0$. Consider for example 
$$
\mt U_1 = \diag (\eta_1,\eta_2, \eta_3),\qquad\mt U_2 = \diag (\eta_2,\eta_1, \eta_3),
$$
for some $\eta_1,\eta_2>0.$ These deformation gradients describe in a suitable basis an orthorhombic to monoclinic transformation. Let further $\mt F_1=\mt U_1,$ $\mt F_2=\mt U_2$,
$$ 
\vc e_1 := [1 0 0]^T,\qquad\vc e_2 := [0 1 0]^T,\qquad\vc e_3 := [0 0 1]^T,
$$
and
\begin{align*}
\vc b_1 = \frac{\sqrt2(\eta_1-\eta_2)}{\eta_1+\eta_2}(-\eta_1\vc e_1 +\eta_2\vc e_2),&\qquad \vc b_2 = \frac{\sqrt{\eta_1^2+\eta_2^2}(\eta_1-\eta_2)}{\eta_1+\eta_2}(\vc e_1 + \vc e_2),\\
\vc m_1 = \frac{1}{\sqrt{2}}(\vc e_1 +\vc e_2),&\qquad \vc m_2 = \frac{1}{\sqrt{\eta_1^2+\eta_2^2}}(\eta_2\vc e_1 -\eta_1\vc e_2).
\end{align*}
We choose $\vc n_1,\vc n_2 \in \mathbb{S}^2$ such that 
$$
\vc n_1 \cdot \vc e_3 = \vc n_2\cdot\vc e_3 = 0,\qquad (\vc e_2-\vc e_1)\cdot \vc n_1 \leq 0,\qquad (\eta_2\vc e_1+\eta_1\vc e_2)\cdot \vc n_2 \leq 0   ,
$$
so that the situation becomes fully two-dimensional (cf. Figure \ref{Remark36}). Indeed, $\vc u_* = \vc n_\perp = \vc e_3.$ Then, we can construct $\vc y\in W^{1,\infty}(\Omega;\R^3)$ as 
\[
\vc y(\vc x) = 
\begin{cases}
\mt F_1 \vc x,\qquad&\text{if $\vc x\cdot \vc m_1\leq 0,$}\\
\bigl(\mt F_1 + \vc b_1\otimes \vc m_1 \bigr)\vc x,\qquad&\text{if $0<\vc x\cdot \vc m_1,\,0<\vc x\cdot \vc m_2,$}\\
\mt F_2 \vc x,\qquad&\text{if $\vc x\cdot \vc m_1\leq 0,$}
\end{cases}
\]
where continuity is guaranteed by the fact that $ \mt F_1 + \vc b_1\otimes\vc m_1-\mt F_2 = \vc b_2\otimes\vc m_2.$ In this case, following \cite{dolzman}, $\nabla\vc y \in K^{qc}$ if and only if $\mt B :=\mt F_1 + \vc b_1\otimes \vc m_1$ satisfies
$$
\det \mt B = \det \mt U_1,\qquad |\mt B(\vc e_1\pm \vc e_2)|^2\leq \eta_1^2+\eta_2^2.
$$
It can be checked that both the first and the second property are satisfied for every $\eta_1,\eta_2>0.$ Therefore, if $\vc u_* \times \vc n_\perp = \vc 0,$ \eqref{conclusion} can fail. We remark that, in this situation, the key ingredient is not the type of transformation (represented here by its stretch tensors $\mt U_1,\mt U_2$), but  the two-dimensional structure of the problem. Indeed, in this case, both the boundary conditions imposed on $\Gamma_1,\Gamma_2$ (which in direction $\vc e_3$ are both a constant elongation/contraction of magnitude $\eta_3$) and the domain (whose shape does not depend on the $\vc e_3$ axis) make the problem essentially two-dimensional.	
\end{remark}
\begin{figure}
\centering
\begin{tikzpicture}[scale=0.50]
   \draw [step=0.05, black,thick,domain=0:240] plot ({5*cos(\x)}, {5*sin(\x)});
   \draw [step=0.05,black,thick,domain=300:360] plot ({5*cos(\x)}, {5*sin(\x)});
   \draw[red, thick] ({5*cos(240)},{5*sin(240)}) -- (0,0);
   \draw[blue, thick] ({5*cos(300)},{5*sin(300)}) -- (0,0);
  \draw [->,ultra thick] ({3.5*cos(300)},{3.5*sin(300)}) -- ({1.5*sin(300)+3.5*cos(300)},{-1.5*cos(300)+3.5*sin(300)}); 
  \draw [->,ultra thick] ({3.5*cos(240)},{3.5*sin(240)}) -- ({-1.5*sin(240)+3.5*cos(240)},{1.5*cos(240)+3.5*sin(240)}); 
\filldraw [red] ({-1*sin(240)+3.5*cos(240)},{1*cos(240)+3.5*sin(240)}) circle (0pt) node[anchor=north east,black] {$\vc n_1 $};
\filldraw [red] ({1*sin(300)+3.5*cos(300)},{-1*cos(300)+3.5*sin(300)}) circle (0pt) node[anchor=north west,black] {$\vc n_2 $};
\filldraw [red] ({0.5*sin(240)+2.5*cos(240)},{-0.5*cos(240)+2.5*sin(240)}) circle (0pt) node[red,anchor=north east,black] {$\Gamma_1$};
\filldraw [red] ({-0.5*sin(300)+2.5*cos(300)},{0.5*cos(300)+2.5*sin(300)}) circle (0pt) node[blue,anchor=north west,black] {$\Gamma_2 $};
\filldraw [red] ({2.5*cos(90)+6},{2.5*sin(90)}) circle (0pt) node[blue,anchor=north west,black] {$\Omega$};

   \draw[green, thick] ({5*cos(135)},{5*sin(135)}) -- (0,0);
     \draw [->,ultra thick] ({3.5*cos(135)},{3.5*sin(135)}) -- ({1.5*sin(135)+3.5*cos(135)},{-1.5*cos(135)+3.5*sin(135)}); 
\filldraw [red] ({1*sin(135)+3.5*cos(135)},{-1*cos(135)+3.5*sin(135)}) circle (0pt) node[anchor=north west,black] {$\vc m_1 $};

   \draw[yellow, thick] ({5*cos(30)},{5*sin(30)}) -- (0,0);
  \draw [->,ultra thick] ({3.5*cos(30)},{3.5*sin(30)}) -- ({-1.5*sin(30)+3.5*cos(30)},{1.5*cos(30)+3.5*sin(30)}); 
\filldraw [red] ({-1*sin(30)+3.5*cos(30)},{1*cos(30)+3.5*sin(30)}) circle (0pt) node[anchor=north east,black] {$\vc m_2$};
\end{tikzpicture}
\caption{\label{Remark36} %
Reduction to a two dimensional situation where Theorem \ref{Incompatible wells} fails, as shown in Remark \ref{RemarkFail}.
}
\end{figure}
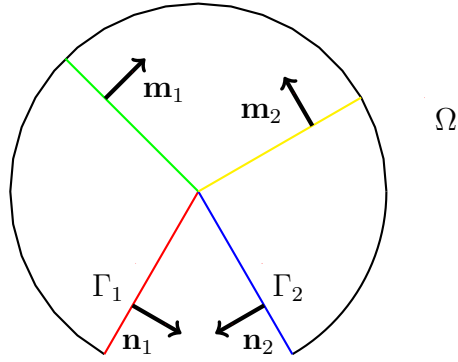

\section{Plastic junctions}
\label{ShearComp}
In this section we want to investigate when, given two matrices $\mt F_1,\mt F_2\in\R^{3\times 3},$ with $\rank (\mt F_1-\mt F_2)=2$, there exist two simple shears $\mt S_i = \mt 1 + s_i\vcg\phi_i\otimes\vcg\psi_i$, $\vcg\phi_i\otimes\vcg\psi_i\in\mathcal{S}$, $i=1,2$, such that $\rank(\mt F_1\mt S_1-\mt F_2\mt S_2)\leq 1.$ {These results are useful for the mathematical characterisation of $V_{II}$ junctions given in the next section. Here and below}, we denote by $\mathcal{S}$ the set of admissible slip systems (or a suitable subset of it), and by $\mathcal{M}$ the set of martensitic variants $\bigcup_{i=1}^N\mt U_i$ (or a suitable subset of it). 

Under our hypotheses on $\mt F_1,\mt F_2,$ there exist $\vc b_1,\vc b_2\in\R^3$ and $\vc m_1,\vc m_2\in\mathbb{S}^2$ such that 
$$
\mt F_2 = \mt F_1 + \vc b_1\otimes \vc m_1+ \vc b_2\otimes \vc m_2.
$$
Therefore, our problem becomes equivalent to finding $\vcg\phi_1\otimes\vcg\psi_1,\,\vcg\phi_2\otimes\vcg\psi_2\in\mathcal{S}$ and $s_1,s_2\in\R$ such that
\beq
\label{rango1eqdasolve}
\rank\bigl(s_1\mt F_1\vcg\phi_1\otimes\vcg\psi_1 - \vc b_1\otimes \vc m_1 - \vc b_2\otimes \vc m_2	- s_2\mt F_2\vcg\phi_2\otimes\vcg\psi_2	\bigr) \leq 1.
\eeq
Lemma \ref{NecessityCond} below gives necessary conditions for the existence of solutions to \eqref{rango1eqdasolve}. There and throughout this section, $\hat{\vcg\phi}_i$ can be interpreted as $\mt F_i\vcg\phi_i$.
\begin{lemma}
\label{NecessityCond}
Let $\vc a_1,\vc a_2,\hat{\vcg\phi}_1,\hat{\vcg\phi}_2,\vc n_1,\vc n_2,\vcg \psi_1,\vcg \psi_2\in\R^3$ and $\rank\bigl(	\vc a_1\otimes \vc n_1 - \vc a_2\otimes \vc n_2 \bigr)=2$. Then, a necessary condition for the existence of $s_1,s_2\in\R$ such that
\beq
\label{rankEq}
\rank\bigl(	\vc a_1\otimes \vc n_1 - \vc a_2\otimes \vc n_2 + s_1\hat{\vcg\phi}_1\otimes \vcg \psi_1 - s_2\hat{\vcg\phi}_2\otimes \vcg \psi_2\bigr) \leq 1
\eeq
is that at least one of the following four conditions hold:
\begin{align*}
\hat{\vcg\phi}_1\cdot(\vc a_1\times\vc a_2)=\hat{\vcg\phi}_2\cdot(\vc a_1\times\vc a_2)=0,\qquad \hat{\vcg\phi}_1\cdot(\vc a_1\times\vc a_2)=\vcg \psi_1\cdot(\vc n_1\times\vc n_2)=0,
\\
\hat{\vcg\phi}_2\cdot(\vc a_1\times\vc a_2)=\vcg \psi_2\cdot(\vc n_1\times\vc n_2)=0,\qquad \vcg \psi_1\cdot(\vc n_1\times\vc n_2)=\vcg \psi_2\cdot(\vc n_1\times\vc n_2)=0.
\end{align*}
\end{lemma}
\begin{proof}
Since $\cof(\mt F) = 0$ if and only if $\rank(\mt F)\leq 1$, \eqref{rankEq} is equivalent to
\beq
\label{rankCof}
\begin{split}
\mt 0 = -&(\vc a_1\times\vc a_2)\otimes (\vc n_1\times \vc n_2) +s_1 (\vc a_1\times\hat{\vcg\phi}_1)\otimes (\vc n_1\times \vcg \psi_1) 
\\
&- s_2 (a_1\times\hat{\vcg\phi}_2)\otimes (\vc n_1\times \vcg \psi_2)
-s_1 (\vc a_2\times\hat{\vcg\phi}_1)\otimes (\vc n_2\times \vcg \psi_1)
\\
&+s_2 (\vc a_2\times\hat{\vcg\phi}_2)\otimes (\vc n_2\times \vcg \psi_2)
-s_1s_2 (\hat{\vcg\phi}_1\times\hat{\vcg\phi}_2)\otimes (\vcg \psi_1\times \vcg \psi_2).
\end{split}
\eeq
Taking now the scalar product of \eqref{rankCof} with $\hat{\vcg\phi}_1 \otimes \vcg\psi_2$ and $\hat{\vcg\phi}_2\otimes\vcg\psi_1$ we respectively obtain
\beq
\label{equazioneZero}
\bigl[(\vc a_1\times\vc a_2) \cdot \hat{\vcg\phi}_1 \bigr]\bigl[(\vc n_1\times\vc n_2) \cdot \vcg \psi_2 \bigr] = 0,\qquad
\bigl[(\vc a_1\times\vc a_2) \cdot \hat{\vcg\phi}_2 \bigr]\bigl[(\vc n_1\times\vc n_2) \cdot \vcg \psi_1 \bigr] = 0.
\eeq
Recalling that $\rank\bigl(	\vc a_1\otimes \vc n_1 - \vc a_2\otimes \vc n_2 \bigr)=2$ implies that $\vc a_1\times\vc a_2\neq\vc 0$ and $\vc n_1\times\vc n_2\neq\vc 0$, from \eqref{equazioneZero} we deduce the claim.
\end{proof}

In general, the  necessary conditions provided by Lemma \ref{NecessityCond} are not sufficient. In other cases, infinitely many solutions $s_1,s_2$ may exist given two slip systems ${\vcg\phi}_1\otimes\vcg\psi_1,{\vcg\phi}_2\otimes\vcg\psi_2\in\mathcal{S}.$ In Proposition \ref{PropMeet} we prove that, under certain hypotheses on the shear systems which are relevant in the following section, 
there exists a unique couple $(s_1,s_2)$ such that \eqref{rankEq} is satisfied.
{
\begin{proposition}
\label{PropMeet}
Let $\vc a_1,\vc a_2,\hat{\vcg\phi}_1,\hat{\vcg\phi}_2,\vc n_1,\vc n_2,\vcg \psi_1,\vcg \psi_2\in\R^3$. Suppose further that $\rank\bigl(	\vc a_1\otimes \vc n_1 - \vc a_2\otimes \vc n_2 \bigr)=2$. Then, 
\begin{itemize}
\item 
if $\vcg\psi_1 = \alpha_1\vc n_1+\alpha_2\vc n_2$, $\vcg\psi_2 = \beta_1\vc n_1+\beta_2\vc n_2$ for some $\alpha_1,\alpha_2,\beta_1,\beta_2\in\R$, and if 
one out of $(\vc a_1\times\vc a_2)\cdot \hat{\vcg\phi}_2\neq 0,$ $(\vc a_1\times\vc a_2)\cdot \hat{\vcg\phi}_1\neq 0$ holds, then $s_1,s_2\in\R$ are such that \eqref{rankEq} is satisfied if and only if they satisfy
\beq
\label{cond s1 a}
\begin{split}
(\vc a_1\times\vc a_2)\cdot \hat{\vcg\phi}_2 &= s_1 (\alpha_2\vc a_1 + \alpha_1 \vc a_2)\cdot (\hat{\vcg\phi}_1\times\hat{\vcg\phi}_2),\qquad \\
(\vc a_1\times\vc a_2)\cdot \hat{\vcg\phi}_1 &= s_2 (\beta_2\vc a_1 + \beta_1 \vc a_2)\cdot (\hat{\vcg\phi}_1\times\hat{\vcg\phi}_2);
\end{split}
\eeq
\item 
if $\hat{\vcg\phi}_1 = \gamma_1\vc a_1+\gamma_2\vc a_2$, $\hat{\vcg\phi}_2 = \delta_1\vc a_1+\delta_2\vc a_2$ for some $\gamma_1,\gamma_2,\delta_1,\delta_2\in\R$, and if 
one out of $(\vc n_1\times\vc n_2)\cdot \hat{\vcg\psi}_2\neq 0,$ $(\vc n_1\times\vc n_2)\cdot \hat{\vcg\psi}_1\neq 0$ holds, 
then $s_1,s_2\in\R$ are such that \eqref{rankEq} is satisfied if and only if they satisfy
\beq
\label{cond s1 b}
\begin{split}
(\vc n_1\times\vc n_2)\cdot \vcg \psi_2 &= s_1 (\gamma_2\vc n_1 + \gamma_1 \vc n_2)\cdot (\vcg \psi_1\times\vcg \psi_2),\\
\qquad (\vc n_1\times\vc n_2)\cdot\vcg  \psi_1 &= s_2 (\delta_2\vc n_1 + \delta_1 \vc n_2)\cdot (\vcg \psi_1\times\vcg \psi_2).
\end{split}
\eeq
\item if $\hat{\vcg\phi}_1 = \gamma_1\vc a_1+\gamma_2\vc a_2$, $\hat{\vcg\phi}_2 = \delta_1\vc a_1+\delta_2\vc a_2$ and $\vcg\psi_1 = \alpha_1\vc n_1+\alpha_2\vc n_2$, $\vcg\psi_2 = \beta_1\vc n_1+\beta_2\vc n_2$ for some $\alpha_1,\alpha_2,\beta_1,\beta_2,\gamma_1,\gamma_2,\delta_1,\delta_2\in\R$, then $s_1,s_2\in\R$ are such that \eqref{rankEq} is satisfied if and only if they satisfy
\beq
\label{cond s1 c}
1= s_1(\alpha_2\gamma_2-\alpha_1\gamma_1) - s_2(\beta_2\delta_2-\beta_1\delta_1) - s_1s_2(\alpha_1\beta_2-\alpha_2\beta_1)(\gamma_1\delta_2-\gamma_2\delta_1).
\eeq
In particular, there may be a one parameter family of solutions.
\end{itemize}
\end{proposition}
\begin{proof}
We just prove the first case, as the second case can be proved in a similar way, and the third is a direct consequence of \eqref{rankN} below. Assuming $\vcg\psi_1 = \alpha_1\vc n_1+\alpha_2\vc n_2$ and $\vcg\psi_2 = \beta_1\vc n_1+\beta_2\vc n_2$, solving \eqref{rankCof} is equivalent to solving
\beq
\label{rankN}
\vc 0 = -\vc a_1\times\vc a_2+s_1(\alpha_2\vc a_1+\alpha_1\vc a_2)\times\hat{\vcg\phi}_1 -
s_2(\beta_2\vc a_1+\beta_1\vc a_2)\times\hat{\vcg\phi}_2 
- s_1s_2(\alpha_1\beta_2-\alpha_2\beta_1)\hat{\vcg\phi}_1\times\hat{\vcg\phi}_2.
\eeq
By testing this equation by $\hat{\vcg\phi}_1$ and $\hat{\vcg\phi}_2$ we obtain the necessity of \eqref{cond s1 a}. Now, let us show that, under our assumptions, \eqref{cond s1 a} are also sufficient conditions. In order to do this, it is sufficient to show that, for $s_1,s_2$ as in \eqref{cond s1 a} the equality in \eqref{rankN} tested with $\vcg \rho$, for some $\vcg \rho\in\R^3$ such that $\vcg \rho \cdot (\hat{\vcg\phi}_1\times\hat{\vcg\phi}_2)\neq 0$, holds. Under our assumptions, and assuming \eqref{cond s1 a}, at least one out of $\vc a_1\cdot(\hat{\vcg\phi}_1\times\hat{\vcg\phi}_2) \neq 0$ and $\vc a_2\cdot(\hat{\vcg\phi}_1\times\hat{\vcg\phi}_2) \neq 0$ holds. Suppose without loss of generality the first one, as the other case can be deduced similarly. 
We can thus multiply 
$$
-\vc a_1\times\vc a_2+s_1(\alpha_2\vc a_1+\alpha_1\vc a_2)\times\hat{\vcg\phi}_1 -
s_2(\beta_2\vc a_1+\beta_1\vc a_2)\times\hat{\vcg\phi}_2 
- s_1s_2(\alpha_1\beta_2-\alpha_2\beta_1)\hat{\vcg\phi}_1\times\hat{\vcg\phi}_2
$$
by $\vc a_1$ 
and deduce that the resulting number is zero, which concludes the proof of the first statement. 
\end{proof}
}

The results above motivate Definition \ref{DefPlJ} below.  
\begin{definition}
\label{DefPlJ}
Let $\mt R_1,\mt R_2\in SO(3)$ and $\mt V_1,\mt V_2\in \mathcal{M}$ such that $\rank(\mt R_1\mt V_1-\mt R_2\mt V_2) = 2.$ Let also $\bar t_1,\bar t_2\in \R\setminus\{0\}$ and $\vcg \phi_{1}\otimes \vcg\psi_1 ,\vcg \phi_{2}\otimes \vcg\psi_2 \in \mathcal{S}$ be such that
$\mt F_i(s):=\mt R_i\mt V_i(\mt 1+s\vcg \phi_{i}\otimes \vcg\psi_i)$ satisfies
$$
\mt F_1(\bar t_1)-\mt F_2(\bar t_2) = \bar{\vc b}\otimes \vc m,
$$
for some $\bar{\vc b}\in\R^3,\vc m\in \mathbb{S}^2$. Then, we say that $\mt F_1$ and $\mt F_2$ form a plastic junction at $(\bar t_1,\bar t_2)$ for $\mt R_1\mt V_1,\mt R_2\mt V_2$. In this case, we call the plane $\{\vc x\in\R^3: \vc x\cdot \vc m=0\}$ the plastic junction plane.

We say that the plastic junction formed by $\mt F_1$ and $\mt F_2$ at $(\bar t_1,\bar t_2)$ is locally rigid if there exists $\delta>0$ such that, for every $\mt R\in SO(3)\setminus\{\mt 1\}$ with $|\mt R-\mt 1|\leq \delta$, and every $t_1,t_2\in \R$ satisfying $|t_1-\bar{t}_1|+|t_2-\bar{t}_2|\leq \delta$, there exists no $\vc b\in\R^3$ such that
\beq
\label{defRk1mfixed}
\mt R \mt F_1(t_1)-\mt F_2(t_2) = {\vc b}\otimes \vc m.
\eeq
\end{definition}

The following result gives sufficient conditions for a plastic junction to be locally rigid. The notation below refers to the notation of Definition \ref{DefPlJ}.
\begin{proposition}
\label{LocalRigidProp}
Let $\mt F_1$ and $\mt F_2$ form a plastic junction at $(\bar t_1,\bar t_2)$ as defined in Definition \ref{DefPlJ}. Let further 
$\vcg\psi_1,\vcg\psi_2 \nparallel \vc m$,  
$\cof(\mt R_1\mt V_1-\mt R_2\mt V_2) = \hat{\vc b}\otimes \hat{\vc m}$ for some $\hat{\vc b}\in\R^3\setminus\{\vc0\}$, $\hat{\vc m}\in\mathbb{S}^2$ such that $\hat{\vc m}\cdot \vc m=\hat{\vc m}\cdot \vcg \psi_1=\hat{\vc m}\cdot \vcg \psi_2=0$, and
\beq
\label{ipotesiR}
\Bigl( {\mt R_1\mt V_1\hat{\vc m}} \times\mt R_1\mt V_1	 \bigl(\vc v + \bar{t}_1\vcg\phi_1 (\vcg\psi_1\cdot\vc v)\bigr)\Bigr) \cdot \Bigl( \mt R_1\mt V_1 \vcg\phi_1\times\mt R_2\mt V_2 \vcg\phi_2
 \Bigr)\neq 0,\qquad\text{where $\vc v := \vc m \times\hat{\vc m}$}.
\eeq
Then the plastic junction formed by $\mt F_1$ and $\mt F_2$ at $(\bar t_1,\bar t_2)$ is locally rigid.
\end{proposition}
\begin{proof}
Let us first notice that \eqref{defRk1mfixed} can be written as
\beq
\label{Rasse}
\mt R \mt R_1\mt V_1(\mt 1+t_1\vcg \phi_{1}\otimes \vcg\psi_1) - (\mt R_1\mt V_1 + \vc b_1\otimes \vc m_1+ \vc b_2\otimes \vc m_2)(\mt 1+t_2\vcg \phi_{2}\otimes \vcg\psi_2)  = \vc b\otimes\vc m,
\eeq
for some $\vc b_1,\vc b_2\in\R^3\setminus\{\vc0\}$, $\vc m_1,\vc m_2\in\mathbb{S}^2$ such that 
$\frac{\vc m_1\times\vc m_2}{|\vc m_1\times\vc m_2|} =  \hat{\vc m}$. Testing \eqref{Rasse} by $\hat{\vc m}$, we deduce that a necessary condition for $\mt R\in SO(3)$ to satisfy \eqref{defRk1mfixed}, is that the rotation axis of $\mt R$ is $\mt R_1\mt V_1\hat{\vc m}$. Furthermore, letting $\vc v := \vc m \times\hat{\vc m}$, a necessary condition for the existence of $\mt R\in SO(3)$ such that \eqref{defRk1mfixed} holds is that
$$
\mt R \mt R_1\mt V_1(\mt 1+t_1\vcg \phi_{1}\otimes \vcg\psi_1)\vc v - (\mt R_1\mt V_1 + \vc b_1\otimes \vc m_1+ \vc b_2\otimes \vc m_2)(\mt 1+t_2\vcg \phi_{2}\otimes \vcg\psi_2)\vc v=\vc 0,
$$
which is \eqref{Rasse} tested by $\vc v$. Let hence 
$\mt R(\theta)\colon[0,2\pi]\to SO(3)$ be the rotation of axis $\mt R_1\mt V_1\hat{\vc m}$ and angle $\theta$. Let us also define the smooth function
$$
\vc f(\theta,t_1,t_2):=\mt R \mt R_1\mt V_1(\mt 1+t_1\vcg \phi_{1}\otimes \vcg\psi_1)\vc v - (\mt R_1\mt V_1 + \vc b_1\otimes \vc m_1+ \vc b_2\otimes \vc m_2)(\mt 1+t_2\vcg \phi_{2}\otimes \vcg\psi_2)\vc v.
$$
Necessary and sufficient condition to have local rigidity is that $\vc f\neq\vc 0$ in a neighbourhood of $(0,\bar{t}_1,\bar{t}_2)$. But 
\begin{align*}
&\frac{\partial}{\partial \theta} \vc f (0,\bar{t}_1,\bar{t}_2) = \frac{\mt R_1\mt V_1\hat{\vc m}}{|\mt R_1\mt V_1\hat{\vc m}|}\times \bigl(\mt R_1\mt V_1 (\vc v + \bar{t}_1\vcg\phi_1 (\vcg\psi_1\cdot\vc v))\bigr),
\\
&\frac{\partial}{\partial t_1} \vc f (0,\bar{t}_1,\bar{t}_2) = (\vcg\psi_1\cdot\vc v) \mt R_1\mt V_1 \vcg\phi_1, 
\qquad 
\frac{\partial}{\partial t_2} \vc f (0,\bar{t}_1,\bar{t}_2) = (\vcg\psi_2\cdot\vc v) \mt R_2\mt V_2 \vcg\phi_2.
\end{align*}
Therefore, if condition \eqref{ipotesiR} is satisfied, $\rank \nabla \vc f (0,\bar{t}_1,\bar{t}_2) = 3$, and hence there exists a neighbourhood of radius $\delta$ of $(0,\bar{t}_1,\bar{t}_2)$ such that for every $\vc w:=(\theta,t_1-\bar{t}_1,t_2-\bar{t}_2)$ with $0<|\vc w|\leq\delta$ 
$$
\vc f(\theta,t_1,t_2) = \nabla \vc f (0,\bar{t}_1,\bar{t}_2)\vc w + o(|\vc w|\delta) \neq \vc 0,
$$
which is the claim.
\end{proof}

\section{Stability of plastic junctions}
\label{Minim}
In this section we give sufficient conditions for plastic junctions to be weak local minimisers of the energy functional $I$. We recall that any Lipschitz continuous map $\vc y$ is a weak local minimiser if there exists $\eps>0$ such that $I(\vcg\rho)\geq I(\vc y)$ for any Lipschitz continuous map $\vcg \rho$ satisfying $\|\vc y-\vcg \rho\|_{W^{1,\infty}_{loc}}\leq \eps$. We start the Section by giving a mathematical definition of $V_{II}$ junctions. Then we state and prove our local stability result in Theorem \ref{StableThm} which gives sufficient conditions for $V_{II}$ junctions to be strict weak local minimisers. At the end of the section we state a stability result for plastic junctions, which relies on the same proof as Theorem \ref{StableThm}.\\
%

The definition of $V_{II}$ junction reads as follows:
\begin{definition}
\label{DefVII}
Let $\mt R_1,\mt R_2\in SO(3)$ and $\mt V_1,\mt V_2\in \mathcal{M}$ be such that $\rank(\mt R_1\mt V_1-\mt R_2\mt V_2) = 2.$ Let also $\bar{\mt F}_1,\bar{\mt F}_2\in\R^{3\times3}$ form a plastic junction at $(\bar{t}_1,\bar{t}_2)$ for $\mt R_1\mt V_1,\mt R_2\mt V_2$ which is locally rigid. Assume further: 
\begin{enumerate}[(1)]
\item\label{HH2} $\bar{\mt F}_1-\bar{\mt F}_2 = \vc b\otimes \vc m$ and $\cof(\mt R_1\mt V_1-\mt R_2\mt V_2) = \hat{\vc b}\otimes \hat{\vc m}$ for some $\vc b,\hat{\vc b}\in\R^3\setminus\{\vc0\}$, $\vc m,\hat{\vc m}\in\mathbb{S}^2$;
\item\label{HH3} (Domain) The domain $\omega$ (cf. Figure \ref{StableFigure}) is defined as $\omega := \{\vc x \in\R^3 : \min\{\vc x\cdot\vc n_1,\vc x\cdot \vc n_2\} < 0 \}$ for some $\vc n_1,\vc n_2\in\mathbb{S}^2$. We also define $\gamma_i:= \{\vc x \in \omega^c : \vc x\cdot\vc n_i= 0 \}$ for $i=1,2;$
\item\label{HH4} (Geometry) $\vc n_1, \vc n_2, \vcg \psi_1, \vcg \psi_2,\vc m \perp \hat{\vc m}.$ Also, (cf. Figure \ref{StableFigure}) there exist $\theta_{\vc m},\theta_{\vcg \psi_1},\theta_{\vcg \psi_2},\theta_{\vc n_2}\in (0,2\pi)$ (or in $(-2\pi,0)$) such that $|\theta_{\vcg \psi_1}|<|\theta_{\vc m}|<|\theta_{\vcg \psi_2}|<|\theta_{\vc n_2}|,$ and 
\begin{align*}
&\mt R_{\hat{\vc m}}(\theta_{\vcg \psi_1}) \gamma_1 \subset\{\vc x\in\R^3\colon\vc x\cdot \vcg \psi_1=0\},\quad \mt R_{\hat{\vc m}}(\theta_{\vc m}) \gamma_1 \subset\{\vc x\in\R^3\colon\vc x\cdot \vc m=0\},\\
&\mt R_{\hat{\vc m}}(\theta_{\vcg \psi_2}) \gamma_1 \subset\{\vc x\in\R^3	\colon\vc x\cdot \vcg \psi_2=0\},\quad \mt R_{\hat{\vc m}}(\theta_{\vc n_2}) \gamma_1= \gamma_2,
\end{align*}
where $\mt R_{\hat{\vc m}}(\theta) \gamma_1$ is the rotation of angle $\theta$ and axis $\hat{\vc m}$ of the half-plane $\gamma_1$. Furthermore, $\mt R_{\hat{\vc m}}(\theta) \gamma_1\subset\omega$ for any $\theta\in(0,\theta_{\vc n_2})$ (resp. $(\theta_{\vc n_2},0)$).
\item \label{HH5}(Structure) $\vc y\in W^{1,\infty}_{loc}(\R^3; \R^3)$ is defined by 
\beq
\label{stabbile}
\vc y(\vc x) = 
\begin{cases}
\bar{\mt F}_1\vc x,\qquad&\text{if $\vc x\in\Omega_1:=\bigl\{ \hat{\vc x}\in\omega\colon \hat{\vc x}\subset\mt R_{\hat{\vc m}}(\theta) \gamma_1,\,\theta\in(\theta_{\vcg \psi_1},\theta_{\vc m})\; \text{(resp. $(\theta_{\vc m},\theta_{\vcg \psi_1})$) }\bigr\}$,}\\
\bar{\mt F}_2\vc x,\qquad&\text{if $\vc x\in\Omega_2:=\bigl\{ \hat{\vc x}\in\omega\colon \hat{\vc x}\subset\mt R_{\hat{\vc m}}(\theta) \gamma_1,\,\theta\in(\theta_{\vc m},\theta_{\vcg \psi_2})\; \text{(resp. $(\theta_{\vcg \psi_2},\theta_{\vc m})$) }\bigr\}$,}\\
\mt R_1\mt V_1\vc x,\qquad&\text{if $\vc x\in\Omega_3:=\bigl\{ \hat{\vc x}\in\omega\colon \hat{\vc x}\subset\mt R_{\hat{\vc m}}(\theta) \gamma_1,\,\theta\in(0,\theta_{\vcg \psi_1})\; \text{(resp. $(\theta_{\vcg \psi_1},0)$) }\bigr\}$,}\\
\mt R_2\mt V_2\vc x,\qquad&\text{if $\vc x\in\Omega_4:=\bigl\{ \hat{\vc x}\in\omega\colon \hat{\vc x}\subset\mt R_{\hat{\vc m}}(\theta) \gamma_1,\,\theta\in(\theta_{\vcg \psi_2},\theta_{\vc n_2})\; \text{(resp. $(\theta_{\vc n_2},\theta_{\vcg \psi_2})$) }\bigr\}$,}\\
\vc x,\qquad&\text{if $\vc x\in\omega^c$.}
\end{cases}
\eeq 
\end{enumerate}
Then, we say that $\vc y$ is a $V_{II}$ junction between $\mt R_1\mt V_1$ and $\mt R_2\mt V_2$.
\end{definition}

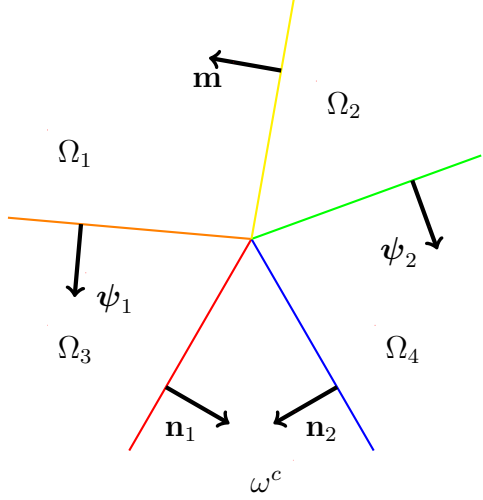
\begin{figure}
\centering
\begin{tikzpicture}[scale=0.65]
   \draw[red, thick] ({5*cos(240)},{5*sin(240)}) -- (0,0);
   \draw[blue, thick] ({5*cos(300)},{5*sin(300)}) -- (0,0);
  \draw [->,ultra thick] ({3.5*cos(300)},{3.5*sin(300)}) -- ({1.5*sin(300)+3.5*cos(300)},{-1.5*cos(300)+3.5*sin(300)}); 
  \draw [->,ultra thick] ({3.5*cos(240)},{3.5*sin(240)}) -- ({-1.5*sin(240)+3.5*cos(240)},{1.5*cos(240)+3.5*sin(240)}); 
\filldraw [red] ({-1*sin(240)+3.5*cos(240)},{1*cos(240)+3.5*sin(240)}) circle (0pt) node[anchor=north east,black] {$\vc n_1 $};
\filldraw [red] ({1*sin(300)+3.5*cos(300)},{-1*cos(300)+3.5*sin(300)}) circle (0pt) node[anchor=north west,black] {$\vc n_2 $};

   \draw[green, thick] ({-5*cos(200)},{-5*sin(200)}) -- ({0*cos(200)},-{0*sin(200)});
   
      \draw[orange, thick] ({5*cos(200-25)},{5*sin(200-25)}) -- ({-0*cos(200)},-{0*sin(200)});
      
     \draw [->,ultra thick] ({3.5*cos(175)},{3.5*sin(175)}) -- ({-1.5*sin(175)+3.5*cos(175)},{1.5*cos(175)+3.5*sin(175)}); 
\filldraw [red] ({1*sin(175)+3.5*cos(175)},{+1*cos(175)+3.5*sin(175)}) circle (0pt) node[anchor= north west,black] {$\vcg \psi_1 $};
     \draw [->,ultra thick] ({3.5*cos(20)},{3.5*sin(20)}) -- ({1.5*sin(20)+3.5*cos(20)},{-1.5*cos(20)+3.5*sin(20)}); 
\filldraw [red] ({1*sin(20)+3.5*cos(20)},{-1*cos(20)+3.5*sin(20)}) circle (0pt) node[anchor= north east,black] {$\vcg \psi_2 $};

   \draw[yellow, thick] ({5*cos(80)},{5*sin(80)}) -- (0,0);
  \draw [->,ultra thick] ({3.5*cos(80)},{3.5*sin(80)}) -- ({-1.5*sin(80)+3.5*cos(80)},{1.5*cos(80)+3.5*sin(80)}); 
\filldraw [red] ({-1*sin(80)+3.5*cos(80)},{1*cos(80)+3.5*sin(80)}) circle (0pt) node[anchor=north east,black] {${\vc m}$};

\filldraw [red] ({1*sin(225)+3.5*cos(225)-1},{-1*cos(225)+3.5*sin(225)}) circle (0pt) node[anchor=north west,black] {$\Omega_3$};

\filldraw [red] ({1*sin(225)+3.5*cos(225)+5.7},{-1*cos(225)+3.5*sin(225)}) circle (0pt) node[anchor=north west,black] {$\Omega_4$};

\filldraw [red] ({1*sin(225)+3.5*cos(225)+4.5},{-1*cos(225)+3.5*sin(225)+5}) circle (0pt) node[anchor=north west,black] {$\Omega_2$};

\filldraw [red] ({1*sin(225)+3.5*cos(225)-1},{-1*cos(225)+3.5*sin(225)+4}) circle (0pt) node[anchor=north west,black] {$\Omega_1$};

\filldraw [red] ({-3*sin(240)+3.5*cos(240)},{3*cos(240)+3.5*sin(240)}) circle (0pt) node[anchor=north east,black] {$\omega^c $};
\end{tikzpicture}
\caption{\label{StableFigure} %
Representation of $\omega$ as defined in Definition \ref{DefVII} \eqref{HH3} projected on the plane orthogonal to $\hat{\vc m}$. Here, $\vcg\psi_1,\vcg\psi_2,\vc m$ and $\Omega_1,\Omega_2,\Omega_3,\Omega_4$ are as in Definition \ref{DefVII} \eqref{HH4}--\eqref{HH5}. We remark that $\overline{\omega} = \overline{\Omega_1\cup\Omega_2\cup\Omega_3\cup\Omega_4}.$
}
\end{figure}
\begin{remark}
\label{CompatCon1}
\rm
The Hadamard jump condition implies that a necessary condition in order to form a $V_{II}$ junction between $\mt R_1\mt V_1$ and $\mt R_2\mt V_2$ is that 
$$
\rank\bigl(\mt R_1\mt V_1-\mt 1\bigr)\leq 1\qquad\text{ and }\qquad\rank\bigl(\mt R_2\mt V_2 -\mt 1\bigr)\leq 1.
$$
\end{remark}
\begin{remark}
\rm
The hypothesis \ref{HH4} requiring that $\vc n_1,\vc n_2,\vcg \psi_1,\vcg\psi_2,\vc m \perp  \hat{\vc m}$ guarantees the continuity of $\vc y$ along the line $s\hat{\vc m}$ for $s\in\R,$ and justifies the bi-dimensional representation of stable plastic junctions given in Figure \ref{StableFigure}.
\end{remark}

\noindent
Before stating our stability result let us introduce the following definition:
\begin{definition}
\label{DefSep}
Let $s\in \R$, $\mt R_{\mt F}\in SO(3)$, $\mt U\in \mathcal{M}$ and $\vcg \phi_{\mt F}\otimes \vcg\psi_{\mt F}\in \mathcal{S}$
. We say that $\mt F = \mt R_{\mt F} \mt U (\mt 1 + s\vc \phi_{\mt F}\otimes \psi_{\mt F})$ enjoys the separation property if there exists $\delta>0$ such that $|\mt F-\mt G|>\delta$ for every $\mt G = \mt R_{\mt G} \mt V (\mt 1 + t\vc \phi_{\mt G}\otimes \psi_{\mt G})$, with $t\in \R$, $\mt R_{\mt G}\in SO(3)$, $\mt V\in \mathcal{M}$, $\vc \phi_{\mt G}\otimes \psi_{\mt G}\in \mathcal{S}$ and where at least one out of $\mt U\neq \mt V$ and $\vc \phi_{\mt F}\otimes \psi_{\mt F}\neq \vc \phi_{\mt G}\otimes \psi_{\mt G}$ holds.  
\end{definition}
\begin{remark}
\label{Rm55}
{\rm If $\mt F$ enjoys the separation property, then in a neighbourhood of $\mt F$ there exists a unique decomposition $\mt F = \mt F^e\mt F^p$ of finite energy.}
\end{remark}
We also introduce the definition of a locally stable $V_{II}$ junction:
\begin{definition}
\label{DefStabY}
We say that a $V_{II}$ junction $\vc y \in W^{1,\infty}_{loc}(\R^3;\R^3)$ is locally stable if there exists $\eps>0$ such that, given any $\vcg \rho \in W^{1,\infty}_{loc}(\R^3;\R^3)$ satisfying \begin{enumerate}[(A)]
\item\label{AA} $\int_{B_r} W(\nabla\vcg\rho)\,\mathrm{d}\vc x <\infty$ for any open ball $B_r$ centred at $\vc 0$ and of arbitrary radius $r>0$,
\item\label{BB} $\| \nabla\vcg\rho - \nabla\vc y\|_{L^{\infty}}\leq \eps$,
\item\label{DD} $\vcg\rho$ is $1-1$,
\end{enumerate}
it holds:
\begin{itemize}
\item[(T1)]\label{TT1} for any measurable $\mathcal{B}\subset\R^3$ bounded
\beq
\label{cheIneq0}
\int_{\mathcal B} \bigl(W(\nabla\vcg\rho)-W(\nabla\vc y) \bigr)\,\mathrm{d}\vc x \geq 0,
\eeq
\item[(T2)]\label{TT2} the equality
\beq
\label{cheIneqPari}
\int_{B_r} \bigl(W(\nabla\vcg\rho)-W(\nabla\vc y) \bigr)\,\mathrm{d}\vc x = 0,
\eeq
holds for any open ball $B_r$ centred at $\vc 0$ and of arbitrary radius $r>0$ if and only if $\vcg \rho = \mt R\vc y +\vc c$ for some $\mt R\in SO(3), \vc c\in\R^3$.
\end{itemize}
%
%
\end{definition}

\begin{remark}
As pointed out in Section \ref{Nonlin}, the energy density $W$ is invariant under rigid motions. That is, given any $\vcg \rho\in W^{1,\infty}_{loc}(\R^3;\R^3)$, any $\mt R\in SO(3)$ and any $\vc c\in\R^3$, we have that $\vcg \rho$ and $\mt R\vcg \rho +\vc c$ have the same energy. As we are not imposing any boundary condition on the variations $\vcg \rho$ in Definition \ref{DefStabY}, any $\vcg \rho =\mt R\vc y+\vc c$, that is a rigid motion of a $V_{II}$ junction $\vc y$, has the same energy as $\vc y$. According to Definition \ref{DefStabY} a locally stable $V_{II}$ junction is a strict weak local minimiser modulo rigid motions. 
\end{remark}

We are now ready to state and prove our stability theorem for $V_{II}$ junctions. The result relies on three main ingredients: first, we assume that any possible small variation $\vcg \rho$ of our $V_{II}$ junction (described by the map $\vc y$) has locally finite energy. This, together with the structure of the energy density $W$ and the separation property (introduced in Definition \ref{DefSep}) give a structure to the gradient of $\vcg \rho$ (cf. Remark \ref{Rm55}). Second, we exploit the result of \cite{BallJamesPlane} characterising plane strains. Indeed, by using this result, we are able to prove that an Hadamard jump condition must hold for $\vcg\rho$ at the plastic junction plane $\{\vc x\cdot\vc m=0\}$ of our plastic junction. Third, we use the local rigidity of the plastic junction to prove that our variation $\vcg \rho$ coincides, up to a rotation, with $\vc y$ in a wedge of $\R^3$ (namely $\Omega_1\cup\Omega_2$). Finally we prove (T1)-(T2). The theorem reads as follows:  
\begin{theorem}
\label{StableThm}
Let $\vc y \in W^{1,\infty}_{loc}(\R^3;\R^3)$ be a $V_{II}$ junction as in Definition \ref{DefVII}. Let also $\bar{\mt F}_1,\bar{\mt F}_2$ enjoy the separation property. Then, if {$(\mt V_i^{2}\vcg\phi_i\times\vcg\psi_i)\cdot\vc m\neq 0,$} for $i=1,2$, the $V_{II}$ junction is locally stable in the sense of Definition \ref{DefStabY}.
%
\end{theorem}

%
%

\begin{remark}
\rm
In Definition \ref{DefVII}, Definition \ref{DefStabY} and hence in the statement of Theorem \ref{StableThm} we consider an unbounded domain. This domain can be interpreted as a blow-up close to  the line given by $\overline{\gamma}_1\cap \overline{\gamma}_2$, where the incompatibility occurs. Mathematically, this choice is motivated by the argument in the proof, which relies on rigidity for plain strains. More precisely, this leads to the fact that the deformation gradient on the plane of compatibility $\{\vc x\cdot\vc m=0\}$ is propagated in $\Omega_1$ along the characteristic lines in direction $(\mt V_1^{2}\vcg\phi_1\times\vcg\psi_1)$, and in $\Omega_2$ along the lines in direction $(\mt V_2^{2}\vcg\phi_2\times\vcg\psi_2)$. A similar theorem could be proved on any connected Lipschitz domain $\Omega$ 
such that for every $\vc x\in \Omega\cap \Omega_i$, $i=1,2$, $\vc x+s(\mt V_i^{2}\vcg\phi_i\times\vcg\psi_i)\in\Omega$ for every $s\in[0,s^*_i]$, where $s_i^*\in\R$ is such that $(\vc x+s_i^*(\mt V_i^{2}\vcg\phi_i\times\vcg\psi_i))\cdot \vc m=0.$ This last condition guarantees that the information is transported by the characteristic lines $(\mt V_i^{2}\vcg\phi_i\times\vcg\psi_i)$ from the plane of compatibility $\{\vc x\cdot \vc m=0\}$ to every point in the domain.
\end{remark}

\begin{proof}
Let $\delta_1,\delta_2>0$ be as in Definition \ref{DefSep} such that $\bar{\mt F}_1,\bar{\mt F}_2$ respectively enjoy the separation property. Let also $\delta_3:=\frac12 \min\bigl\{\|\mt R\mt U-\mt V\| : {\mt U\neq\mt V\in \mathcal{M}\cup\{\mt 1\}, \mt R\in SO(3)}\bigr\}$, and let us take $\eps_0 = \min\bigl\{\delta_1,\delta_2, \delta_3\bigr\}$. 
Consider now any $\vcg \rho\in W^{1,\infty}_{loc}(\R^3;\R^3)$ satisfying \eqref{AA}--\eqref{DD} in Definition \ref{DefStabY}. Then, since the energy is locally finite, by the separation property we have, 
\[
\nabla \vcg \rho(\vc x) = 
\begin{cases}
\nabla\vc z^{(1)},\qquad&\text{in $\Omega_1$,}\\
\nabla\vc z^{(2)},\qquad&\text{if $\Omega_2$,}
\end{cases}
\]
for some locally Lipschitz continuous $\vc z^{{(1)}},\vc z^{{(2)}}$ such that
\beq
\label{ZZZ}
\nabla \vc z^{(1)}(\vc x)=\hat{\mt R}_1(\vc x)\mt V_1 (\mt 1 + t_1(\vc x)\vcg \phi_1 \otimes\vcg \psi_1 ),\qquad \nabla\vc z^{(2)}=\hat{\mt R}_2(\vc x)\mt V_2 (\mt 1 + t_2(\vc x)\vcg \phi_2 \otimes\vcg \psi_2 ),
\eeq
for some measurable $t_i\colon\Omega_i\to\R$, and $\hat{\mt R}_i\colon \Omega_i\to SO(3)$, $i=1,2.$ Define now $\tilde{\vc z}^{(i)}(\vc x) := \vc z^{(i)}(\mt V_i^{-1}\vc x)$. We notice that,
$$
\det \nabla \tilde{\vc z}^{(i)} = 1, \qquad  (\nabla\tilde{\vc z}^{(i)})^T
(\nabla\tilde{\vc z}^{(i)}) = \mt 1 +  t_i(\vc x)\mt V_i\vcg \phi_i \odot \mt V_i^{-1} \vcg \psi_i + t_i^2(\vc x)|\mt V_i\vcg \phi_1|^2\mt V_i^{-1}\vcg \psi_i \otimes\mt V_i^{-1}\vcg \psi_i, 
$$
where $\vc u\odot\vc v = \vc u\otimes \vc v+\vc v\otimes\vc u$ for any $\vc u,\vc v\in\R^3$. It follows then by \cite[Thm. 3.1]{BallJamesPlane} that $\tilde{\vc z}^{(i)}$ is a plain strain, and we can hence deduce the existence of $\mt Q_1,\mt Q_2\in SO(3)$ such that
$$
\tilde{\vc z}^{(i)} = \mt Q_i\bigl(\tilde{z}^{(i)}_1(s_1^{(i)},s_3^{(i)}) \vc u_1^{(i)} + s_2^{(i)}\vc u^{(i)}_2  + \tilde{z}^{(i)}_3(s_1^{(i)},s_3^{(i)}) \vc u^{(i)}_3 \bigr),
$$
for some Lipschitz functions $\tilde{z}_1^{(i)},\tilde{z}_3^{(i)}$, and where
$$
\vc u_1^{(i)}: = \frac{\mt V_i^{-1}\vcg\psi_i}{|\mt V_i^{-1}\vcg\psi_i|},\qquad\vc u_3^{(i)}: = \frac{\mt V_i\vcg\phi_i}{|\mt V_i\vcg\phi_i|},\qquad \vc u_2^{(i)} = \vc u_3^{(i)}\times\vc u_1^{(i)}, \qquad s^{(i)}_j = \vc x\cdot \vc u^{(i)}_j.
$$
Now, given the fact that the $\tilde{\vc z}^{(i)}$ are Lipschitz continuous and that {$(\mt V_i^{2}\vcg\phi_i\times\vcg\psi_i)\cdot\vc m\neq 0,$} (and hence $\vc u_2^{(i)}\cdot \mt V_i^{-1}\vc m\neq 0$) the value of $\nabla\tilde{\vc z}^{(i)}$ is well defined on the plane $\{\vc x\cdot \mt V_i^{-1}\vc m = 0\}$. Indeed, 
\beq
\label{traceZ}
\nabla\tilde{\vc z}^{(i)}(\vc x) = \nabla\tilde{\vc z}^{(i)}(\vc x + r\vc u_2^{(i)})
\eeq
for almost every $\vc x \in \{\vc x\cdot \mt V_i^{-1}\vc m = 0\}$ and almost every $s\in\R$ such that $\vc x + s\vc u_2^{(i)}\in \mt V_i^{-1}\Omega_i$.
%
As a consequence, the value of $\nabla \vc z^{(1)},\nabla \vc z^{(2)}$ on $\{\vc x\cdot \vc m = 0\}$ is well defined, and is respectively in $L^\infty(\gamma_1;\R^{3\times 3})$, $L^\infty(\gamma_2;\R^{3\times 3})$. By the continuity of $\vcg \rho$ and a weak version of the Hadamard jump condition (see \cite[Remark 10]{FDP2}) we deduce that 
\beq
\label{rank1forZ}
\nabla \vc z^{(1)}(\vc x) - \nabla \vc z^{(2)}	(\vc x)
= \hat{\vc b}(\vc x)\otimes \vc m
,\qquad \text{a.e. $\vc x\in\{\hat{\vc x}\in\omega\colon\hat{\vc x}\cdot \vc m = 0\}$},
\eeq
for some measurable $\hat{\vc b}\colon \{\hat{\vc x}\in\omega\colon\hat{\vc x}\cdot \vc m = 0\}\to\R^3$. 

We now claim that this implies the existence of $\mt R_0\in SO(3)$ such that  $\nabla \vc z^{(i)} (\vc x) = \mt R_0 \mt F_i$ a.e. in $\Omega_i$, $i=1,2$. Indeed,
let us consider the smooth functions
$$
f_i(\mt R,t) = |\mt R\mt R_i\mt V_i(\mt 1+t\vcg\phi_i\otimes\vcg\psi_i)-\mt R_i\mt V_i(\mt 1+\bar{t}_i\vcg\phi_i\otimes\vcg\psi_i)|,\qquad i=1,2,
$$
and let $\delta^*$ be as in Definition \ref{DefPlJ}. Since the $f_i$'s are continuous, $f_i\to\infty$ as $|t|\to\infty$ and $f_i=0$ if and only if $\mt R=\mt 1$ and $t=\bar{t}_i$, there exists $\eps_1>0$ such that $f_i\leq \eps_1$ implies $|\mt R-\mt 1| + |t-\bar{t}_i|\leq \frac12\delta^*$ for $i=1,2$. Let us hence fix $\eps :=\min\{\eps_0,\eps_1\}$. 
Therefore, if by \eqref{BB} $|\nabla\vc z^{(i)}-\nabla\vc y|\leq \eps$ a.e. in $\Omega_i$ with $i=1,2$, then by \eqref{ZZZ} $|\hat{\mt R}_1^T(\vc x)\hat{\mt R}_2(\vc x) -\mt 1|  + |t_1(\vc x)-\bar{t}_1|+ |t_2(\vc x)-\bar{t}_2|\leq \delta^*$ a.e. in $\Omega_i$. As a consequence, since $\tilde{\vc z}^{(i)}$ with $i=1,2$ are plain strains and $\mt V_i^{-1}\vc m\cdot\vc u_2^{(i)}\neq 0$, \eqref{traceZ} implies that $|\hat{\mt R}_1^T(\vc x)\hat{\mt R}_2(\vc x) -\mt 1|  + |t_1(\vc x)-\bar{t}_1|+ |t_2(\vc x)-\bar{t}_2|\leq \delta^*$ for a.e. $\vc x\in\{\hat{\vc x}\in\omega\colon\hat{\vc x}\cdot \vc m = 0\}$. 
By the fact that $\mt F_1,\mt F_2$ form a plastic junction which is locally rigid together with \eqref{ZZZ} and \eqref{rank1forZ}, it must hold $\mt R_1^T\mt R_2=\mt 1,$ $t_1=\bar{t}_1,$ $t_2=\bar{t}_2$. Therefore we deduce that there exists a measurable function $\mt R_0\colon \{\hat{\vc x}\in\omega\colon\hat{\vc x}\cdot \vc m = 0\}\to SO(3)$ such that $\nabla\vc z^{(i)}=\mt R_0(\vc x)\mt F_i$ a.e. on $\{\hat{\vc x}\in\omega\colon\hat{\vc x}\cdot \vc m = 0\}$ and for $i=1,2$. By exploiting once more \eqref{ZZZ} and \eqref{traceZ}, we deduce that $\nabla\vc z^{(i)}=\mt R_0(\vc x)\mt F_i$ a.e. in $\Omega_i.$
%
But a result by Reshetnyak (see e.g., \cite{Res,BallJames1}) implies that $\mt R_0$ must be constant, concluding the proof of the claim.
%

As a consequence, since $\tilde{\vc z}^{(i)}$ is a plain strain and linear, ${\vc z}^{(i)}$ must be linear in $\Omega_i$, with $i=1,2$, and of the form \eqref{ZZZ} with $\hat{\mt R}_1 = {\mt R}_0{\mt R}_1$, $\hat{\mt R}_2 = {\mt R}_0{\mt R}_2$ for some ${\mt R}_0\in SO(3)$. We remark that the energy of $\vcg\rho$ in $\Omega_1\cup\Omega_2$ is independent of $\mt R_0$. This, together with the fact that the energy density $W$ is non-negative imply \eqref{cheIneq0}. We remark that, every time we exploit \eqref{traceZ} we implicitly rely on the fact that, for any $\vc x\in\Omega_i$, there exists $\vc x_0\in \{\hat{\vc x}: \hat{\vc x}\cdot\vc m=0\}$ and $r_0\in \R$ such that $\vc x = \vc x_0 + r_0\mt V_i^{-1}\vc u_2^{(i)}$ and that $\Omega_i$ is convex.

%
Assume now that \eqref{cheIneqPari} holds. This, together with the fact that $\|\nabla\vcg \rho-\nabla\vc y\|_{L^\infty_{loc}}<\eps$, the shape of $W$ and the rigidity result by Reshetnyak imply 
\begin{align*}
&\nabla\vcg \rho = \mt R_{a}\mt R_1 \mt V_1,\qquad &&\text{in $\Omega_3$},\\
&\nabla\vcg \rho = \mt R_{b}\mt R_2 \mt V_2,\qquad &&\text{in $\Omega_4$},\\
&\nabla\vcg \rho = \mt R_{c},\qquad &&\text{in $\omega^c$},
\end{align*}
for some $\mt R_a,\mt R_b,\mt R_c \in SO(3).$ Again, by the Hadamard jump condition applied to $\vcg \rho$ on the planes $\{\vc x\cdot\vcg \psi_i=0\},\{\vc x\cdot\vc n_i=0\}$ and by \cite[Prop. 4]{BallJames1} we have $\mt R_a=\mt R_b = \mt R_c = \mt R_0$, which leads to the claim of the theorem. 
\end{proof}
An interesting consequence of the proof of Theorem \ref{StableThm} is the following rigidity result for plastic junctions:
\begin{theorem}
\label{StableThm2}
Let $\mt R_1\mt V_1,\mt R_2\mt V_2$ be as in Definition \ref{DefPlJ}, and $\bar{\mt F}_1,\bar{\mt F}_2\in\R^{3\times3}$ form a plastic junction at $(\bar{t}_1,\bar{t}_2)$ for $\mt R_1\mt V_1,\mt R_2\mt V_2$ which is locally rigid.  
Assume further \eqref{HH2}--\eqref{HH4} in Definition \ref{DefStabY} and:
\begin{enumerate}
\item[\textit{(4')} \label{HH5b}](Local minimiser) $\vc y\in W^{1,\infty}_{loc}(\Omega_1\cup\Omega_2; \R^3)$ is defined by 
\beq
\label{stabbile}
\vc y(\vc x) = 
\begin{cases}
\mt F_1\vc x,\qquad&\text{if $\vc x\in\Omega_1:=\bigl\{ \hat{\vc x}\in\omega\colon \hat{\vc x}\subset\mt R_{\hat{\vc m}}(\theta) \gamma_1,\,\theta\in(\theta_{\vcg \psi_1},\theta_{\vc m})\; \text{(resp. $(\theta_{\vc m},\theta_{\vcg \psi_1})$) }\bigr\}$,}\\
\mt F_2\vc x,\qquad&\text{if $\vc x\in\Omega_2:=\bigl\{ \hat{\vc x}\in\omega\colon \hat{\vc x}\subset\mt R_{\hat{\vc m}}(\theta) \gamma_1,\,\theta\in(\theta_{\vc m},\theta_{\vcg \psi_2})\; \text{(resp. $(\theta_{\vcg \psi_2},\theta_{\vc m})$) }\bigr\}$,}
\end{cases}
\eeq
\item[\textit{(5)}]\label{HH1} $\bar{\mt F}_1,\bar{\mt F}_2$ enjoy the separation property.
\end{enumerate}
Then, if {$(\mt V_i^{2}\vcg\phi_i\times\vcg\psi_i)\cdot\vc m\neq 0,$} for $i=1,2$, there exists $\eps>0$ such that every $\vcg \rho \in W^{1,\infty}_{loc} (\Omega_1\cup\Omega_2;\R^3)$ satisfying
\begin{enumerate}[a)]
\item$\int_{(\Omega_1\cup\Omega_2)\cap B_r} W(\nabla\vcg\rho)\,\mathrm{d}\vc x <\infty$ for any open ball $B_r$ centred at $\vc 0$ and of arbitrary radius $r>0$,
\item  $\| \nabla\vcg\rho - \nabla\vc y\|_{L^{\infty}}\leq \eps$,
\item $\vcg\rho$ is $1-1$,
\end{enumerate}
is of the form $\vcg \rho = \mt R\vc y+\vc c$ for some $\mt R\in SO(3), \vc c\in\R^3$.
\end{theorem}

\section{$V_{II}$ junctions in \Tn}
\label{Appl}
{In this section we study the presence of $V_{II}$ junctions in cubic to orthorhombic transformations when the stretch tensors have both the middle eigenvalue and the determinant equal to one. This is done under the additional hypothesis that a parameter $\lambda$ of the stretch tensors representing the lattice deformations lies in the physically relevant interval $\lambda\in(1,\sqrt2)$. A similar argument could be applied to study the case when $\lambda<1$. As explained below, this situation is a good approximation of the martensitic transformation in \Tn\,and similar materials. We prove that the existing $V_{II}$ junctions are locally stable in the case where the energy has all the wells, that is where the elastic energy is null on $\bigcup_{i=1}^6 SO(3)\mt U_i$, where $\mt U_i$ are the six matrices transforming a cubic lattice into an orthorhombic one, and where we consider all possible slip systems for body centred cubic austenite. However, the generality of the results leads to many long computations and, for this reason, {in this section some of the hypotheses of Theorem \ref{StableThm} are verified numerically or with the aid of a plot.} At the end of the section we compare the results obtained with experimental results.}\\

The transformation in \Tn\, is from a cubic to an orthorhombic lattice, and therefore the stretch tensors $\mt U_i$ describing the change of lattice vectors are given by
{\footnotesize
\beq
\label{cubictoortho}
\begin{split}
{\mt U}_1 = \left[\begin{array}{ ccc } d & 0 & 0 \\ 0 & \frac{1+\lambda}2 & \frac{\lambda-1}2 \\ 0 & \frac{\lambda-1}2 & \frac{1+\lambda}2 \end{array}\right],\quad
{\mt U}_2 = \left[\begin{array}{ ccc } d & 0 & 0 \\ 0 & \frac{1+\lambda}2 &- \frac{\lambda-1}2 \\ 0 & -\frac{\lambda-1}2 & \frac{1+\lambda}2 \end{array}\right],\quad
{\mt U}_3 = \left[\begin{array}{ ccc } \frac{1+\lambda}2 & 0 & \frac{\lambda-1}2 \\ 0 & d & 0 \\ \frac{\lambda-1}2 & 0 & \frac{1+\lambda}2 \end{array}\right],\\
{\mt U}_4 = \left[\begin{array}{ ccc } \frac{1+\lambda}2 & 0 & -\frac{\lambda-1}2 \\ 0 & d & 0 \\ -\frac{\lambda-1}2 & 0 & \frac{1+\lambda}2 \end{array}\right],\quad
{\mt U}_5 = \left[\begin{array}{ ccc } \frac{1+\lambda}2 & \frac{\lambda-1}2 & 0 \\ \frac{\lambda-1}2 & \frac{1+\lambda}2 & 0 \\ 0 & 0 & d \end{array}\right],\quad
\tilde{\mt U}_6 = \left[\begin{array}{ ccc } \frac{1+\lambda}2 & -\frac{\lambda-1}2 & 0 \\ -\frac{\lambda-1}2 & \frac{1+\lambda}2 & 0 \\ 0 & 0 & d \end{array}\right].
\end{split}
\eeq
}
Since in \Tn\, the middle eigenvalue of the transformation matrices $\lambda_2$ is such that (see \cite{Inamura}) $|\lambda_2-1|<4\cdot 10^{-6}$ we implicitly assumed it to be equal to one in \eqref{cubictoortho}. Therefore, the eigenvalues of the $\mt U_i$'s are $d,1,\lambda$, and, coherently with the lattice deformation in \Tn, we assume also that $0<d<1<\lambda.$ A similar analysis could be worked out in the case where $d>1>\lambda>0.$ Under these assumptions, \cite[Prop. 4]{BallJames1} guarantees for every $i=1,\dots,6$ the existence of two couples of vectors $(\vc a_i^-,\vc n_i^-)$ and $(\vc a_i^+,\vc n_i^+)$ such that
$$
\mt R_i^+\mt U_i = \mt 1 + \vc a_i^+\otimes\vc n_i^+,\qquad \mt R_i^-\mt U_i = \mt 1 + \vc a_i^-\otimes\vc n_i^-,
$$
for some $\mt R_i^+,\mt R_i^-\in SO(3).$ The different $\vc a_i^\pm,\vc n_i^\pm$ depending on $\lambda,d$ are given by:
\begin{multicols}{2}
\begin{align*}
&\vc a^+_1 = \alpha(-\gamma,1,1),\quad &&\vc n^+_1 = (\beta,1,1),\\ 
&\vc a^-_1 = \alpha(\gamma,1,1),\quad &&\vc n^-_1 = (-\beta,1,1),\\
&\vc a^+_2 = \alpha(-\gamma,-1,1),\quad &&\vc n^+_2 = (\beta,-1,1),\\ 
& \vc a^-_2 = \alpha(\gamma,-1,1),\quad && \vc n^-_2 = (-\beta,-1,1),\\
&\vc a^+_3 = \alpha(1,-\gamma,1),\quad &&\vc n^+_3 = (1,\beta,1),\\
& \vc a^-_3 = \alpha(1,\gamma,1),\quad &&\vc n^-_3 = (1,-\beta,1),
\end{align*}

\begin{align*}
&\vc a^+_4 = \alpha(-1,-\gamma,1),\quad &&\vc n^+_4 = (-1,\beta,1),\\
& \vc a^-_4 = \alpha(-1,\gamma,1),\quad &&\vc n^-_4 = (-1,-\beta,1),\\
&\vc a^+_5 = \alpha(1,1,-\gamma),\quad &&\vc n^+_5 = (1,1,\beta),\\
& \vc a^-_5 = \alpha(1,1,\gamma),\quad &&\vc n^-_5 = (1,1,-\beta),\\
&\vc a^+_6 = \alpha(-1,1,-\gamma),\quad &&\vc n^+_6 = (-1,1,\beta),\\
& \vc a^-_6 = \alpha(-1,1,\gamma),\quad &&\vc n^-_6 = (-1,1,-\beta),
\end{align*}
\end{multicols}
where 
$$
\alpha = \frac{d(\lambda^2-1)}{2(d+\lambda)}, \qquad \beta = -\frac{\sqrt{2(1 - d^2)}}{\sqrt{\lambda^2 - 1}} ,\qquad \gamma = -\frac{\lambda}{d}\frac{\sqrt{2(1 - d^2)}}{\sqrt{\lambda^2 - 1}}.
$$
As explained in the introduction, in experiments for \Tn\,\cite{Inamura} one observes the nucleation of different plates of martensite $\mt F_i$ with $\mt F_i = \mt 1+\vc a_i^{\sigma_i}\otimes \vc n_i^{\sigma_i}$
, where $\sigma_i\in\{+,-\}$ and $i\in\{1,\dots,6\}$, which expand until they encounter another plate of martensite $\mt F_j$ with similar properties. The nucleation is occurring at the interior of the domain, 
that is, an island of martensite with deformation gradient $\mt F_i$ grows in the middle of an austenite domain with deformation gradient $\mt 1$. Therefore, the deformation in the martensite region must be close to volume preserving, i.e., in a first aproximation $\det \mt F_i = \det \mt U_i=1$, and hence, $d=\lambda^{-1}.$ 
In reality, for \Tn\,some elasto-plastic effects take place to accommodate the nucleation at the interior. However, in order to simplify the analysis below, motivated by the experimental value of $\det \mt U_i$ which is very close to one (the experimental values yield $|\det \mt U_i-1|<1.9\cdot 10^{-3}$ \cite{Inamura}), we assume $d=\lambda^{-1}.$ 
We remark that, the analysis below holds also in the case $d = 0.9661, \lambda = 1.0331$ (the lattice parameters for \Tn) and for every $(d,\lambda)\in (0,1)\times(1,\infty)\setminus\bigcup_{i=1}^N \mathrm{Im}(c_i)$, where $\mathrm{Im}(c_i)$ is the image of $c_i$, and $c_i$ are a finite number $N\in\mathbb{N}$ of polynomial curves $c_i\colon(0,1)\to(1,\infty)$. Furthermore we restrict ourselves to the physically relevant range $\lambda\in(1,\sqrt 2)$. It is worth noticing that when $\lambda=\sqrt{2}$ the cofactor conditions are satisfied, and hence stress free triple junctions are possible (see e.g., \cite{JamesHyst,FDP3}). We now want to find plastic junctions as in Definition \ref{DefPlJ} and where $\mt R_1\mt V_1= \mt 1 +\vc a_1^+\otimes \vc n_1^+$ and $\mt R_2\mt V_2$ is of the form (cf. Remark \ref{CompatCon1})
\beq
\label{rango1Plate}
\mt 1 +\vc a_i^{\sigma_i}\otimes \vc n_i^{\sigma_i}
\eeq 
for $\sigma_i\in\{+,-\}$ and some $i\in\{1,\dots,6\}$. The case where $\mt R_1\mt V_1$ has the form \eqref{rango1Plate} but $(i,\sigma_i)\neq (1,+)$ can be treated similarly, or simply deduced from our case by symmetry. We remark that, under our assumptions,
$$
\vc a_i^{\sigma_i} \times \vc a_j^{\sigma_j}\neq\vc 0,\qquad \vc n_i^{\sigma_i} \times \vc n_j^{\sigma_j}\neq\vc 0,
$$
for any $(i,\sigma_i) \neq (j,\sigma_j) \in \{1,\dots,6\}\times\{+,-\}$. As a consequence $\rank\bigl(\mt R_1\mt V_1-\mt R_2\mt V_2\bigr) = 2$. 

We are now ready to state Theorem \ref{TnThm} which investigates the possibility to form plastic junctions and  $V_{II}$ junctions in a one-parameter family of deformation gradients, and in particular in \Tn. The stability of the existing $V_{II}$ junctions is also proved by verifying the hypotheses of Theorem \ref{StableThm}. The results are compared with experimental results in Section \ref{Comparison with Experiments}. The theorem reads as follows:
\begin{theorem}
\label{TnThm}
Let $\lambda\in(1,\sqrt2)$. Let $\mathcal{M}=\bigcup_{i=1}^6 \mt U_i$ and $\mathcal{S}$ be the set of all possible simple slips for body centred cubic lattices. 
Let us also define 
$$
\eta_1 =  \frac{2\lambda^4 + 5\sqrt2 \lambda^3+4\lambda^2-5\sqrt2\lambda-6}{2(2\lambda^4 + 5\sqrt2 \lambda^3 - 4\lambda^2 + 3\sqrt2\lambda+2)},\qquad 
 \eta_2 = \frac{2\lambda^4 + \sqrt2\lambda^3 - 4\lambda^2-\sqrt2\lambda+2 }{2(2\lambda^4 + 5\sqrt2\lambda^3 - 4\lambda^2 + 3\sqrt2\lambda+2)};
$$
and
$$
\xi_1= -\frac{2\lambda^4 - 5\sqrt2\lambda^3 + 4\lambda^2 + 5\sqrt2\lambda-6}{2(2\lambda^4 - 5\sqrt2\lambda^3 - 4\lambda^2 - 3\sqrt2\lambda + 2)},\qquad 
\xi_2 = \frac{2\lambda^4 - \sqrt2\lambda^3 - 4\lambda^2 + \sqrt2\lambda + 2 }{2(2\lambda^4 - 5\sqrt2\lambda^3 - 4\lambda^2 - 3\sqrt2\lambda + 2)};
$$
Then, there exist a plastic junction (in the sense of Definition \ref{DefPlJ}) for $\mt 1 +\vc a_1^+\otimes\vc n_1^+$ and $\mt 1 +\vc a_i^{\sigma_i}\otimes\vc n_i^{\sigma_i}$ with $i\in\{2,\dots,6\}$, $\sigma_i\in\{+,-\}$ if and only if
\begin{enumerate}[(a)]
\item \label{Ca}$(i,\sigma_i)=(3,+)$, $\vcg\psi_1=\vcg\psi_2 = (-1,1,0)$ and
\begin{align*}
\vcg\phi_1 = -(1,1,1), \, \vcg \phi_2 = (1,1,-1),\, (\bar t_1,\bar t_2) = (\eta_1,\eta_2),\quad \text{ or }\\
\vcg\phi_1 = (1,1,-1), \, \vcg \phi_2 = -(1,1,1),\, (\bar t_1,\bar t_2) = (-\eta_2,-\eta_1);
\end{align*}
\item \label{Cb}$(i,\sigma_i)=(4,-)$, $\vcg\psi_1=\vcg\psi_2 = (1,1,0)$ and
\begin{align*}
\vcg\phi_1 = (-1,1,1), \, \vcg \phi_2 = (-1,1,-1),\, (\bar t_1,\bar t_2) = (\xi_1,\xi_2),\quad \text{ or }\\
\vcg\phi_1 = (-1,1,-1), \, \vcg \phi_2 = (-1,1,1),\, (\bar t_1,\bar t_2) = (-\xi_2,-\xi_1);
\end{align*}
\item \label{Cc}$(i,\sigma_i)=(5,+)$, $\vcg\psi_1=\vcg\psi_2 = (-1,0,1)$ and
\begin{align*}
\vcg\phi_1 = -(1,1,1), \, \vcg \phi_2 = (1,-1,1),\, (\bar t_1,\bar t_2) = (\eta_1,\eta_2),\quad \text{ or }\\
\vcg\phi_1 = (1,-1,1), \, \vcg \phi_2 = -(1,1,1),\, (\bar t_1,\bar t_2) = (-\eta_2,-\eta_1);
\end{align*}
\item \label{Cd}$(i,\sigma_i)=(6,-)$, $\vcg\psi_1=\vcg\psi_2 = (1,0,1)$ and
\begin{align*}
\vcg\phi_1 = (-1,1,1), \, \vcg \phi_2 = (-1,-1,1),\, (\bar t_1,\bar t_2) = (\xi_1,\xi_2),\quad \text{ or }\\
\vcg\phi_1 = (-1,-1,1), \, \vcg \phi_2 = (-1,1,1),\, (\bar t_1,\bar t_2) = (-\xi_2,-\xi_1);
\end{align*}
\end{enumerate}
All these plastic junctions can form locally stable $V_{II}$ junctions in the sense of Definition \ref{DefStabY}.
There exists no $V_{II}$ junction (in the sense of Definition \ref{DefVII}) between $\mt 1 +\vc a_1^+\otimes\vc n_1^+$ and $\mt 1 +\vc a_1^{-}\otimes\vc n_1^{-}$.
\end{theorem}

Figure \ref{FigSs} shows the dependence of $\eta_1,\eta_2$ and $\xi_1,\xi_2$ on $\lambda$.
\begin{figure}
\begin{subfigure}{.45\textwidth}
  \centering
  \includegraphics[width=.9\linewidth]{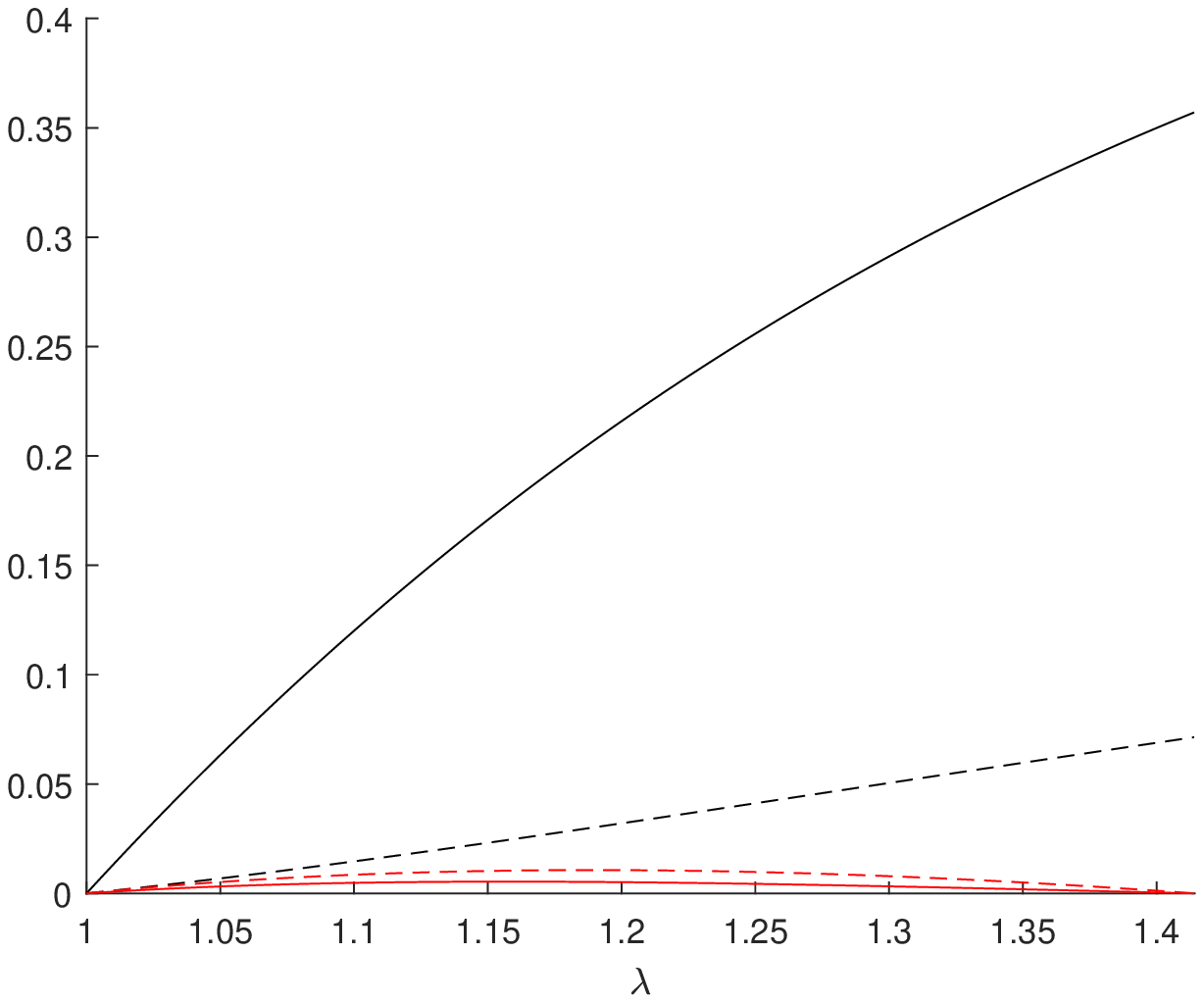}
  \caption{}
  \label{fig99}
\end{subfigure}%
\begin{subfigure}{.45\textwidth}
  \centering
  \includegraphics[width=.9\linewidth]{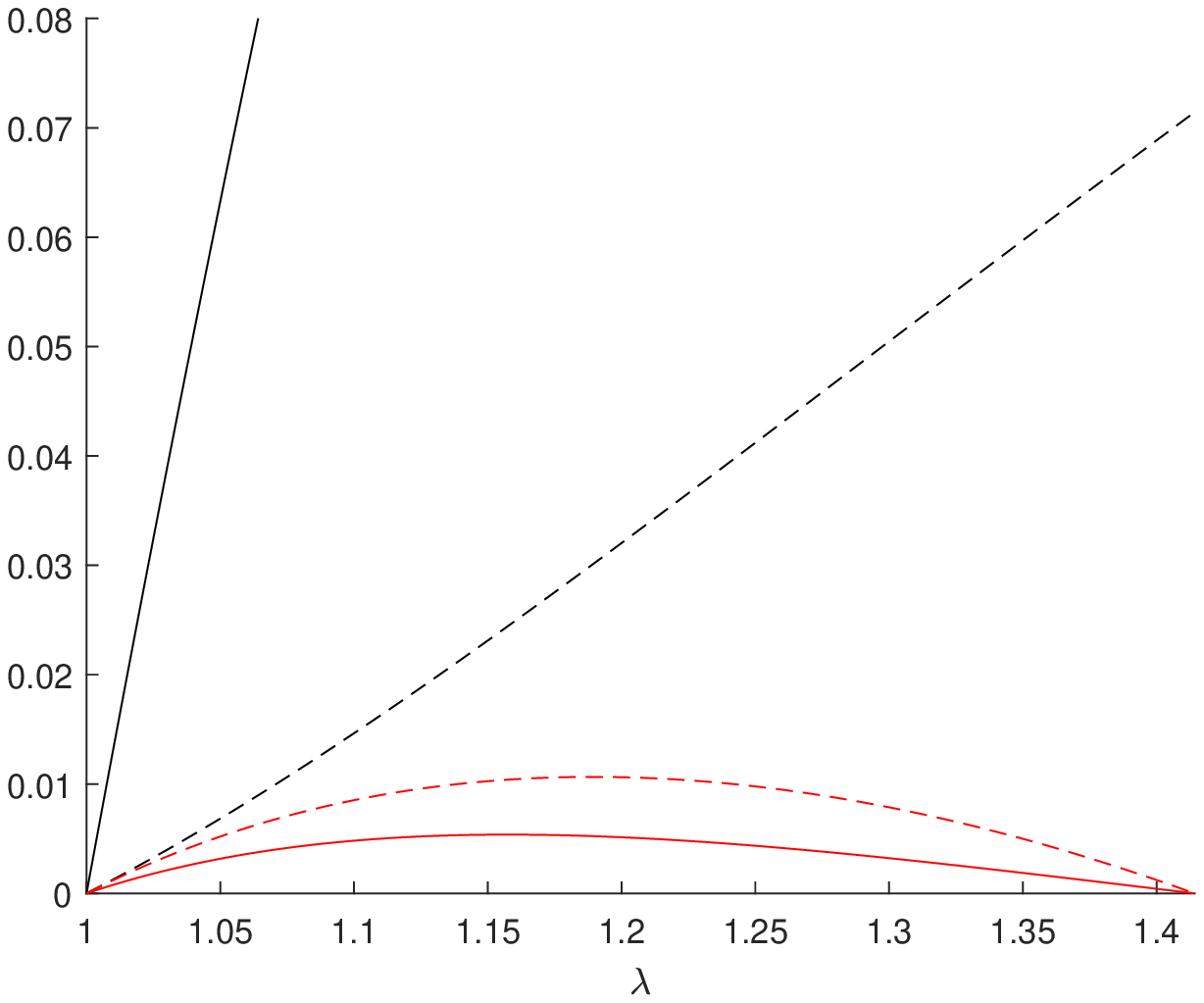}
  \caption{}
  \label{fig1}
\end{subfigure}%
\caption{Plotting the dependence of $\eta_1,\eta_2$ and $\xi_1,\xi_2$ on $\lambda$. In black $\eta_1$ (continuous line) and $\eta_2$ (dashed line). In red $\xi_1$ (continuous line) and $\xi_2$ (dashed line). On the right hand side, the plot is a zoom of the plot on the left.}
\label{FigSs}
\end{figure}
The results in Theorem \ref{TnThm} are compared with experimental observations in Section \ref{Comparison with Experiments}.

\subsection{Verification of Theorem \ref{TnThm}}
The proof investigates first the existence of plastic junctions when $i\in\{2,\dots,6\}$. We then check that this plastic junctions can form a locally stable $V_{II}$ junction. To this aim, we need the verification of the assumptions of Theorem \ref{StableThm}. These are technical and require long and uninteresting computations. Therefore, the verification of some of the assumptions of Theorem \ref{StableThm} is checked numerically or by means of a plot. Finally, we show that no $V_{II}$ junction (according to Definition \ref{DefVII}) exists when $(i,\sigma_i)=(1,-).$ We divide the argument into steps to simplify the presentation. 

\paragraph{Existence of plastic junctions.}
By Lemma \ref{NecessityCond}, and taking in consideration all the slip systems for body centred cubic lattices (see Section \ref{Nonlin}), we can see that the necessary conditions to have plastic junctions for $\mt 1 +\vc a_1^+\otimes\vc n_1^+$ and $\mt 1 +\vc a_i^{\sigma_i}\otimes\vc n_i^{\sigma_i}$ with $i\in\{2,\dots,6\}$, $\sigma_i\in\{+,-\}$ are satisfied by each of the points \eqref{caso2}--\eqref{caso5} below:
\begin{multicols}{2}
\begin{enumerate}[(i)]
\item \label{caso2}$(i,\sigma_i)=(3,+)$ and 
$\vcg \psi = (-1,1,0)$;
\item \label{caso3}$(i,\sigma_i)=(4,-)$ 
and $\vcg \psi = (1,1,0)$;
\item \label{caso4}$(i,\sigma_i)=(5,+)$ 
and $\vcg \psi = (-1,0,1)$;
\item \label{caso5}$(i,\sigma_i)=(6,-)$ 
and $\vcg \psi = (1,0,1)$.
\end{enumerate}
\end{multicols}
In all the above cases $\vcg \psi_1=\vcg \psi_2$ and we therefore simplified notation by writing $\vcg \psi.$ We now show that these conditions are sufficient to have plastic junctions. Thanks to Proposition \ref{PropMeet} we can find $t_1,t_2\in\R$ such that
\beq
\label{Rk1eq1s}
\rank\bigl((\mt 1 + \vc a_1^+\otimes\vc n_1^+)(\mt 1 + t_1\vcg \phi_1\otimes\vcg \psi) - (\mt 1 + \vc a_i^{\sigma_i}\otimes\vc n_i^{\sigma_i})(\mt 1 + t_2\vcg \phi_2\otimes\vcg \psi)\bigr) = 1.
\eeq
Here, again, $\vcg\phi_1,\vcg\phi_2$ are the two different Burger's vectors in the plane orthogonal to $\vcg\psi$, among the slip systems for body centred cubic lattices. We recall that, in these cases, for every $\vcg\psi$ there are exactly two (up to sign change) $\vcg \phi$ such that $(\vcg\phi,\vcg\psi)$ is a slip system for body centred cubic lattices. By post-multiplying the above equation by $(\mt 1 + t_1\vcg \phi_1\otimes\vcg \psi)^{-1}(\mt 1 + t_2\vcg \phi_2\otimes\vcg \psi)^{-1}$ we get
\beq
\label{Rk1eq2s}
\rank\bigl((\mt 1 + \vc a_1^+\otimes\vc n_1^+)(\mt 1 - t_2\vcg \phi_2\otimes\vcg \psi) - (\mt 1 + \vc a_i^{\sigma_i}\otimes\vc n_-^{\sigma_i})(\mt 1 - t_1\vcg \phi_1\otimes\vcg \psi)\bigr) = 1.
\eeq
Therefore, if the solution of \eqref{Rk1eq2s} is unique, it can be identified with the unique solution of \eqref{Rk1eq1s}. Some computations conclude the proof of \eqref{Ca}--\eqref{Cd}.

\paragraph{Local rigidity of plastic junctions.}
In order to verify that the constructed plastic junctions are locally rigid (in the sense of Definition \ref{DefPlJ}) we make use of Proposition \ref{LocalRigidProp}. Under our hypotheses, $\cof(\mt R_1\mt V_1-\mt R_2\mt V_2) = (\vc a_1^+\times \vc a_i^{\sigma_i})\otimes (\vc n_1^+\times \vc n_i^{\sigma_i})$, and, in the notation of Proposition \ref{LocalRigidProp}, $\hat{\vc m} = \frac{\vc n_1^+\times\vc n_i^{\sigma_i}}{|\vc n_1^+\times\vc n_i^{\sigma_i}|}$ and $\hat{\vc b} = |\vc n_1^+\times\vc n_i^{\sigma_i}|\,\vc a_1^+\times\vc a_i^{\sigma_i}$. Furthermore, defining
\begin{align*}
M^+_1 &:= - 2\sqrt2 \lambda^5-8\lambda^4+7\sqrt2\lambda^3 + 2\lambda^2 + 3\sqrt2\lambda - 2,\quad M_2^+:= 2\lambda^4 + 7\sqrt2\lambda^3-16\lambda^2+\sqrt2\lambda + 6,\\ 
M_3^+ &:= -2\lambda\bigl(\sqrt2\lambda^4 + 5\lambda^3 - 2\sqrt2\lambda^2 + 3\lambda +\sqrt2 \bigr ),\quad M^-_1:= -\bigl(2\lambda^4 - 7\sqrt2\lambda^3-16\lambda^2-\sqrt2\lambda + 6\bigr),\\
M_2^- &:=  2\sqrt2 \lambda^5-8\lambda^4-7\sqrt2\lambda^3 + 2\lambda^2 - 3\sqrt2\lambda - 2,\quad
M_3^- := 2\lambda\bigl(\sqrt2\lambda^4 - 5\lambda^3 - 2\sqrt2\lambda^2 - 3\lambda +\sqrt2 \bigr ),
\end{align*}
we have that for the first option in the cases \eqref{Ca}--\eqref{Cd} $\vc m$ is respectively parallel to
\beq
\label{glims}
(M_1^+,M_2^+,M_3^+),\quad (M_1^-,M_2^-,M_3^-),\quad (M_1^+,M_3^+,M_2^+),\quad (M_1^-,M_3^-,M_2^-).
\eeq
For the second option in the cases \eqref{Ca}--\eqref{Cd}, $\vc m$ can be deduced by pre-multiplying the vectors in \eqref{glims} by $(\mt 1 + t_2\vcg \phi_2\otimes\vcg \psi)^{-T}(\mt 1 + t_1\vcg \phi_1\otimes\vcg \psi)^{-T}$. We now have all the ingredients to show (see \eqref{ipotesiR})
\beq
\label{funzRigid}
f(\lambda):=\Bigl( {\mt R_1\mt V_1\hat{\vc m}} \times\mt R_1\mt V_1	 \bigl(\vc v + \bar{t}_1\vcg\phi_1 (\vcg\psi\cdot\vc v)\bigr)\Bigr) \cdot \Bigl( \mt R_1\mt V_1 \vcg\phi_1\times\mt R_2\mt V_2 \vcg\phi_2
 \Bigr)\neq 0, \qquad\vc v = \vc m \times \hat{\vc m}.
\eeq
{The easiest way to show this is graphically, by plotting the function $f$ for the cases \eqref{Ca}--\eqref{Cd} in Figure \ref{RigidPict}.
\begin{figure}
\begin{subfigure}{.45\textwidth}
  \centering
  \includegraphics[width=.9\linewidth]{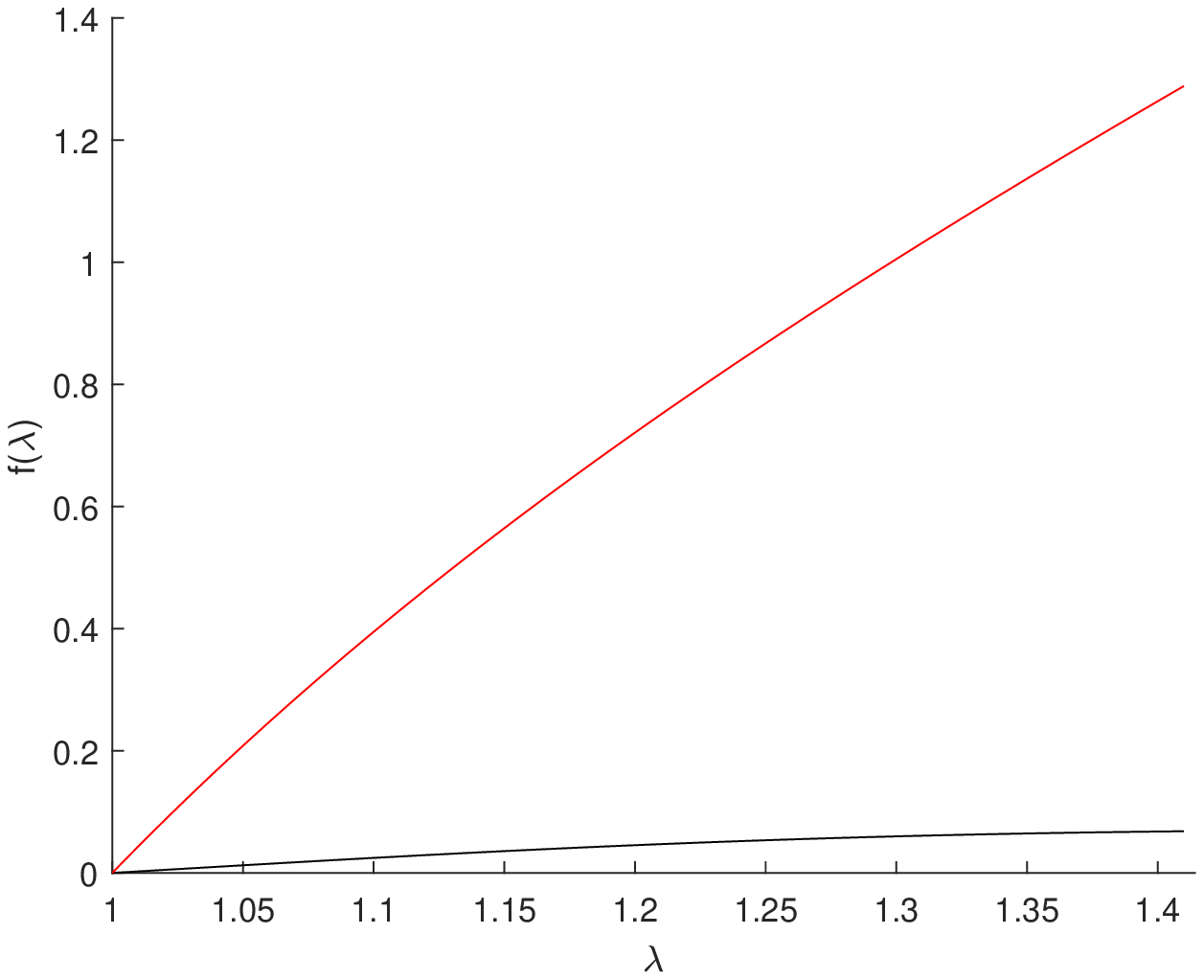}
  \caption{}
  \label{fig99}
\end{subfigure}%
\begin{subfigure}{.45\textwidth}
  \centering
  \includegraphics[width=.9\linewidth]{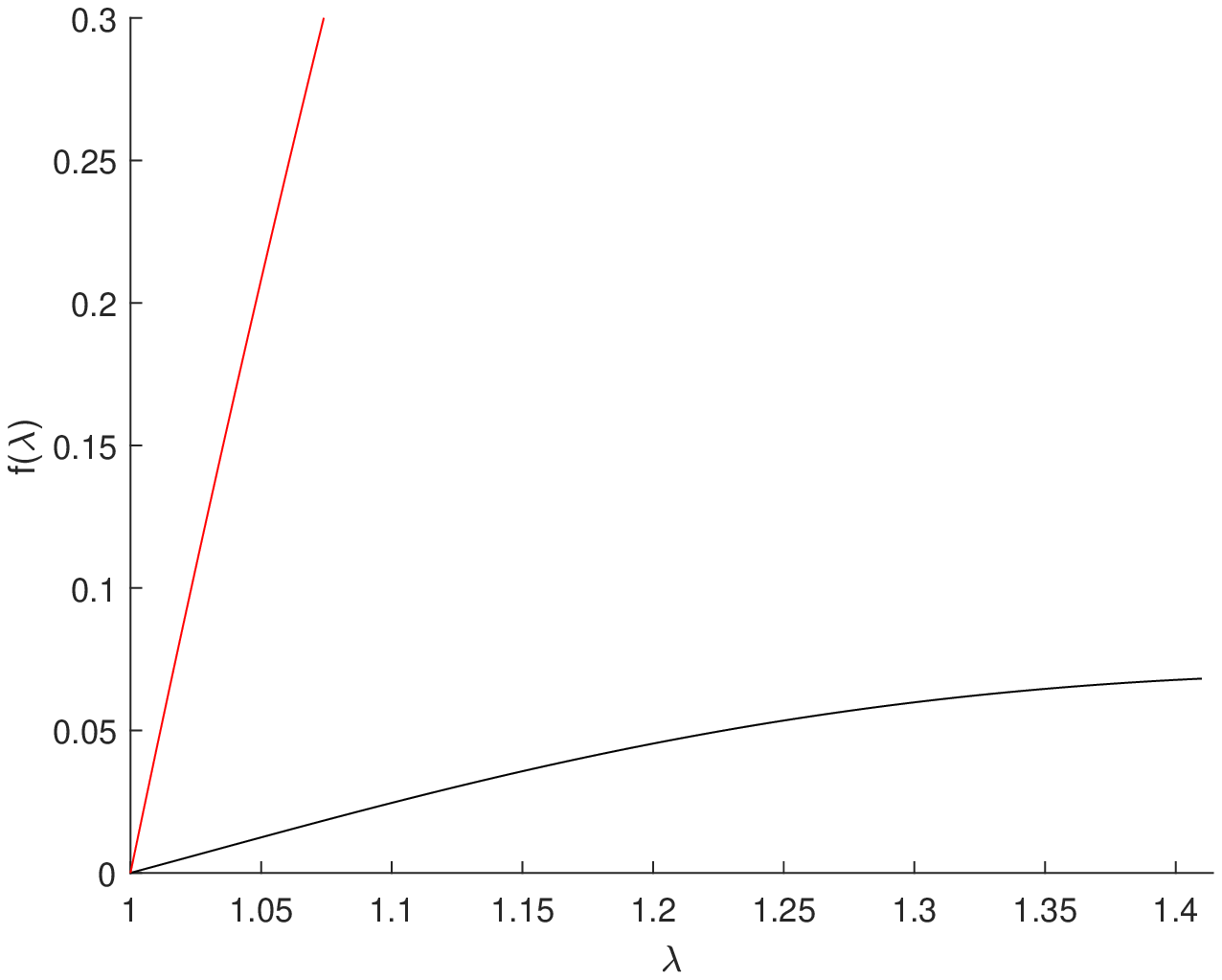}
  \caption{}
  \label{fig1}
\end{subfigure}%
\caption{\label{RigidPict} Plotting $f(\lambda)$ against $\lambda$ where $f$ is as in \eqref{funzRigid}. In black the cases given in \eqref{Ca} and in \eqref{Cc}, while in red the cases given in \eqref{Cb} and in \eqref{Cd}. On the right a zoom of the plot.}
\end{figure}
\paragraph{Separation property.}
Let $\mt F_1 = (\mt 1 + \vc a_1^+\otimes\vc n^+_1)(\mt 1 + \bar t_1\vcg\phi_1\otimes\vcg\psi)$ and $\mt F_2 = (\mt 1 + \vc a_i^{\sigma_i}\otimes\vc n_i^{\sigma_i})(\mt 1 + \bar t_2\vcg\phi_2\otimes\vcg\psi)$, where $(i,\sigma_i),$ $\bar t_1,\bar t_2$ and $\vcg\phi_1,\vcg\phi_2,\vcg\psi$ are as in \eqref{Ca}--\eqref{Cd}. We first claim that for each $\lambda\in(1,\sqrt2)$ there exists $\rho_0>0$ such that
\begin{align}
\label{sepaIneq}
g_1(t):=\bigl|\mt F_1^T\mt F_1 - (\mt 1 + t \vcg \psi_l\otimes\vcg\phi_l)\mt U_j^2(\mt 1 + t \vcg \phi_l\otimes\vcg\psi_l)\bigr|^2\geq \rho_0^2, \\
\label{sepaIneq2}
g_2(t):=\bigl|\mt F_2^T\mt F_2 - (\mt 1 + t\vcg \psi_l\otimes\vcg\phi_l)\mt U_j^2(\mt 1 + t \vcg \phi_l\otimes\vcg\psi_l)\bigr|^2\geq \rho_0^2
\end{align}
for any $t\in\R$, whenever at least one out of 
\begin{align*}
\mt U_j\neq \mt U_1 \quad \text{ or }\quad \vcg \phi_1\otimes\vcg\psi \neq \vcg \phi_l\otimes\vcg\psi_l \in \mathcal{S},\qquad&\text{in the case of \eqref{sepaIneq},}\\
\mt U_j\neq \mt U_i \quad \text{ or }\quad \vcg \phi_2\otimes\vcg\psi \neq \vcg \phi_l\otimes\vcg\psi_l \in \mathcal{S},\qquad&\text{in the case of \eqref{sepaIneq2},}
\end{align*}
holds. The amount of cases to be checked is huge. Indeed, there are four different junctions to be checked, that is case \eqref{Ca}--\eqref{Cd}, each with two subcases. For each of these cases we have to verify two inequalities, namely \eqref{sepaIneq}--\eqref{sepaIneq2}, which must hold for six possible different $j$'s, and for forty-eight possible slip-systems. The total amount of cases to be checked is hence $4\cdot 2\cdot 2 \cdot ( 6 \cdot 48-1) = 4592$. Since we were not able to identify a unique simple algorithm to verify \eqref{sepaIneq}--\eqref{sepaIneq2} in all these cases, we verified it numerically. Indeed, for any $\lambda>0$, any $\mt U_j$, $j=\{1,\dots,6\}$ and $\vcg\phi_l\otimes\vcg\psi_l \in \mathcal{S}$ the functions $g_1,g_2$ are fourth order polynomials  in $t$ which can be minimised numerically. The smooth dependence of $g_1,g_2$ on $\lambda,t$ make the numerical problem well posed. 
Numerically one observes that the claim is true for any $\lambda\in(1,\sqrt{2})$ (cf. Figure \ref{fig:fig}).
\begin{figure}
\begin{subfigure}{.5\textwidth}
  \centering
  \includegraphics[width=.99\linewidth]{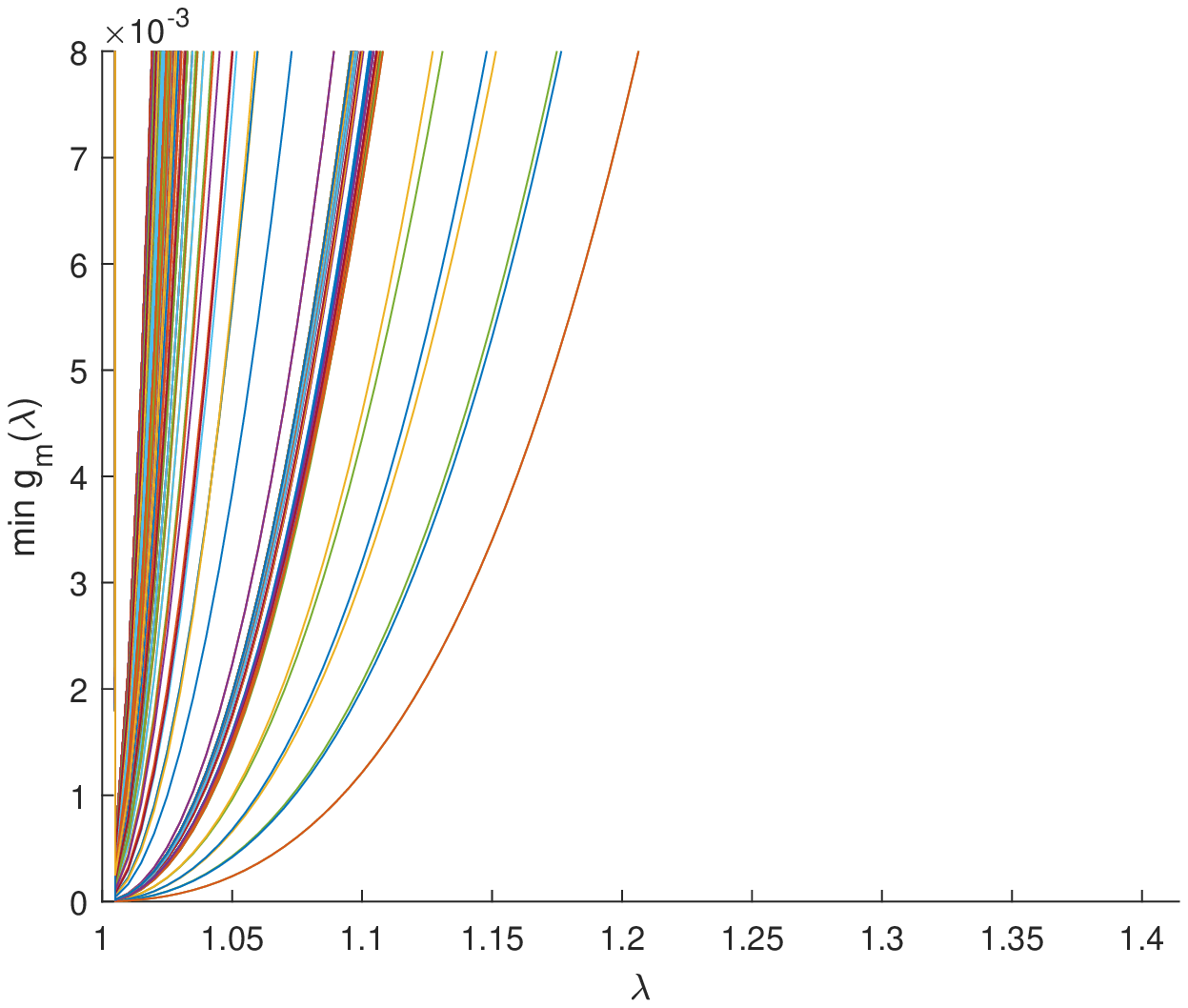}
  \caption {}
  \label{fig0a}
\end{subfigure}%
\begin{subfigure}{.5\textwidth}
  \centering
  \includegraphics[width=.99\linewidth]{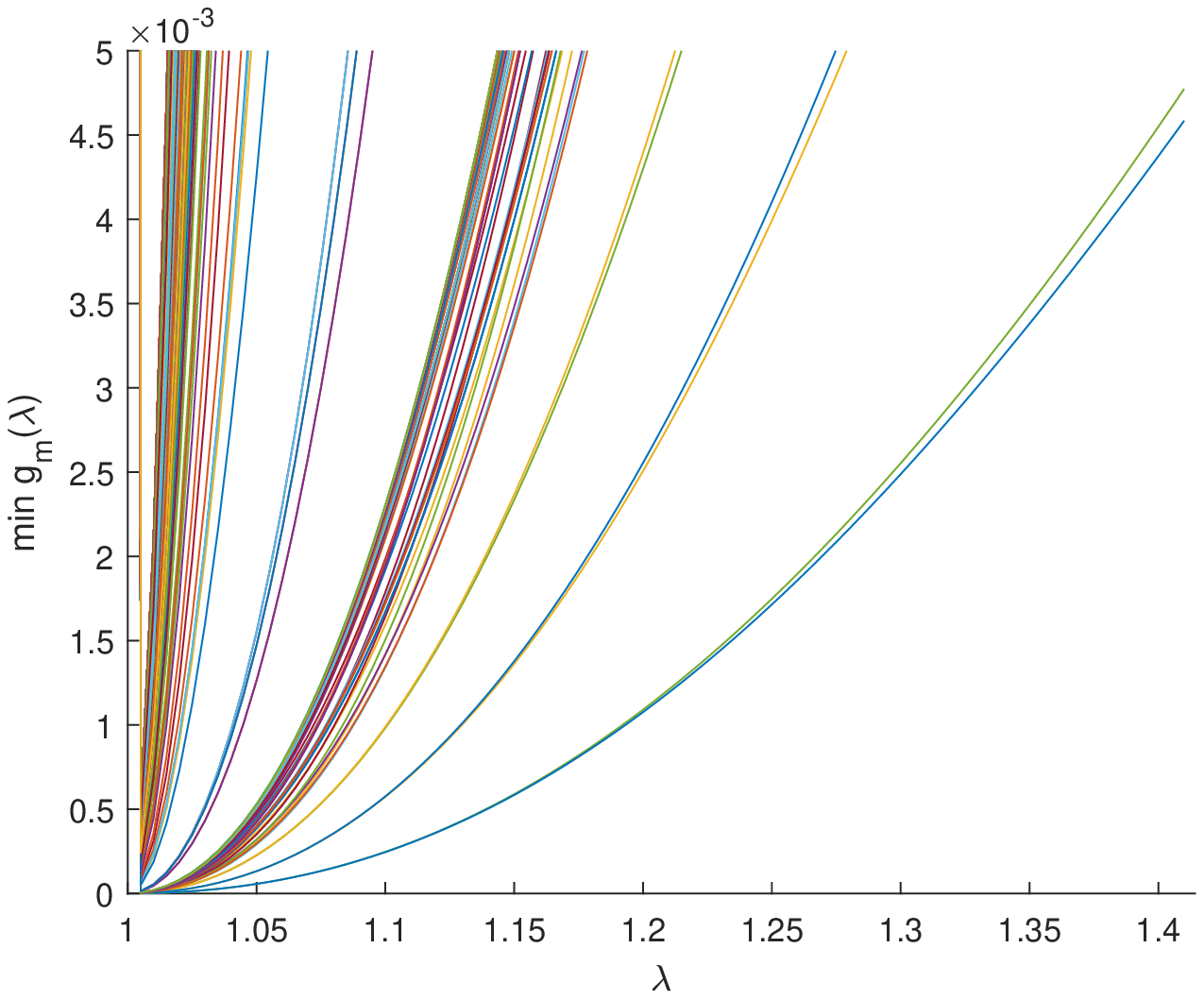}
  \caption{}
  \label{fig0b}
\end{subfigure}
\newline
\begin{subfigure}{.5\textwidth}
  \centering
  \includegraphics[width=.99\linewidth]{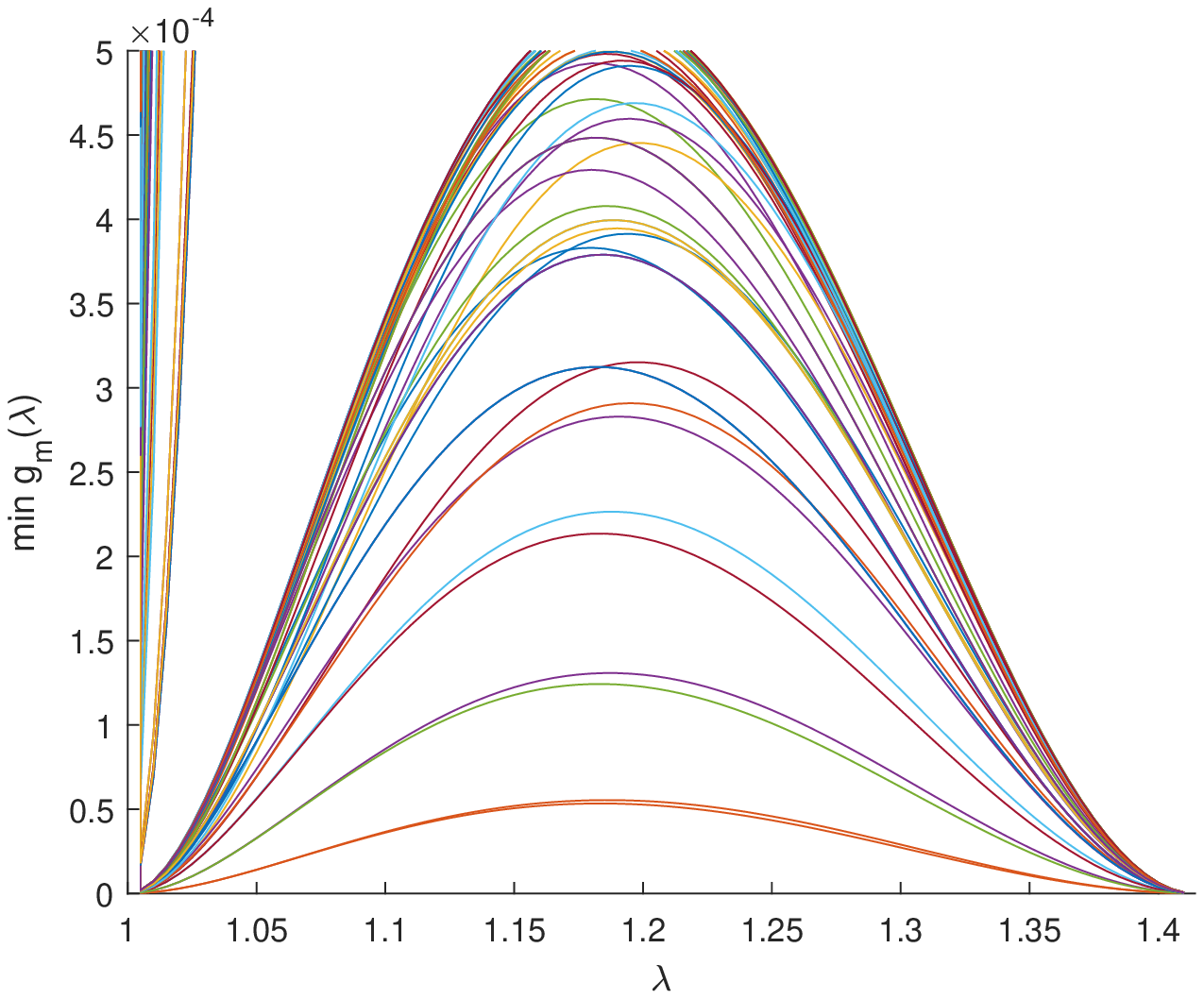}
  \caption{}
  \label{fig0c}
\end{subfigure}
\begin{subfigure}{.5\textwidth}
  \centering
  \includegraphics[width=.99\linewidth]{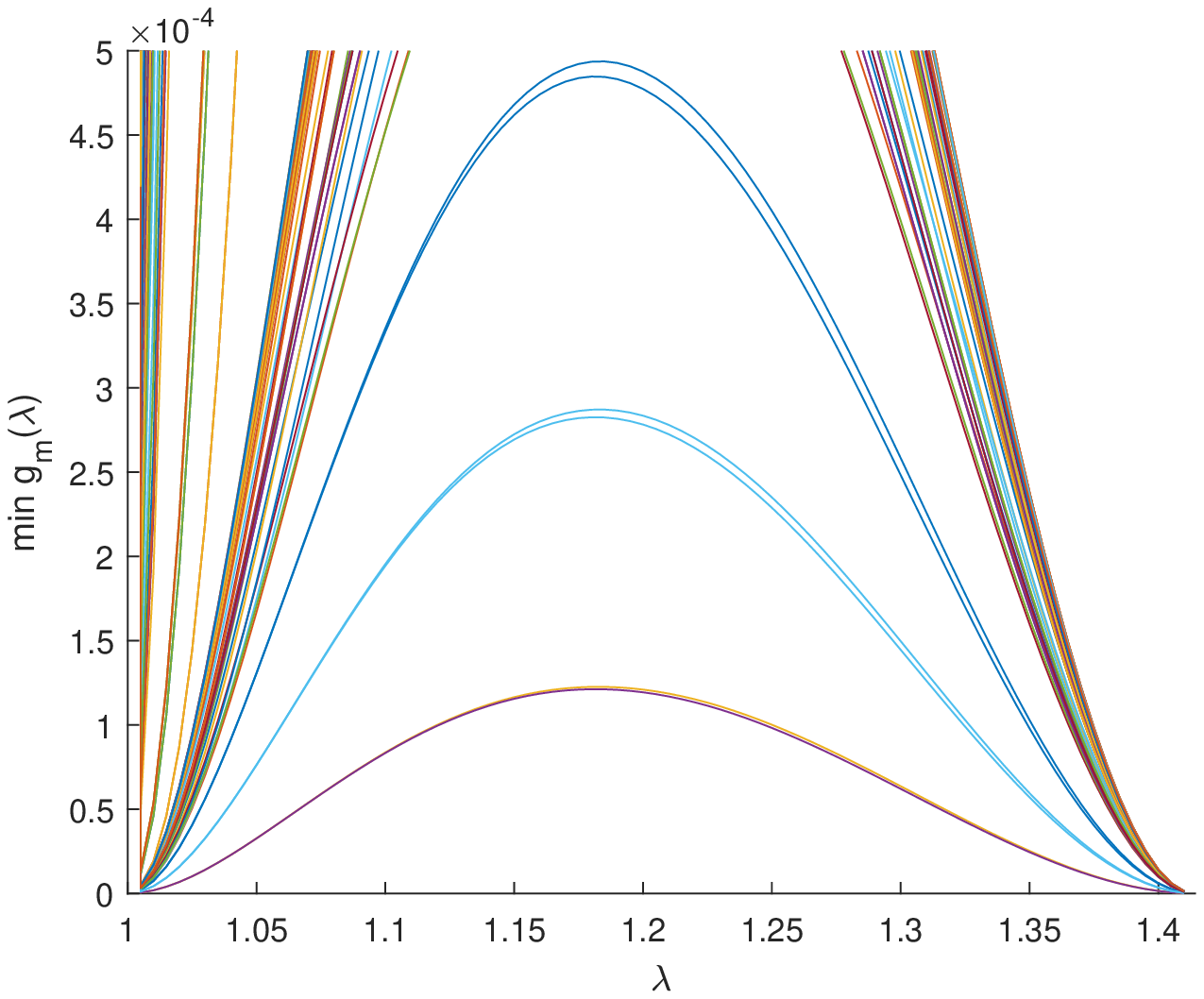}
  \caption{}
  \label{fig0d}
\end{subfigure}
\caption{Figure \ref{fig0a} and Figure \ref{fig0b} respectively represent $\min_{t\in\R} g_1$ and $\min_{t\in\R} g_2$ for the first option in both the cases \eqref{Ca} and \eqref{Cc} in Theorem \ref{TnThm}. Also, Figure \ref{fig0a} and Figure \ref{fig0b} respectively represent $\min_{t\in\R} g_2$ and $\min_{t\in\R} g_1$ for the second option in the cases \eqref{Ca} and \eqref{Cc}. In Figure \ref{fig0c} and Figure \ref{fig0d} we respectively plot $\min_{t\in\R} g_1$ and $\min_{t\in\R} g_2$ for the first option in both the cases \eqref{Cb} and \eqref{Cd} in Theorem \ref{TnThm}. Also, Figure \ref{fig0c} and Figure \ref{fig0d} respectively represent $\min_{t\in\R} g_2$ and $\min_{t\in\R} g_1$ for the second option in the cases \eqref{Cb} and \eqref{Cd}. 
%
%
Each line corresponds to a different value of $j\in\{1,\dots,6\},\, l\in \{1,\dots,48\}$.
\label{fig:fig}
}
\end{figure}

Now, given $\rho_0$ as in the claim, we know that there exists $r= \rho_0 + \max_i|\mt F_i|$ such that if $\mt G\in\R^{3\times 3}$ satisfies $|\mt G|\geq r$ then $|\mt F_i-\mt G|\geq \rho_0.$ Furthermore, the function $H\colon \{\mt G\in \R^{3\times 3}: |\mt G|<r\}\to \R^{3\times3}$ defined by $H(\mt G) = \mt G^T\mt G$ is Lipschitz on its domain, and hence there exists $c_0>0$ such that
$$
|\mt F_i - \mt G|\geq c_0 |H(\mt F_i) - H(\mt G)|.
$$ 
Therefore, combining this inequality with the claim we obtain that $\mt F_i(s_i)$ enjoys the separation property with $\rho=\rho_0\min\{1,c_0\}$.}

\paragraph{$V_{II}$ junctions and local stability.}
First, we have to construct $\omega$ such that \eqref{HH3}--\eqref{HH4} in Definition \ref{DefVII} are satisfied. But for $(i,\sigma_i)$ as in \eqref{Ca}--\eqref{Cd}, fixed $\vc n_1 = \vc n_1^+$ we can choose $\vc n_2 = \pm \vc n_i^{\sigma_i}$ such that \eqref{HH3}--\eqref{HH4} are satisfied. Let us now define $\vc y$ as in \eqref{stabbile}. This is well defined because of the Hadamard jump condition, and leads to a $V_{II}$ junction for each of the cases \eqref{Ca}--\eqref{Cd}. Given the steps above, in order to show that the $V_{II}$ junctions are stable, we just need to verify the assumption in Theorem \ref{StableThm} that $(\mt V_j^{2}\vcg\phi_j\times\vcg\psi)\cdot\vc m\neq 0$, with $j=1,2$, where in the notation of Theorem \ref{StableThm} $\mt V_1 = \mt U_1$ and $\mt V_2 = \mt U_i$ and $i$ is given by \eqref{Ca}--\eqref{Cd}. This is done by using \eqref{glims}. We plot {$(\mt V_j^{2}\vcg\phi_j\times\vcg\psi)\cdot\vc m,$} against $\lambda$ in Figure \ref{FigUltimaIp}, and we deduce that it is satisfied for all the cases \eqref{Ca}--\eqref{Cd} and $j=1,2$.
The $V_{II}$ junctions given by \eqref{Ca}--\eqref{Cd} are hence locally stable. 
\begin{figure}
  \centering
  \includegraphics[width=.5\linewidth]{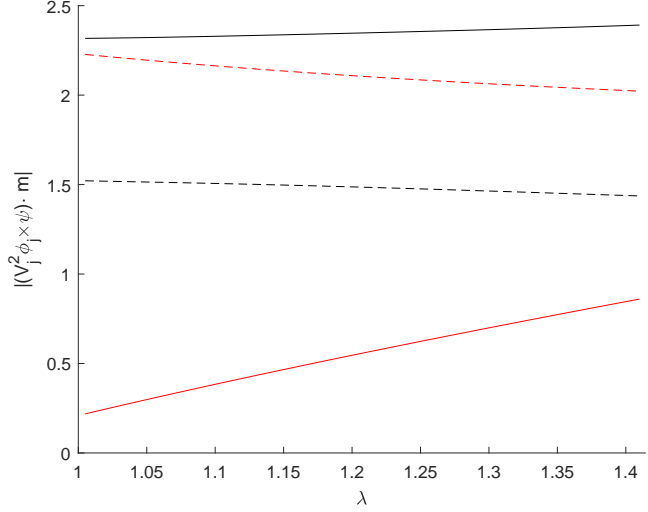}
\caption{\label{FigUltimaIp}Plotting {$|(\mt V_j^{2}\vcg\phi_j\times\vcg\psi)\cdot\vc m|,$} against $\lambda$. In black the cases \eqref{Ca} and \eqref{Cc}, while in red the cases \eqref{Cb} and \eqref{Cd}. Continuous and dashed lines are respectively for $j=1$ and $j=2$ for the first out of the two options in \eqref{Ca}--\eqref{Cd}, and for $j=2$ and $j=1$ for the second options in \eqref{Ca}--\eqref{Cd}.}
\end{figure}

\paragraph{$V_{II}$ junctions between $\mt 1 +\vc a_1^+\otimes\vc n_1^+$ and $\mt 1 +\vc a_1^{-}\otimes\vc n_1^{-}$.}
In this case there are many slip systems which make plastic junctions possible. However, the only ones which satisfy the necessary conditions of Lemma \ref{NecessityCond}, and such that $\vcg\psi_1,\vcg\psi_2\perp \hat{\vc m}$ (where $\hat{\vc m}$ is parallel to $\vc n_1\times\vc n_2$) as required by hypothesis \ref{HH4} in Definition \ref{DefVII}, are couples of slip systems among
\begin{multicols}{2}
\begin{enumerate}[(I)]
\item\label{CasiI} $\vcg\phi = (-1,1,1)$ and $\vcg \psi = (2,1,1)$;
\item $\vcg\phi = (1,1,1)$ and $\vcg \psi = (-2,1,1)$;
\item $\vcg\phi = (1,-1,1)$ and $\vcg \psi = (0,1,1)$;
\item\label{CasiIV} $\vcg\phi = (1,1,-1)$ and $\vcg \psi = (0,1,1)$.
\end{enumerate}
\end{multicols}
Below we denote by case $(j,k)$ the case where $\vcg\phi_1\otimes\vcg\psi_1,\vcg\phi_2\otimes\vcg\psi_2$ are respectively given by $j$ and $k$ among \eqref{CasiI}--\eqref{CasiIV} above. Let us study the situation in the different cases:\\

\textit{Case $(III,III)$ and case $(IV,IV)$.}
In these cases Proposition \ref{PropMeet} guarantees that there are no plastic junctions as $(\alpha_2\vc a_1 + \alpha_1 \vc a_2)\cdot (\hat{\vcg\phi}_1\times\hat{\vcg\phi}_2)= (\beta_2\vc a_1 + \beta_1 \vc a_2)\cdot (\hat{\vcg\phi}_1\times\hat{\vcg\phi}_2) = 0$, but $(\vc a_1\times\vc a_2)\cdot \hat{\vcg\phi}_i\neq 0,$ for $i=1,2$, in \eqref{cond s1 a}.\\

\textit{Cases $(I,III),(I,IV),(II,III),(II,IV),$ $(III,I),(III,II),(IV,I),(IV,II)$.} By Proposition \ref{PropMeet} there exists a unique plastic junction, and $\bar t_i=0$ for the slip on the plane $(0,1,1)$. Therefore, this cases can be studied within the context of cases $(I,I)$ and $(II,II)$ below.\\

\textit{Case $(I,II)$ and case $(II,I)$.}
In these cases, Proposition \ref{PropMeet} guarantees the existence of a one parameter family of plastic junctions. However, no local rigidity (in the sense of Definition \ref{DefPlJ}) holds. Indeed, let $\bar t_1,\bar t_2\in\R$, $\vc b\in\R^3$ and $\vc m\in\mathbb{S}^2$ be such that
$$
(\mt 1 + \vc a_1^+\otimes\vc n_1^+)(\mt1 + \bar t_1\vcg\phi_1\otimes\vcg\psi_1) - (\mt 1 + \vc a_1^-\otimes\vc n_1^-)(\mt1 + \bar t_2\vcg\phi_2\otimes\vcg\psi_2) = \vc b\otimes \vc m.
$$
Let $\mt R\in SO(3)$ be a rotation of angle $\theta$ and axis $\hat{\vc m} = \frac{\vc n_1^+\times\vc n_1^-}{|\vc n_1^+\times\vc n_1^-|}.$ We notice that $\hat{\vc m}\perp \vcg \phi_1,\vcg \phi_2,\vcg \psi_1,\vcg \psi_2, \vc a_1^+, \vc a_1^-$, and hence
\begin{align*}
\vc 0 = \bigl(\mt R(\mt 1 + \vc a_1^+\otimes\vc n_1^+)(\mt1 + t_1\vcg\phi_1\otimes\vcg\psi_1) - (\mt 1 + \vc a_1^-\otimes\vc n_1^-)(\mt1 + t_2\vcg\phi_2\otimes\vcg\psi_2) \bigr)\hat{\vc m} ,
\end{align*}
for any $t_1,t_2\in\R.$ Therefore, if for any small $\theta$ we can show that there exists $t_1^*,t_2^*\in\R$ such that
\beq
\label{testedv}
\vc 0 = \bigl(\mt R(\mt 1 + \vc a_1^+\otimes\vc n_1^+)(\mt1 + t_1^*\vcg\phi_1\otimes\vcg\psi_1) - (\mt 1 + \vc a_1^-\otimes\vc n_1^-)(\mt1 + t_2^*\vcg\phi_2\otimes\vcg\psi_2) \bigr){\vc v}, \qquad \vc v = \frac{\vc m\times\hat{\vc m}}{|\vc m\times\hat{\vc m}|},
\eeq
we have for any small $\theta$,
$$
\mt R(\mt 1 + \vc a_1^+\otimes\vc n_1^+)(\mt1 + t_1^*\vcg\phi_1\otimes\vcg\psi_1) - (\mt 1 + \vc a_1^-\otimes\vc n_1^-)(\mt1 + t_2^*\vcg\phi_2\otimes\vcg\psi_2) = \vc c\otimes\vc m,
$$
for some $\vc c\in\R^3$, and hence the plastic junction is not rigid. But \eqref{testedv} simplifies to 
\beq
\label{PerV}
\begin{split}
\mt R\vc a_1^+(\vc n_1^+\cdot\vc v) &- \vc a_1^-(\vc n_1^-\cdot\vc v) + t_1^* \mt R(\mt 1 + \vc a_1^+\otimes\vc n_1^+)\vcg\phi_1(\vcg\psi_1\cdot\hat{\vc m}) \\
&- t_2^* (\mt 1 + \vc a_1^-\otimes\vc n_1^-)\vcg\phi_2(\vcg\psi_2\cdot\hat{\vc m}) + (\cos(\theta)-1)\vc v + \sin(\theta)\vc m = \vc 0.
\end{split}
\eeq
If $\vcg\psi_1\cdot\hat{\vc m}=0$ or $\vcg\psi_2\cdot\hat{\vc m}=0$, that is if $\vcg\psi_1\parallel\vc m$ or if $\vcg\psi_2\parallel\vc m$, then by hypothesis \ref{HH4} in Theorem \ref{StableThm} the case reduces to case $(I,I)$ or case $(II,II)$ below. Otherwise, since $(\mt 1 + \vc a_1^+\otimes\vc n_1^+)\vcg\phi_1$ and $(\mt 1 + \vc a_1^-\otimes\vc n_1^-)\vcg\phi_2$ are linearly independent, there exists an open neighbourhood $\mathcal{U}$ of $0$ such that $\mt R(\mt 1 + \vc a_1^+\otimes\vc n_1^+)\vcg\phi_1$ and $(\mt 1 + \vc a_1^-\otimes\vc n_1^-)\vcg\phi_2$ are linearly independent for any $\theta\in\mathcal{U}$. Taking in account that all the terms in \eqref{PerV} are orthogonal to $\hat{\vc m}$, \eqref{PerV} is solvable for some $t_1^*,t_2^*\in\R$. As a consequence the junctions are not locally rigid.\\

\textit{Case $(I,I)$ and case $(II,II)$.}
In these cases Proposition \ref{PropMeet} guarantees the existence of a one parameter family of solutions respectively given by
$$
s_1 = s_2 + \frac{\lambda(\lambda^2-1)}{\sqrt{2}(2\lambda^4+1)},\qquad s_1 = s_2 - \frac{\lambda(\lambda^2-1)}{\sqrt{2}(2\lambda^4+1)}.
$$
In the cases (I,I) and (II,II), we respectively have 
\begin{equation}
\label{Mcaso1last}
\begin{split}
\vc m \parallel  \Bigl(2 - \frac{ 4(2\lambda^4 +1)}{4\lambda^4(2s_2+1) + \sqrt2\lambda^3 - \sqrt{2}\lambda + 4s_2 } , 1, 1\Bigr),\\
\vc m \parallel  \Bigl(\frac{ 4(2\lambda^4 +1)}{4\lambda^4(2s_2+1) - \sqrt2\lambda^3 +\sqrt{2}\lambda + 4s_2} -2, 1, 1\Bigr)
.
\end{split}
\end{equation}
By arguing as in the case $(I,II)$ and the case $(II,I)$ we can deduce that, as long as $(2,1,1)\nparallel\vc m$ and $(-2,1,1)\nparallel\vc m$ then the plastic junctions constructed in the case $(I,I)$ and in the case $(II,II)$ are not locally rigid. But we notice that, given $\lambda\in(1,\sqrt{2})$ and $\vc m$ as in \eqref{Mcaso1last} this never occurs, concluding that no local rigidity holds for these junctions.\\

\textit{Case $(III,IV)$ and case $(IV,III)$.}
In these cases there exists plastic junctions if and only if $s_2= - s_1 = \frac{\lambda(\lambda^2-1)}{2\sqrt{2}}$, and $\vc m =(1,0,0)$. Let now $\mt R\in SO(3)$ be a rotation of angle $\theta\in(-\pi,\pi]$ and axis $\hat{\vc m}=\frac{\vc n_1^+\times\vc n_1^-}{|\vc n_1^+\times\vc n_1^-|}.$ In this case we can solve explicitly 
$$
\cof\bigl(\mt R \mt R_1\mt V_1(\mt 1+t_1\vcg \phi_{1}\otimes \vcg\psi_1) - (\mt R_1\mt V_1 + \vc b_1\otimes \vc m_1+ \vc b_2\otimes \vc m_2)(\mt 1+t_2\vcg \phi_{2}\otimes \vcg\psi_2)  \bigr)= \mt 0,
$$
in terms of $(t_1,t_2)$, and deduce that the unique solution is given by
$$
\bar t_2 = -\bar t_1 = \frac{\lambda^2\Bigl((\lambda^2-1)\cos\Bigl(\frac\theta2\Bigr) - 2\lambda\sin\Bigl(\frac\theta2\Bigr) \bigr)}{\sqrt2\Bigl((\lambda^2-1)\sin\Bigl(\frac\theta2\Bigr) + 2\lambda\cos\Bigl(\frac\theta2\Bigr) \bigr)}.
$$ 
In this case, however,
$$
\mt R \mt R_1\mt V_1(\mt 1+t_1\vcg \phi_{1}\otimes \vcg\psi_1) - (\mt R_1\mt V_1 + \vc b_1\otimes \vc m_1+ \vc b_2\otimes \vc m_2)(\mt 1+t_2\vcg \phi_{2}\otimes \vcg\psi_2)  = \vc b\otimes\vc m,
$$
for some $\vc b\in\R^3$ depending on $\theta$. Therefore, also in this case no local rigidity holds.\\

\noindent
The verification of the Theorem is thus completed.

\subsection{Comparison with experimental results}
\label{Comparison with Experiments}
We now compare the results obtained in Theorem \ref{TnThm} to the experimental observations in \cite{Inamura} for \Tn. We recall that for \Tn, $V_{II}$ junctions with $\mt 1 + \vc a_1^+\otimes \vc n_1^+$ are observed only for $\mt 1 + \vc a_i^{\sigma_i}\otimes \vc n_i^{\sigma_i}$, with $(i,\sigma_i)$ equal to $(4,-)$ and $(6,-)$. This is coherent with the result in Theorem \ref{TnThm}. Indeed, although Theorem \ref{TnThm} predicts the existence of $V_{II}$ junctions also for the cases $(i,\sigma_i)$ equal to $(3,+)$ and $(5,+)$, Figure \ref{FigSs} shows that the energy required for a single slip in these cases is consistently bigger than the energy required in the case $(i,\sigma_i)$ equal to $(4,-)$ and $(6,-)$. 

If we approximate the transformation matrices for the phase transition in \Tn\,with the matrices in \eqref{cubictoortho} with $d=\frac1\lambda$, $\lambda\in(1.033,1.035)$ we get that, in some regions of the domain, the shear amount required to form $V_{II}$ junctions in the cases $(i,\sigma_i)$ equal to $(3,+)$ and $(5,+)$, is about ten times bigger than in the case $(i,\sigma_i)$ equal to $(4,-)$ and $(6,-)$. Therefore, one can explain the lack of $V_{II}$ junctions between $\mt 1 + \vc a_1^+\otimes \vc n_1^+$ and $\mt 1 + \vc a_i^{\sigma_i}\otimes \vc n_i^{\sigma_i}$, with $(i,\sigma_i)$ equal to $(3,+)$ and $(5,+)$ with the fact that they are energetically expensive. We report the above discussed results in Table \ref{Table 01}. 

Another factor influencing the presence of $V_{II}$ junctions may be the norm of the dislocation density tensor $\nabla\times \mt F^p$ (see e.g., \cite{ReinaConti}). For $V_{II}$ junctions as in Definition \ref{DefStabY} we have that $\nabla\times \mt F^p$ is a Radon measure and $\nabla\times \mt F^p = \bigl( \bar{t}_1\vcg\phi_1\otimes \vcg \psi_1-\bar{t}_2\vcg\phi_2\otimes \vcg \psi_2\bigr) \times \vc m \,\mathscr H^2\,\mres\{\vc x\cdot\vc m=0\}$. Here $\mathscr H^2\,\mres\{\vc x\cdot\vc m=0\}$ is the two-dimensional Hausdorff measure restricted to the plane $\{\vc x\cdot\vc m=0\}$, and the cross product is taken row-wise. We report in Figure \ref{CurlPict} the values of $|\bigl( \bar{t}_1\vcg\phi_1\otimes \vcg \psi_1-\bar{t}_2\vcg\phi_2\otimes \vcg \psi_2\bigr) \times \vc m|$ for the the constructed $V_{II}$ junctions. Again, the results confirm that the cases $(i,\sigma_i)$ equal to $(4,-)$ and $(6,-)$ are more preferable than the cases $(i,\sigma_i)$ equal to $(3,+)$ and $(5,+)$.
\begin{figure}
\begin{subfigure}{.45\textwidth}
  \centering
  \includegraphics[width=.9\linewidth]{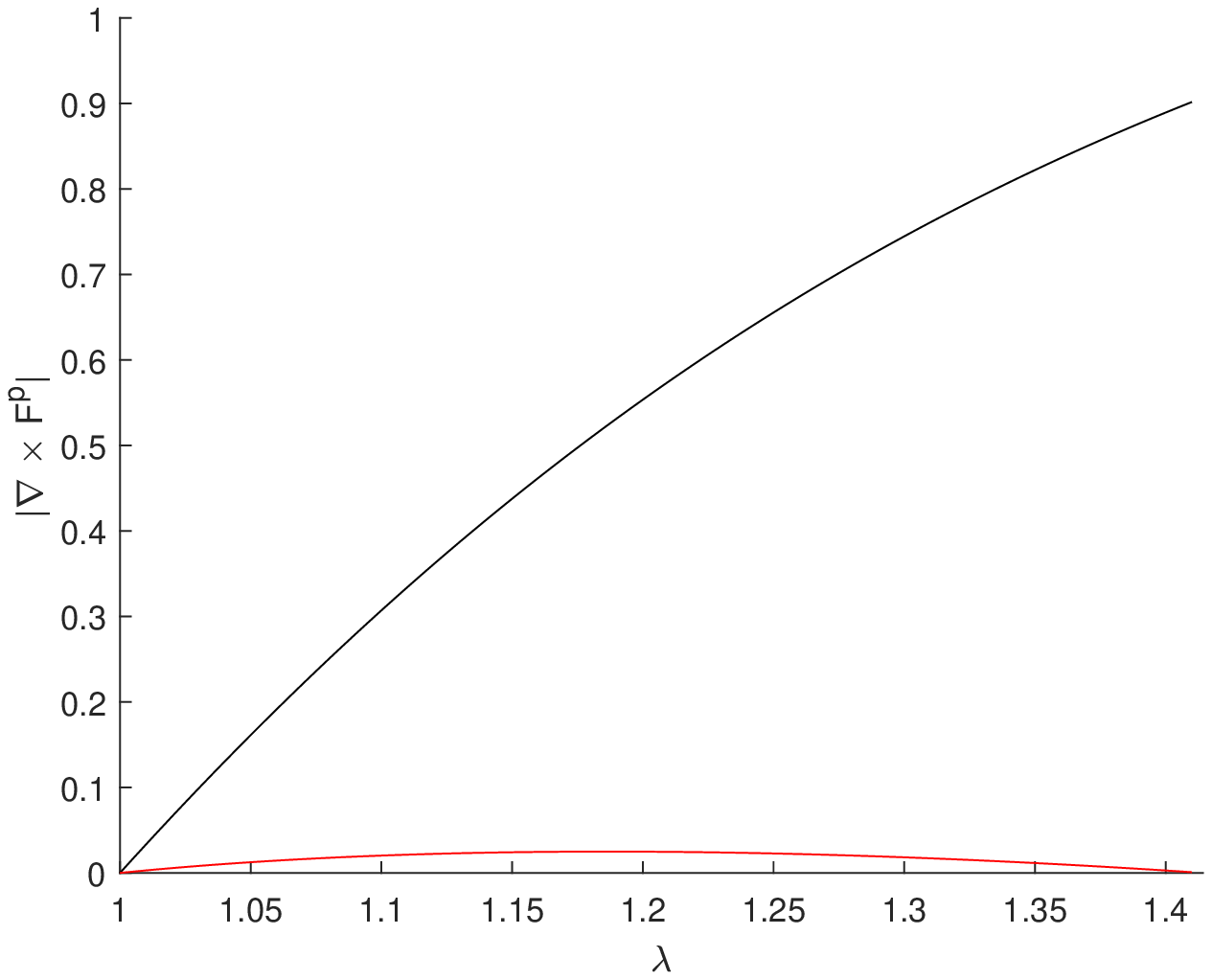}
  \caption{}
  \label{fig99}
\end{subfigure}%
\begin{subfigure}{.45\textwidth}
  \centering
  \includegraphics[width=.9\linewidth]{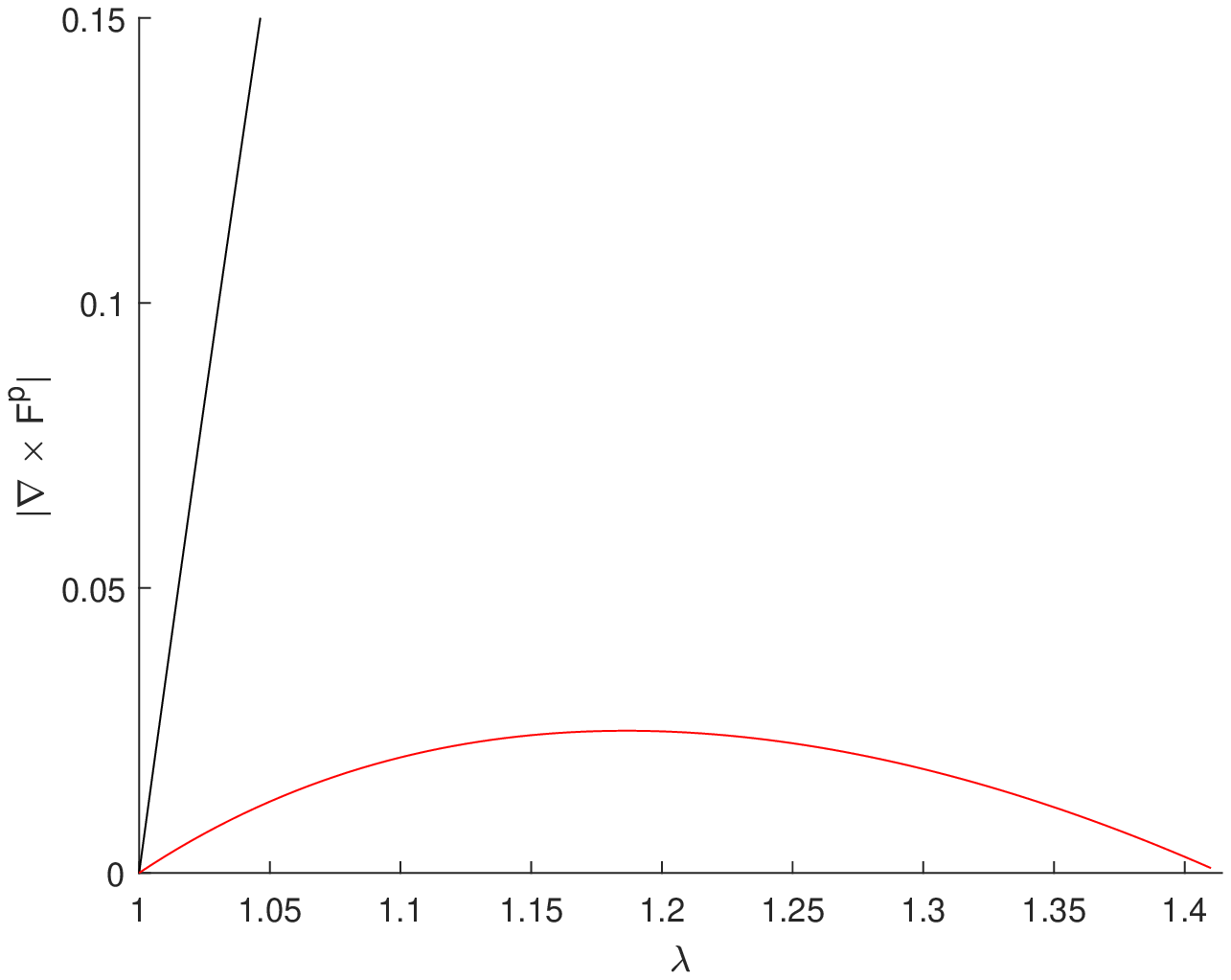}
  \caption{}
  \label{fig1}
\end{subfigure}%
\caption{Plotting $|\nabla\times \mt F^p|$ against $\lambda$. In black the cases $(i,\sigma_i)$ equal to $(3,+)$ and $(5,+)$, while in red the cases $(i,\sigma_i)$ equal to $(4,-)$ and $(6,-)$. On the right a zoom of the plot.}
\label{CurlPict}
\end{figure}
\begin {table}[h]
\begin{center}
\begin{tabular}{ |l|l|l|l| }
\hline
$(i,\sigma_i)$ & $|\theta|$ (approx. in dgs.) & Observed junction  & $(|\bar t_1|,|\bar t_2|)$ (values$\cdot 10^2$)  \\ \hline
$(1,-)$ & $3.84$ & none & none \\ 
$(2,+)$ & $3.28$ & none & none \\ 
$(2,-)$ & $3.28$ & none & none \\ 
$(3,+)$ & $0.69$ & $V_I$ & $(0.44,4.25)-(0.47,4.5)$ \\
$(3,-)$ & $3.70$ & none & none \\ 
$(4,+)$ & $3.70$ & none & none \\
$(4,-)$ & $0.57$ & $V_{II}$ & $(0.23,0.37)-(0.24,0.39)$ \\ 
$(5,+)$ & $0.69$ & $V_I$ & $(0.44,4.25)-(0.47,4.5)$ \\ 
$(5,-)$ & $3.70$ & none & none \\
$(6,+)$ & $3.70$ & none & none \\
$(6,-)$ & $0.57$ & $V_{II}$ & $(0.23,0.37)-(0.24,0.39)$ \\
 \hline
\end{tabular}
\end{center}
\caption {\label{Table 01} 
Incompatible junctions observed in \Tn: comparison between experimental data and results obtained in Theorem \ref{TnThm}. In the second column we give the incompatibility between $\mt 1+\vc a_1^+\otimes\vc n_1^+$ and $\mt 1+\vc a_i^{\sigma_1}\otimes\vc n_i^{\sigma_1}$ measured as in \cite{BallS1} (see Introduction). The approximate values obtained for the angles of incompatibility $\theta$ are expressed in degrees. In the third column we report the type of incompatible junction observed in experiments. In the last column we report the values of $|\bar t_1|,|\bar t_2|$, the amount of simple shear for the $V_{II}$ junctions given by Theorem \ref{TnThm}. For this values we have given a range, corresponding to the value of $\lambda = 1.033$ and $\lambda = 1.035$ respectively. This range approximates the deformation gradient for \Tn\,best. The obtained results confirm that $V_{II}$ junctions are energetically convenient when $(i,\sigma_i)$ is equal to $(4,-)$ or $(6,-).$
The data in the second and third column are taken from \cite[Table 4]{Inamura}. } 	
\end{table}

\section{Concluding remarks}
\label{ConcludingRmk}
In Section \ref{Minim} we provided a mathematical characterisation of $V_{II}$ junctions in martensitic transformations. Our $V_{II}$ junctions are weak local minimisers of a physically relevant energy introduced in Section \ref{Nonlin}. In Section \ref{Appl} we have showed that our model is successful in capturing the $V_{II}$ junctions observed in \Tn.
There are nonetheless a few directions in which the present work can be extended or improved. 

Despite $V_{II}$ junctions look very similar to the inexact junctions observed in Ni\textsubscript{65}Al\textsubscript{35} \cite{BallS1,BallS2}, the theory developed in this paper cannot be applied to that case. This is mainly for three reasons: first, as reported in \cite{BS3} elastic distortions are experimentally observed and seem to play an important role for the formation of incompatible junctions in Ni\textsubscript{65}Al\textsubscript{35}. Second, when considering average deformation gradients like laminates (and hence a relaxed elastic energy), one should also consider average plastic shears (and thus a relaxed plastic energy). In that case, also the compatibility results of Section \ref{ShearComp} should be re-proven. Third, it seems that a rigidity argument based on the separation of wells as the one in the proof of Theorem \ref{StableThm} does not work for a relaxed elastic energy.  

The aim of this work is to study $V_{II}$ junctions; it would be interesting to understand also $V_I$ junctions within this framework. This would allow to better understand nucleation of martensite in \Tn. Indeed, as reported in \cite{Inamura}, nucleation in \Tn\,occurs mostly through the formation of new $V_{I}$ junctions. However we were not able to find a mathematical characterisation of $V_I$ junctions which is both simple and well-defined, as in this case one should consider plastic deformations both in austenite and in the martensite plates. This will hopefully be discussed in future work.

In our opinion, taking in account small elastic effects would improve the physical accuracy of the model discussed in Section \ref{Nonlin}, but would make the proof of local stability much harder. The context of linear elasto-plasticity and the geometrically linear theory of elasticity for martensitic transformations (see e.g., \cite{Batt}) may provide a better framework to approach this problem analytically. Indeed, in geometrically linear elasticity the composition of subsequent deformations reduces to summing the respective deformation gradients, rather than multiplying them as in the context of nonlinear elasticity. Therefore, by giving up some accuracy in the model, this theory guarantees a more approachable framework for analytic results. 
Examples of recent studies of martensitic transformation within this context are \cite{AS,FO,AR}. However, we remark that in some particular cases the nonlinear elasticity theory and the geometrically linear theory may give different results (cf. the case of triple stars in \cite[Sec. 2-3]{CDPRZZ}).

\subsection*{Acknowledgments}
This work was partially supported by the Engineering and Physical Sciences Research Council [EP/L015811/1]. The author would like to thank John Ball, Tomonari Inamura and Angkana R\"uland for the useful discussions. The author would
like to acknowledge the two anonymous reviewers for improving this paper with their comments.

\bibliographystyle{plain}
\bibliography{biblio}

\end{document}